\pdfoutput=1

\documentclass[a4paper]{article}
\usepackage{xcolor}
\usepackage[final]{microtype}
\usepackage{mathtools}
\usepackage{amsmath}
\usepackage{amssymb}
\usepackage{amsthm}
\usepackage{natbib}
\usepackage{local-tikz}

\allowdisplaybreaks
\usepackage{hyperref}
\hypersetup{hidelinks}
\usepackage{stmaryrd}
\usepackage[capitalise]{cleveref}
\usepackage{multicol}

\usepackage{jms-sections}
\usepackage{local-abbrv}
\usepackage{local-jmsdelim}
\usepackage{mathpartir}
\usepackage{jms-operators}
\usepackage{newtxtext,newtxmath}
\usepackage{inconsolata}
\usepackage{iblock}

\theoremstyle{plain}
\newtheorem{therm}{Theorem}

\newtheorem{lemma}[therm]{Lemma}
\newtheorem{observation}[therm]{Observation}
\newtheorem{corollary}[therm]{Corollary}

\newtheorem{proposition}[therm]{Proposition}

\newtheorem{principle}{Principle}

\theoremstyle{definition}

\newtheorem{example}[therm]{Example}
\newtheorem{definition}[therm]{Definition}

\newtheorem{remark}[therm]{Remark}

\newtheorem{construction}[therm]{Construction}
\newtheorem{computation}[therm]{Computation}
\newtheorem{notation}[therm]{Notation}

\Crefname{paragraph}{Section}{Sections}
\Crefname{observation}{Observation}{Observations}
\Crefname{construction}{Construction}{Constructions}

\author{Jonathan Sterling}
\title{Reflexive graph lenses in univalent foundations}

\begin{document}

\maketitle
\begin{abstract}
Martin-L\"of's identity types provide a generic (albeit opaque) notion of identification or ``equality'' between any two elements of the same type, embodied in a canonical reflexive graph structure $(=_A, \mathbf{refl})$ on any type $A$. The miracle of Voevodsky's \emph{univalence principle} is that it ensures, for essentially any naturally occurring structure in mathematics, that this the resultant notion of identification is equivalent to the type of \emph{isomorphisms} in the category of such structures. Characterisations of this kind are not automatic and must be established one-by-one; to this end, several authors have employed \emph{reflexive graphs} and \emph{displayed reflexive graphs} to organise the characterisation of identity types.
We contribute \textbf{reflexive graph lenses}, a new family of intermediate abstractions lying between families of reflexive graphs and displayed reflexive graphs that simplifies the characterisation of identity types for complex structures. Every reflexive graph lens gives rise to a (more complicated) displayed reflexive graph, and our experience suggests that many naturally occurring displayed reflexive graphs arise in this way. Evidence for the utility of reflexive graph lenses is given by means of several case studies, including the theory of reflexive graphs itself as well as that of polynomial type operators. Finally, we exhibit an equivalence between the type of reflexive graph fibrations and the type of \emph{univalent} reflexive graph lenses.
 \end{abstract}

\tableofcontents

\ExplSyntaxOn

\def\covsym{+}
\def\ctrvsym{-}
\def\bivsym{\pm}

\def\lint{\sum^{\covsym}}
\def\rint{\sum^{\ctrvsym}}
\def\uint{\sum^{\bivsym}}

\def\lprod{\prod^{\covsym}}
\def\rprod{\prod^{\ctrvsym}}

\NewDocumentCommand\PolySymbol{}{\jms_scale_bigop:n {\operatorname{P}}}

\NewDocumentCommand\AugSpx{}{\Delta\Sub{a}}
\NewDocumentCommand\Poly{ommm}{
  \mathbf{P}\IfValueT{#1}{\Sup{#1}}\Sub{#2}\prn{#4,#3}
}
\NewDocumentCommand\CovPoly{mmm}{
  \mathcal{P}^{+}\Sub{#1}\prn{#3,#2}
}
\NewDocumentCommand\CtrvPoly{mmm}{
  \mathcal{P}^{-}\Sub{#1}\prn{#3,#2}
}

\jms_define_big_quantifier:NN \LSum \lint
\jms_define_big_quantifier:NN \RSum \rint
\jms_define_big_quantifier:NN \USum \uint

\NewDocumentCommand\LProd{m}{
  \DelimMin{1}
  \jms_bigop:N {\lprod} {#1\bullet}
}

\NewDocumentCommand\RProd{m}{
  \DelimMin{1}
  \jms_bigop:N {\rprod} {#1\bullet}
}

\ExplSyntaxOff

\NewDocumentCommand\Inv{}{\Sup{-1}}
\NewDocumentCommand\Bad{}{\Sup{\color{red}\dagger}}

\NewDocumentCommand\DiscPO{}{{\triangle}}
\NewDocumentCommand\CodiscPO{}{{\triangledown}}

\NewDocumentCommand\IdCon{m}{=_{#1}}
\NewDocumentCommand\Id{mmm}{#2 \IdCon{#1} #3}
\NewDocumentCommand\TYPE{}{\mathbf{type}}
\NewDocumentCommand\Refl{}{\mathbf{refl}}
\NewDocumentCommand\Bind{m}{#1\mathpunct{.}}
\NewDocumentCommand\J{mmmmmm}{\mathbf{J}_{#1}^{#2}\prn{#3;#4,#5,#6}}
\NewDocumentCommand\DefEmph{m}{\textit{\textbf{#1}}}
\NewDocumentCommand\Lam{m}{\lambda{#1}\mathpunct{.}}
\NewDocumentCommand\Singleton{mm}{\brc{#2}_{#1}}
\NewDocumentCommand\Nat{}{\mathbb{N}}
\NewDocumentCommand\RxCon{}{\mathsf{rx}}

\NewDocumentCommand\Rx{m}{\RxCon\Sub{#1}}
\NewDocumentCommand\DRx{mm}{\Rx{#1}\Sup{#2}}

\NewDocumentCommand\Act{m}{\mathsf{ap}_{#1}}
\NewDocumentCommand\ActIdn{}{\mathsf{actIdn}}

\NewDocumentCommand{\ct}{o}{%
  \mathbin{%
    \mathchoice%
    {\raisebox{0.5ex}{$\displaystyle\centerdot$}}%
    {\raisebox{0.5ex}{$\centerdot$}}%
    {\raisebox{0.25ex}{$\scriptstyle\,\centerdot\,$}}%
    {\raisebox{0.1ex}{$\scriptscriptstyle\,\centerdot\,$}}%
  }%
}

\NewDocumentCommand\ctl{}{\ct_{\mathsf{l}}}
\NewDocumentCommand\ctr{}{\ct_{\mathsf{r}}}

\NewDocumentCommand\FmtGph{m}{\mathcal{#1}}
\NewDocumentCommand\gA{}{\FmtGph{A}}
\NewDocumentCommand\gB{}{\FmtGph{B}}
\NewDocumentCommand\gU{}{\FmtGph{U}}
\NewDocumentCommand\gG{}{\FmtGph{G}}
\NewDocumentCommand\gC{}{\FmtGph{C}}
\NewDocumentCommand\gE{}{\FmtGph{E}}
\NewDocumentCommand\gF{}{\FmtGph{F}}
\NewDocumentCommand\gO{}{\FmtGph{O}}
\NewDocumentCommand\Edge{mo}{\mathrel{\approx\Sub{#1}\IfValueT{#2}{\Sup{#2}}}}
\NewDocumentCommand\LeftDisplay{}{\mathsf{d}_{\covsym}}
\NewDocumentCommand\RightDisplay{}{\mathsf{d}_{\ctrvsym}}
\NewDocumentCommand\LooseDisplay{}{\mathsf{d}_{\bivsym}}
\NewDocumentCommand\Op{}{\Sup{\mathsf{op}}}
\NewDocumentCommand\TotOp{}{\Sup{\widetilde{\mathsf{op}}}}

\NewDocumentCommand\GphOn{g}{\mathsf{GphOn}_U\IfValueT{#1}{\prn{#1}}}
\NewDocumentCommand\Gph{}{\mathsf{Gph}_U}
\NewDocumentCommand\RxOn{g}{\mathsf{RxOn}_U\IfValueT{#1}{\prn{#1}}}
\NewDocumentCommand\RxGph{}{\mathsf{RxGph}_U}

\NewDocumentCommand\RxGphOn{}{\mathsf{RxGphOn}_U}
\NewDocumentCommand\RoundTrip{}{\mathsf{RoundTrip}}
\NewDocumentCommand\RoundTripOn{}{\mathsf{RoundTripOn}}
\NewDocumentCommand\IsSRPair{}{\mathsf{isSecRetPair}}
\NewDocumentCommand\SRPair{}{\mathsf{SecRetPair}}

\NewDocumentCommand\EdgeSpace{m}{\widetilde{#1}}
\NewDocumentCommand\LeftLiftRefl{mm}{\mathscr{l}_{#1}^{#2}}
\NewDocumentCommand\RightLiftRefl{mm}{\mathscr{r}_{#1}^{#2}}
\NewDocumentCommand\DepLeftLiftRefl{mm}{\underline{\mathscr{l}}_{#1}^{#2}}
\NewDocumentCommand\DepRightLiftRefl{mm}{\underline{\mathscr{r}}_{#1}^{#2}}
\NewDocumentCommand\DepBiLiftRefl{mm}{\underline{\mathscr{u}}_{#1}^{#2}}

\NewDocumentCommand\Diag{m}{#1\Sub{\mathsf{rx}}}
\NewDocumentCommand\IdToEdge{mom}{\floors{#3}\IfValueT{#2}{\Sup{#2}}\Sub{#1}}
\NewDocumentCommand\EdgeToId{mom}{\bbrk{#3}\IfValueT{#2}{\Sup{#2}}\Sub{#1}}

\NewDocumentCommand\EdgeToIdGen{momm}{\bbrk{#3; #4}\IfValueT{#2}{\Sup{#2}}\Sub{#1}}

\NewDocumentCommand\DGphOn{mg}{\mathsf{DGphOn}_U\Sup{#1}\IfValueT{#2}{\prn{#2}}}
\NewDocumentCommand\DGph{m}{\mathsf{DGph}_U\Sup{#1}}
\NewDocumentCommand\DRxOn{mg}{\mathsf{DRxOn}_U\Sup{#1}\IfValueT{#2}{\prn{#2}}}
\NewDocumentCommand\DRxGphOver{g}{\mathsf{DRxGphOver}_U\IfValueT{#1}{\prn{#1}}}
\NewDocumentCommand\DRxGph{}{\mathsf{DRxGph}_U}

\NewDocumentCommand\CovFlatten{mm}{#1\Sub{{\downarrow}#2}\Sup{+}}
\NewDocumentCommand\CtrvFlatten{mm}{#1\Sub{{\downarrow}#2}\Sup{-}}

\NewDocumentCommand\Fan{mm}{\brc{#2}\Sub{#1}\Sup{+}}
\NewDocumentCommand\CoFan{mm}{\brc{#2}\Sub{#1}\Sup{-}}
\NewDocumentCommand\CovDisp{m}{\mathsf{disp}^+\Sub{#1}}
\NewDocumentCommand\CtrvDisp{m}{\mathsf{disp}^-\Sub{#1}}
\NewDocumentCommand\UnbDisp{m}{\mathsf{disp}^\pm\Sub{#1}}

\NewDocumentCommand\Push{mm}{\mathsf{push}\Sub{#1}\Sup{#2}}
\NewDocumentCommand\Pull{mm}{\mathsf{pull}\Sub{#1}\Sup{#2}}
\NewDocumentCommand\PushRx{mo}{\mathsf{pushRx}\Sub{#1}\IfValueT{#2}{\Sup{#2}}}
\NewDocumentCommand\PullRx{mo}{\mathsf{pullRx}\Sub{#1}\IfValueT{#2}{\Sup{#2}}}
\NewDocumentCommand\LJ{mm}{\mathsf{lext}\Sub{#1}\Sup{#2}}
\NewDocumentCommand\RJ{mm}{\mathsf{rext}\Sub{#1}\Sup{#2}}
\NewDocumentCommand\MidJRx{mo}{\mathsf{extRx}\Sub{#1}\IfValueT{#2}{\Sup{#2}}}
\NewDocumentCommand\RJRx{mo}{\mathsf{rextRx}\Sub{#1}\IfValueT{#2}{\Sup{#2}}}
\NewDocumentCommand\CovUpgrade{m}{\brk{#1}_+^\pm}
\NewDocumentCommand\CtrvUpgrade{m}{\brk{#1}_-^\pm}

\NewDocumentCommand\CovRepl{mm}{\mathsf{drep}^+\Sub{#1}#2}
\NewDocumentCommand\CtrvRepl{mm}{\mathsf{drep}^-\Sub{#1}#2}
\NewDocumentCommand\Vtx{}{\mathsf{Vtx}_U}

\NewDocumentCommand\CovLensStr{mg}{\mathsf{LensStr}^+\Sub{#1}\IfValueT{#2}{\prn{#2}}}
\NewDocumentCommand\CovLensOver{m}{\mathsf{Lens}^+\Sub{#1}}

\paragraph*{Acknowledgments}

I thank Rafa\"el Bocquet for suggesting the terminology of \emph{path objects}. I am also grateful to Fredrik Bakke, Marcelo Fiore, Daniel Gratzer, Tom de Jong, Egbert Rijke and David W\"arn for helpful conversations, and to Ulrik Buchholtz and Johannes Schipp von Branitz for many helpful suggestions and clarifications. I owe a great debt to Ian Ray for finding many mistakes in an earlier draft of this paper. This work was funded by the United States Air Force Office of Scientific Research under grant FA9550-23-1-0728 (\href{http://www.jonmsterling.com/jms-008K.xml}{\emph{New Spaces for Denotational Semantics}}; Dr.\ Tristan Nguyen, Program Manager). Views and opinions expressed are however those of the author only and do not necessarily reflect those of AFOSR.

\paragraph*{Conventions}
We work in Martin-L\"of's intensional dependent type theory with products $\Prod{x:A}{B\prn{x}}$, sums $\Sum{x:A}B\prn{x}$, identity types $\Id{A}{x}{y}$, a unit type $\mathbf{1}$, and type $\mathbb{N}$ of natural numbers; products, sums, and the unit type are assumed to have definitional $\eta$-laws. We do not globally assume any universes, nor function extensionality, nor univalence; these constructs are assumed locally where needed. We use the symbol $\equiv$ to denote definitional equality, and $=_A$ to denote identity types.
We follow the vernacular of univalent foundations, in which \emph{propositions} and \emph{sets} are defined extrinsically in terms of identity types. When we refer to existence, we mean it in the usual propositionally truncated sense of standard mathematical vernacular rather than in the sense of the HoTT Book~\citep{hottbook}.

\begin{xsect}{Introduction}

  \begin{xsect}{Martin-L\"of's identity types}
    In the setting of Martin-L\"of's intensional type theory~\citep{nordstrom-peterson-smith:1990}, every type $A$ comes equipped with a type constructor $x:A,y:A\vdash \Id{A}{x}{y}:\TYPE$ that classifies \emph{identifications} between elements of $A$, \ie witnesses that two given elements are ``equal''. Rather than being defined separately for each $A$, the identity type constructor $\IdCon{A}$ is specified generically for all types by means of formation, introduction, elimination, and computation rules. The formation and introduction rule for the identity type
    \begin{mathpar}
      \inferrule{
        \Gamma\vdash A:\TYPE\\
        \Gamma\vdash u,v:A
      }{
        \Gamma\vdash\Id{A}{u}{v}:\TYPE
      }
      \and
      \inferrule{
        \Gamma\vdash A:\TYPE\\
        \Gamma\vdash u:A
      }{
        \Gamma\vdash \Refl : \Id{A}{u}{u}
      }
    \end{mathpar}
    can be thought of as equipping $A$ with the structure of a \emph{reflexive graph}\footnote{In the context of category theory, one uses the term ``graph'' to refer to what combinatorists would normally refer to as \emph{multigraphs} or \emph{pseudo-graphs}. A graph in the traditional sense would be a graph (in the present generalised sense) in which the type of vertices is a set and the type of edges between two fixed vertices is a proposition; we might refer to such objects as ``locally propositional set-graphs''.} (see \cref{def:reflexive-graph}). The elimination and computation rules
    \begin{mathpar}
      \inferrule{
        \Gamma\vdash A:\TYPE\\
        \Gamma,x:A,y:A,p:\Id{A}{x}{y}\vdash C\prn{x,y,p}:\TYPE\\\\
        \Gamma,x:A\vdash c\prn{x} : C\prn{x,x,\Refl}\\
        \Gamma\vdash u,v:A\\
        \Gamma\vdash w:\Id{A}{u}{v}
      }{
        \Gamma\vdash \J{A}{C}{c}{u}{v}{w} : C\prn{u,v,w}
      }
      \and
      \inferrule{
        \Gamma\vdash A:\TYPE\\
        \Gamma,x:A,y:A,p:\Id{A}{x}{y}\vdash C\prn{x,y,p}:\TYPE\\
        \Gamma,x:A\vdash c\prn{x} : C\prn{x,x,\Refl}\\
        \Gamma\vdash u:A
      }{
        \Gamma\vdash \J{A}{C}{c}{u}{u}{\Refl} \equiv c\prn{u} : C\prn{u,u,\Refl}
      }
    \end{mathpar}
    then ensure that this reflexive graph structure is the smallest possible one that can be formed with vertices in $A$. This universal property is most commonly employed in terms of the principle of \emph{identification induction} below.

    \begin{principle}[{Identification induction}]\label[principle]{principle:identification-induction}
      Given a family of types $x,y:A;p:\Id{A}{x}{y}\vdash C\prn{x,y,p}: \TYPE$, to define a dependent function $x,y:A;p:\Id{A}{x}{y}\vdash f\prn{x,y,p}:C\prn{x,y,p}$ it suffices to specify $x:A\vdash f\prn{x,x,\Refl}:C\prn{x,x,\Refl}$.
    \end{principle}

    There is an alternative, equivalent, formulation of identification induction due to \citet{paulin-mohring:1993}, that fixes an element $u:A$ and allows induction on data of the form $y:A, p:u=y$. This is called ``based identification induction'':

    \begin{principle}[{Based identification induction}]\label[principle]{principle:based-identification-induction}
      Given a fixed element $u:A$ and a family of types $y:A,p:\Id{A}{u}{y}\vdash B\prn{y,p}:\TYPE$, to define a dependent function $y:A,p:\Id{A}{u}{y}\vdash f\prn{y,p}:B\prn{y,p}$ it suffices to specify $f\prn{u,\Refl} : B\prn{u,\Refl}$.
    \end{principle}

    Independently, \citet{van-den-berg-garner:2011} and \citet{lumsdaine:2010} have shown that identification induction exhibits a weak globular $\infty$-groupoid structure on each type $A$. In the lowest dimension, this corresponds to the construction of compositors for identifications and associators for the resulting compositions, \etc.
  \end{xsect}

  \begin{xsect}{Characterising identity types using (displayed) reflexive graphs}

    The role of the identity type is similar to that of equality in set theory: it expresses a single \emph{global} notion of identification whose rules apply to all types simultaneously.\footnote{Naturally, intensional type theory can speak about more than sets --- but the spirit of having a single global notion of identity is the same in both type theory and set theory.} On the other hand, because the rules of identity types have to apply uniformly to all types, these rules do not reflect any of the specific characteristics of individual types that would, if taken seriously, make it easier to both exhibit and use identifications.

    The goal of reflexive graphs is to provide a methodology \emph{within} Martin-L\"of's type theory by which the special properties of identifications in different types can be made explicit. An alternative design would be to use a type theory in which the identity types are not defined uniformly in each type, and instead each type comes with its own identity type, as in \emph{Observational Type Theory}~\citep{altenkirch-mcbride:2006,altenkirch-mcbride-swierstra:2007} and \emph{Higher Observational Type Theory}~\citep{altenkirch-chamoun-kaposi-shulman:2024}. Each design has trade-offs, but in this paper we are restricting attention to what can be done without changing the type theory.

    \begin{remark}
      Everything in this section is derived from prior works, such as the HoTT Book~\citep{hottbook}, Rijke's textbook on homotopy type theory~\citep{rijke:2025}, and the work of \citet{schipp-von-branitz-buchholtz:2021}, although we do impose our own terminological conventions.
    \end{remark}

    \begin{xsect}[sec:shallow]{Shallow characterisation of identity types}
      For example, it is very natural to build up an identification of type $\Id{A\times B}{\prn{x,u}}{\prn{y,v}}$ from a pair of identifications $\Id{A}{x}{y}$ and $\Id{B}{u}{v}$. This principle can be formalised by \emph{characterising} the identity type of $A\times B$ in terms of the identity types of $A$ and $B$.

      \begin{proposition}[Shallow characterisation of the binary product]\label[proposition]{ex:shallow-bin-prod}
        Each of the following functions is an equivalence:
        \begin{align*}
          x,y:A; u,v:B & \vdash \mathsf{splitId}_{A,B} \colon : \Id{A\times B}{\prn{x,u}}{\prn{y,v}} \to \prn{\Id{A}{x}{y}}\times \prn{\Id{B}{u}{v}}
          \\
          x:A, u:B     & \vdash \mathsf{splitId}_{A,B} ~\Refl :\equiv \prn{\Refl, \Refl}
        \end{align*}
      \end{proposition}
      \begin{proof}
        We can explicitly construct an inverse to $\mathsf{splitId}_{A,B}$ as follows:
        \begin{align*}
          x,y:A; u,v:B
           & \vdash
          \mathsf{unsplitId}_{A,B} \colon : \prn{\Id{A}{x}{y}}\times \prn{\Id{B}{u}{v}} \to \Id{A\times B}{\prn{x,u}}{\prn{y,v}}
          \\
          x:A, u:V
           & \vdash
          \mathsf{unsplitId}_{A,B}\,\prn{\Refl,\Refl} :\equiv \Refl
        \end{align*}

        The coherences of the equivalence are given by identification induction.
      \end{proof}

      Of course, although the characterisation of identifications in $A\times B$ in \cref{ex:shallow-bin-prod} is useful, it does not help us directly to characterise identifications in more complex types. For example, we may wish to characterise the identity type of $\prn{A\times B}\times C$ by asserting that the canonical map $\Id{\prn{A\times B}\times C}{\prn{\prn{x,u},m}}{\prn{\prn{y,v},n}}\to \prn{\prn{\Id{A}{x}{y}}\times\prn{\Id{B}{u}{v}}}\times\prn{\Id{C}{m}{n}}$ is an equivalence. This can be established directly in the same way as \cref{ex:shallow-bin-prod} or by making use of the lemma, but the latter saves very little time.  We would prefer to have a toolkit for building up ``deep'' characterisations of identity types piece-by-piece.
    \end{xsect}

    \begin{xsect}{Deep characterisation using reflexive graphs}
      A ``deep'' characterisation of the identity type for binary products would take the following form:
      \emph{If we have characterised the identity type of $A$ and the identity type of $B$, then we may characterise the identity type of $A\times B$.}

      \cref{ex:shallow-bin-prod} can be ``deepened'' along these lines by replacing the types $A$ and $B$ with \emph{reflexive graphs} $\gA$ and $\gB$ so that we have the following data:

      \begin{multicols}{2}
        \iblock{
          \mrow{\vrt{\gA} : \TYPE}
          \mrow{{\Edge{\gA}} : \vrt{\gA}\to \vrt{\gA}\to \TYPE}
          \mrow{\Rx{\gA} : \Prod{x:\vrt{\gA}} x\Edge{\gA}x}
          \columnbreak
          \mrow{\vrt{\gB} : \TYPE}
          \mrow{{\Edge{\gB}} : \vrt{\gB}\to \vrt{\gB}\to \TYPE}
          \mrow{\Rx{\gB} : \Prod{x:\vrt{\gB}} x\Edge{\gB}x}
        }
      \end{multicols}

      These reflexive graph structures follow the pattern of instances of the formation and introduction rules for the identity type.  Tentatively, we shall say that a given reflexive graph $\gG$ is \DefEmph{univalent} when the following function is an equivalence:\footnote{See \cref{def:univalent-reflexive-graph}; in fact, our provisional definition of the univalence condition for a reflexive graph is not the most practical one for everyday use, but it is one of several equivalent conditions that we shall expose later.}
      \begin{align*}
        x,y:\vrt{\gG} & \vdash \mathsf{idToEdge}_\gG^{x,y} : \Id{\vrt{\gG}}{x}{y} \to x\Edge{\gG}y
        \\
        x:\vrt{\gG}   & \vdash \mathsf{idToEdge}_\gG^{x,x} \Refl :\equiv \Rx{\gG}{x}
      \end{align*}

      We shall often refer to a univalent reflexive graph as a \DefEmph{path object} for short. With these constructions on reflexive graphs in hand, a \emph{deep} characterisation of the identity type for binary products can be formulated.

      First we define a new reflexive graph $\gA\times \gB$ as follows:
      \begin{align*}
        \vrt{\gA\times \gB}
         & :\equiv
        \vrt{\gA}\times \vrt{\gB}
        \\
        \prn{x,u} \Edge{\gA\times \gB} \prn{y,v}
         & :\equiv
        \prn{x\Edge{\gA}y}
        \times
        \prn{u\Edge{\gB}v}
        \\
        \Rx{\gA\times\gB}\prn{x,u}
         & :\equiv
        \prn{\Rx{\gA}{x},\Rx{\gB}{u}}
      \end{align*}

      Then the deep characterisation of binary products amounts to saying that path objects (\ie univalent reflexive graphs) are closed under binary products.

      \begin{proposition}[Deep characterisation of the binary product]\label[proposition]{prop:deep-bin-prod}
        If $\gA$ and $\gB$ are two path objects, then $\gA\times \gB$ is a path object.
      \end{proposition}

      The compositional character of \cref{prop:deep-bin-prod} allows us to deduce other useful corollaries immediately by iteration. For example, if $\gA,\gB,\gC$ are all univalent reflexive graphs, then $\prn{\gA\times\gB}\times\gC$ is a univalent reflexive graph, \etc. On the other hand, the shallow characterisation (\cref{ex:shallow-bin-prod}) is immediately recoverable via the following definition of the \DefEmph{discrete} reflexive graph on a type $A$, which is univalent by definition:
      \begin{align*}
        \vrt{\DiscPO{A}}
         & :\equiv
        A
        \\
        x\Edge{\DiscPO{A}} y
         & :\equiv
        \Id{A}{x}{y}
        \\
        \Rx{\DiscPO{A}}{x}
         & :\equiv
        \Refl
      \end{align*}

      Then \cref{ex:shallow-bin-prod} is precisely the instantiation of \cref{prop:deep-bin-prod} with $\gA :\equiv \DiscPO{A}$ and $\gB :\equiv \DiscPO{B}$.
      We can also generalise the binary product of reflexive graphs to arbitrary arity: if $A$ is a type and $\gB\prn{x}$ is a reflexive graph for each $x:A$, define $\Prod{x:A}\gB\prn{x}$ to be the following reflexive graph:
      \begin{align*}
        \vrt{\Prod{x:A}\gB\prn{x}}
         & :\equiv
        \Prod{x:A}\vrt{\gB\prn{x}}
        \\
        f \Edge{\Prod{x:A}\gB\prn{x}} g
         & :\equiv
        \Prod{x:A}
        fx\Edge{\gB\prn{x}} gx
        \\
        \Rx{\Prod{x:A}\gB\prn{x}}{f}
         & :\equiv
        \Lam{x}\Rx{\gB\prn{x}}\prn{fx}
      \end{align*}

      \begin{proposition}[Deep characterisation of product]
        Given $\gB\prn{x}$ a path object for each $x:A$, the product $\Prod{x:A}\gB\prn{x}$ is univalent assuming dependent function extensionality holds.
      \end{proposition}
    \end{xsect}

    \begin{xsect}{Dependent types and displayed reflexive graphs}

      So far we have described how to give shallow and deep characterisations of \emph{non-dependent} types. Of course, most mathematical structures of any interest (\eg monoids, groups, rings, \etc, or even reflexive graphs themselves!) are described by dependent types, as the types of the operations depend on the carrier types.

      Given a type $A:\TYPE$ and a family of types $x:A\vdash B\prn{x} : \TYPE$, how can we characterise the identity type of the sum $\Sum{x:A}B\prn{x}$? A shallow characterisation in the style of \cref{sec:shallow} is at least not out of reach.

      \begin{proposition}[Shallow characterisation of the indexed sum]\label[proposition]{prop:shallow-dependent}
        Each of the following functions is an equivalence:

        \iblock{
          \mhang{
            x,y:A; u:B\prn{x}, v:B\prn{y} \vdash
          }{
            \mrow{
              \mathsf{dsplitId}_{A,B} \colon \Id{\Sum{x:A}{B\prn{x}}}{\prn{x,u}}{\prn{y,v}}
              \to
              \Sum{p:\Id{A}{x}{y}}
              \Id{B\prn{y}}{p_*^B u}{v}
            }
            \mrow{
              \mathsf{dsplitId}_{A,B}\,\Refl :\equiv \prn{\Refl,\Refl}
            }
          }
        }
      \end{proposition}

      The shallow characterisation in \cref{prop:shallow-dependent} is more complex than that of \cref{ex:shallow-bin-prod} because the second component of the the sum $\Sum{p:\Id{A}{x}{y}}
        \Id{B\prn{y}}{p_*^B u}{v}$ involves a \emph{transport}. From a higher vantage point, the problem being solved by transport here is that we wish to identify $u:B\prn{x}$ with $v:B\prn{y}$ but these do not have the same type. Transport uses the identification $p:\Id{A}{x}{y}$ to find a common type in which $u$ and $v$ can be compared; but note that the choice of a forward transport $p_*^B\colon B\prn{x}\to B\prn{y}$ is somewhat arbitrary, as we might have also defined a backward transport $p^*_B \colon B\prn{y}\to B\prn{x}$ and compared $u,v$ in $B\prn{x}$.

      In order to generalise \cref{prop:shallow-dependent} to a \emph{deep} characterisation of the indexed sum, we must naturally generalise the concept of reflexive graph to one in which vertices form a \emph{dependent} type and edges go between different components linked by an edge in the base. This is precisely the purpose of the \emph{displayed reflexive graphs} of \cite{schipp-von-branitz-buchholtz:2021}.

      In particular, let $\gA$ be a reflexive graph as before and let $\gB$ be a \DefEmph{displayed reflexive graph} over $\gA$ so that we have the following data:

      \begingroup
      \setlength\columnsep{-4cm}
      \begin{multicols}{2}
        \iblock{
          \mrow{\vrt{\gA} : \TYPE}
          \mrow{{\Edge{\gA}} : \vrt{\gA}\to \vrt{\gA}\to \TYPE}
          \mrow{\Rx{\gA} : \Prod{x:\vrt{\gA}} x\Edge{\gA}x}
          \columnbreak
          \mrow{\vrt{\gB} : \vrt{\gA}\to \TYPE}
          \mrow{{\Edge{\gB}[\bullet]} : \Prod[\brc]{x,y:\vrt{\gA}} x\Edge{\gA}y\to \vrt{\gB}\prn{x}\to \vrt{\gB}\prn{y}\to \TYPE}
          \mrow{\DRx{\gB}{\bullet} : \Prod{x:\vrt{\gA}}\Prod{u:\vrt{\gB}\prn{x}} u\Edge{\gB}[\Rx{\gA}{x}]u}
        }
      \end{multicols}
      \endgroup

      Any displayed reflexive graph such as $\gB$ gives rise to a family of ordinary reflexive graphs $\gB\prn{x}$ indexed in vertices $x:\vrt{\gA}$.
      \begin{align*}
        \vrt{\gB\prn{x}}
         & :\equiv
        \vrt{\gB}\prn{x}
        \\
        u \Edge{\gB\prn{x}} v
         & :\equiv
        u \Edge{\gB}[\Rx{\gA}{x}] v
        \\
        \Rx{\gB\prn{x}}{u}
         & :\equiv
        \DRx{\gB}{x}{u}
      \end{align*}

      A displayed reflexive graph $\gB$ will be called \DefEmph{univalent} (or a displayed path object for short) when each of its components is univalent. Considering only the components for univalence may seem at first counterintuitive, but this definition is justified by \cref{prop:deep-dependent} below. For this, we first define the \DefEmph{total reflexive graph} of $\gB$ with vertices in the sum of $\vrt{\gB}$ over $\vrt{\gA}$ as follows:
      \begin{align*}
        \vrt{\gA.\gB}
         & :\equiv
        \Sum{x:\vrt{\gA}}\vrt{\gB}\prn{x}
        \\
        \prn{x,u} \Edge{\gA.\gB} \prn{y,v}
         & :\equiv
        \Sum{p:x\Edge{\gA}y}
        u\Edge{\gB}[p] v
        \\
        \Rx{\gA.\gB}{\prn{x,u}}
         & :\equiv
        \prn{\Rx{\gA}{x},\DRx{\gB}{x}{u}}
      \end{align*}

      Then the deep characterisation for indexed sums is exactly the statement that the total reflexive graph of a displayed path object is a path object.

      \begin{proposition}[Deep characterisation of the indexed sum]\label[proposition]{prop:deep-dependent}
        Let $\gB$ be a displayed path object over a reflexive graph $\gA$. Then the total reflexive graph $\gA.\gB$ is a path object.
      \end{proposition}

    \end{xsect}
  \end{xsect}

  \begin{xsect}{Characterising \emph{transport} with reflexive graph lenses}
    Everything we have described so far is more or less standard. Path objects (and equivalent concepts such as \emph{torsorial families}) have been used to great effect in a variety of formalised libraries of univalent mathematics, including \texttt{agda-unimath}~\citep{agda-unimath} and the 1Lab~\citep{1lab}. The purpose of the present paper is to give a deeper analysis of a large class of displayed path objects that arise in a particularly simple way, in fact obviating the need to separately prove many univalence lemmas. We return to the shallow characterisation of the identity types of dependent sums from \cref{prop:shallow-dependent}:
    \[
      \Id{\Sum{x:A}B\prn{x}}{\prn{x,u}}{\prn{y,v}} \cong
      \Sum{p:\Id{A}{x}y}
      \Id{B\prn{y}}{p_*^Bu}{v}
    \]

    In our deep characterisation (\cref{prop:deep-dependent}), we generalised $\Id{A}{x}{y}$ to the edges $x\Edge{\gA}y$ of a reflexive graph $\gA$ and we generalised $\Id{B\prn{y}}{p_*^Bu}{v}$ to the \emph{displayed edges} $u\Edge{\gB}[p]{v}$ of a displayed reflexive graph $\gB$ over $\gA$. This generalisation abstracts away the use of transport to get the vertices $u$ and $v$ to lie in the same component of $B$, at the cost of needing to specify $\gB$ as a displayed reflexive graph rather than as a family of reflexive graphs indexed in $\gA$.

    Our starting point is to consider a different, less abstract, generalisation of \cref{prop:shallow-dependent} in which we replace $A$ with a reflexive graph $\gA$ as before, but we replace $B$ not with a displayed reflexive graph over $\gA$, but instead with a \emph{family} of reflexive graphs $\gB\prn{x}$ indexed in vertices $x:\vrt{\gA}$ that is equipped with a \emph{transport} or \emph{pushforward} operation to jump from one component to the next. To a first approximation, we have assumed the following data:

    \iblock{
      \setlength\columnsep{-3cm}
      \begin{multicols}{2}
        \mrow{\vrt{\gA} : \TYPE}
        \mrow{{\Edge{\gA}} : \vrt{\gA}\to \vrt{\gA}\to \TYPE}
        \mrow{\Rx{\gA} : \Prod{x:\vrt{\gA}} x\Edge{\gA}x}
        \columnbreak
        \mrow{\vrt{\gB\prn{-}} : \vrt{\gA}\to \TYPE}
        \mrow{{\Edge{\gB\prn{-}}} : \Prod{x:\vrt{\gA}} \vrt{\gB\prn{x}}\to\vrt{\gB\prn{x}}\to \TYPE}
        \mrow{\Rx{\gB\prn{-}} : \Prod{x:\vrt{\gA}}\Prod{u:\vrt{\gB}\prn{x}} u\Edge{\gB}u}
      \end{multicols}
      \mrow{
        \Push{\gB}{\bullet} :
        \Prod[\brc]{x,y:\vrt{\gA}}
        \Prod{p:x\Edge{\gA}y}
        \vrt{\gB\prn{x}}\to \vrt{\gB\prn{y}}
      }
    }

    With this data in hand, we can define an appropriate displayed reflexive graph $\CovDisp{\gA}{\gB}$ over $\gA$ that mirrors the passage from the shallow characterisation (\cref{prop:shallow-dependent}) to the deep characterisation (\cref{prop:deep-dependent}):
    \begin{align*}
      \vrt{\CovDisp{\gA}{\gB}}\prn{x}
       & :\equiv
      \vrt{\gB\prn{x}}
      \\
      u \Edge{\CovDisp{\gA}{\gB}}[p:x\Edge{\gA}y] v
       & :\equiv
      \Push{\gB}{p}{u}\Edge{\gB\prn{y}} v
    \end{align*}

    It remains to define the displayed reflexivity datum $\DRx{\CovDisp{\gA}{\gB}}{x}{u} : \Push{\gB}{\Rx{\gA}{x}}{u} \Edge{\gB\prn{x}} u$, but we do not have anything on hand from which to define this. We are therefore led to assert this reflexivity datum as part of the \emph{data} of $\gB$:

    \iblock{
      \mrow{
        \PushRx{\gB}[\bullet] : \Prod{x:\vrt{\gA}} \Prod{u:\vrt{\gB\prn{x}}} \Push{\gB}{\Rx{\gA}{x}}{u} \Edge{\gB\prn{x}} u
      }
    }

    The above has the appearance of an \emph{oplax unitor} for the pushforward operation. Using this oplax unitor, we may finish defining the displayed reflexive graph $\CovDisp{\gA}{\gB}$:
    \begin{align*}
      \DRx{\CovDisp{\gA}{\gB}}{x}{u}
       & :\equiv
      \PushRx{\gB}[x]{u}
    \end{align*}

    Altogether, we shall refer to a family of reflexive graphs $\gB$ equipped with $\Push{\gB}{\bullet}$ and $\PushRx{\gB}[\bullet]$ as an \DefEmph{oplax covariant lens} of reflexive graphs (see \cref{def:oplax-cov-lens}) over $\gA$ --- \emph{covariant} because $\Push{\gB}{\bullet}$ implements a forward transport, and \emph{oplax} because of the orientation of the unitor $\PushRx{\gB}$. The terminology of \emph{lenses} is borrowed from \citet{johnson-rosebrugh:2013,chollet-et-al:2022:lenses}, who use it to refer to an algebraic generalisation of fibrations in which the chosen lifts are not required to have a universal property but instead satisfy a unit law (strictly, in the case of \opcit).

    Naturally, we can (and will) introduce a dual notion of \emph{lax contravariant lens} for characterising backward transport in a family of reflexive graphs. When $\gB$ is a lax contravariant lens in this sense (see \cref{def:lax-ctrv-lens}), we can likewise associate a displayed reflexive graph $\CtrvDisp{\gA}{\gB}$ in which a displayed edge $u\Edge{\CtrvDisp{\gA}{\gB}}[p:x\Edge{\gA} y] v$ is given by an edge $u \Edge{\gB\prn{x}} \Pull{\gB}{p}{v}$.

    \begin{proposition}[See \cref{lem:cov-disp-po}]
      Let $\gB$ be an oplax covariant (\resp lax contravariant) lens of path objects over a reflexive graph $\gA$. Then the displayed reflexive graph $\CovDisp{\gA}{\gB}$ (\resp $\CtrvDisp{\gA}{\gB}$) is univalent.
    \end{proposition}

    The pay-off of introducing lenses of reflexive graphs is twofold. First of all, many naturally occurring displayed reflexive graphs have the shape of $\CovDisp{\gA}{\gB}$ or $\CtrvDisp{\gA}{\gB}$ already; more importantly, however, it frequently happens that the components of the given displayed reflexive graph arise as a pre-existing family of reflexive graphs for which we have already proved univalence. Therefore, it is advantageous to obtain a displayed path object automatically from a very simple algebraic structure on a pre-existing family of path objects: a pushforward operator and a lax unitor.

  \end{xsect}

  \begin{xsect}{Characterising \emph{identification induction} with dependent lenses}
    \NewDocumentCommand\MAGMA{}{\mathsf{Magma}_U}
    \NewDocumentCommand\BINOP{}{\mathsf{BinOp}_U}

    Naturally, not all useful displayed path objects arise from (oplax covariant, lax contravariant) lenses. We have found, however, that many of the important counterexamples instead arise from a common generalisation of oplax covariant and lax contravariant lenses, which shares the advantages thereof over working directly with displayed reflexive graphs.
    Consider for example the case of \emph{magmas} in a universe $U$, which are specified by the following data:

    \iblock{
      \mrow{
        \vrt{M} : \TYPE
      }
      \mrow{
        \otimes_M : \vrt{M}\times\vrt{M}\to \vrt{M}
      }
    }

    An equivalence of magmas $M\Edge{\MAGMA} N$ is given by an equivalence of types $f : \vrt{M} \cong \vrt{N}$ that preserves the binary operation in the following sense:
    \[
      \DiagramSquare{
        nw = \vrt{M}\times\vrt{M},
        sw = \vrt{N}\times\vrt{N},
        ne = \vrt{M},
        se = \vrt{N},
        east = f,
        west = f\times f,
        south = \otimes_N,
        north = \otimes_M
      }
    \]

    If we write $\gU$ for the reflexive graph structure on the universe $U$ given by equivalences of types, then we would expect that the reflexive graph $\MAGMA$ should arise as the total reflexive graph of the following displayed reflexive graph over $\gU$:
    \begin{align*}
      \vrt{\BINOP}\prn{A}
       & :\equiv
      A\times A\to A
      \\
      {\otimes_A} \Edge{\BINOP}[f:A\Edge{\gU}B] {\otimes_B}
       & :\equiv
      \Prod{x,y:A}
      \Id{B}{f\prn{x \otimes_A y}}{\prn{fx\otimes_B fy}}
      \\
      \DRx{\BINOP}{A}{\otimes_A}
       & :\equiv
      \Lam{x,y} \Refl
    \end{align*}

    The displayed reflexive graph above does \emph{not} arise from an oplax covariant lens, nor from a lax contravariant lens. The following somewhat artificial displayed reflexive graphs could have been obtained from such lenses, but they would be difficult to use:
    \begin{align*}
      \vrt{\BINOP^+}\prn{A}
       & :\equiv
      A\times A\to A
      \\
      {\otimes_A} \Edge{\BINOP^+}[f:A\Edge{\gU}B] {\otimes_B}
       & :\equiv
      \Prod{x,y:B}
      \Id{B}{f\prn{f\Inv x\otimes_A f\Inv y}}{x\otimes_B y}
      \\
      \DRx{\BINOP^+}{A}{\otimes_A}
       & :\equiv
      \Lam{x,y} \Refl
      \\[10pt]
      \vrt{\BINOP^-}\prn{A}
       & :\equiv
      A\times A\to A
      \\
      {\otimes_A} \Edge{\BINOP^-}[f:A\Edge{\gU}B] {\otimes_B}
       & :\equiv
      \Prod{x,y:A}
      \Id{A}{x\otimes_A y}{f\Inv\prn{fx\otimes_B fy}}
      \\
      \DRx{\BINOP^+}{A}{\otimes_A}
       & :\equiv
      \Lam{x,y} \Refl
    \end{align*}

    In order to capture the structure necessary to obtain the displayed reflexive graph $\BINOP$ from a \emph{family} of ordinary reflexive graphs, we must consider a different kind of lens. Before generalising, we first work with this specific case. Let $\BINOP^\pm\prn{f}$ be a reflexive graph as defined below for each equivalence $f:A\Edge{\gU}B$ of types:
    \[
      \BINOP^\pm\prn{f:A\Edge{\gU}B} :\equiv
      \Prod{\_:A\times A}
      \DiscPO{B}
    \]

    Then we define the following operations on $\BINOP^\pm$:
    \iblock{
      \mrow{
        \LJ{\BINOP^\pm}{\bullet} :
        \Prod[\brc]{A,B:U}
        \Prod{f:A\Edge{\gU}B}
        \BINOP^\pm\prn{\Rx{\gU}{A}}
        \to \BINOP^\pm\prn{f}
      }
      \mrow{
        \LJ{\BINOP^\pm}{f:A\Edge{\gU}B}{\otimes_A}
        :\equiv
        \Lam{\prn{x,y}:A\times A}
        f\prn{x\otimes_A y}
      }
      \row
      \mrow{
        \RJ{\BINOP^\pm}{\bullet} :
        \Prod[\brc]{A,B:U}
        \Prod{f:A\Edge{\gU}B}
        \BINOP^\pm\prn{\Rx{\gU}{B}}
        \to \BINOP^\pm\prn{f}
      }
      \mrow{
        \RJ{\BINOP^\pm}{f:A\Edge{\gU}B}{\otimes_B}
        :\equiv
        \Lam{\prn{x,y}:A\times A}
        fx\otimes_B fy
      }

      \row

      \mrow{
        \MidJRx{\BINOP^\pm}[\bullet]
        :
        \Prod{A:U}
        \Prod{\otimes_A : A\times A\to A}
        \LJ{\BINOP^\pm}{\Rx{\gU}{A}}{\otimes_A}
        \Edge{\BINOP^\pm\prn{\Rx{\gU}{A}}}
        \RJ{\BINOP^\pm}{\Rx{\gU}{A}}{\otimes_A}
      }

      \mrow{
        \MidJRx{\BINOP^\pm}[A]{\otimes_A}
        :\equiv
        \Rx{\Prod{\_:A\times A}\DiscPO{B}}{\otimes_A}
      }
    }

    We can now reconstruct the displayed reflexive graph $\BINOP$ in terms of $\BINOP^\pm$:
    \begin{align*}
      \vrt{\BINOP}\prn{A}
       & \equiv
      \vrt{\BINOP^\pm\prn{\Rx{\gU}{A}}}
      \\
      {\otimes_A} \Edge{\BINOP}[f:A\Edge{\gU}B] {\otimes_B}
       & \equiv
      \LJ{\BINOP^\pm}{f}{\otimes_A}
      \Edge{\BINOP^\pm\prn{f}}
      \RJ{\BINOP^\pm}{f}{\otimes_B}
      \\
      \DRx{\BINOP}{A}{\otimes_A}
       & \equiv
      \MidJRx{\BINOP^\pm}[A]{\otimes_A}
    \end{align*}

    Abstracting from the specific example of binary operations, the critical move above has been (1) to replace a given family of reflexive graphs $x:\vrt{\gA}\vdash \gB\prn{x}$ with a more general family $x,y:\vrt{\gA};p:x\Edge{\gA}y \vdash \gB'\prn{x,y,p}$
    such that $\gB'\prn{x,x,\Rx{\gA}{x}}\equiv \gB\prn{x}$, and then (2) define coercions $\LJ{\gB'}{\bullet} : \gB\prn{x} \to \gB'\prn{x,y,p}$ and $\RJ{\gB'}{\bullet} : \gB\prn{y}\to \gB'\prn{x,y,p}$ from the ``left-hand'' and ``right-hand'' diagonal components to the ``centre''. Naturally these coercions must be further equipped with at least a coherence of the form $\MidJRx{\gB'}[x]u : \LJ{\gB'}{\Rx{\gA}{x}}u \Edge{\gB\prn{x}} \RJ{\gB'}{\Rx{\gA}{x}}u$ in order to generate a reflexivity datum; it will happen that in order for the univalence property carry over from $\gB$ to the associated displayed reflexive graph, we shall also need a coherence of the form $\RJRx{\gB'}[x]{u} : u \Edge{\gB\prn{x}} \RJ{\gB'}{\Rx{\gA}{x}}{u}$. Together, all this data forms what we shall refer to as a \DefEmph{unbiased dependent lens} of reflexive graphs over $\gA$.

    \begin{remark}
      In the same way as oplax covariant (\resp lax contravariant) lenses express the interface of forward (\resp backward) \emph{transport} for a family of reflexive graphs, unbiased dependent lenses express the interface of (forward and backward) \emph{based identification induction}: the coercions $\LJ{}{}$ and $\RJ{}{}$ extend data defined on a given reflexivity datum $\Rx{\gA}{x}$ to data defined on an arbitrary path based at $x:\vrt{\gA}$.
    \end{remark}
  \end{xsect}

  \begin{xsect}{Discussion of related work}

    \begin{xsect}{Path objects in homotopy type theory}
      A biased version of path objects was introduced already in the HoTT Book~\citep[\S~5.8]{hottbook} under the name \emph{identity system}; this concept is developed much further in the displayed direction by Rijke in his introductory textbook on homotopy type theory~\citep[\S~11.2]{rijke:2025}. In their guise as \emph{torsorial families}, biased identity systems are employed pervasively in \texttt{agda-unimath}, an ``online encyclopedia of formalized mathematics ... from a univalent point of view''~\citep{agda-unimath}. An unbiased version of identity systems is used extensively in the 1Lab~\citep{1lab}, another library of univalent mathematics formalised in Cubical Agda~\citep{vezzosi-mortberg-abel:2019}.

      Our work is most directly inspired by that of \citet{schipp-von-branitz-buchholtz:2021,schipp-von-branitz:2020:thesis}, who have emphasised the importance of displayed path objects as an organising principle for univalent mathematics. The various notions of lens that we introduce can be thought of as organisational devices to simplify carrying out the methods suggested by \opcit in practice.
    \end{xsect}

    \begin{xsect}{Reflexive graphs and reflexive graph fibrations}
      Although \emph{univalent} reflexive graphs are the ones that are important for characterising identity types, we have found it very useful to develop as much structure as possible in the language of ordinary reflexive graphs; this is because the reflexive graph structure captures the aspects of an identity type's characterisation that we intend to have good definitional/computational behavior, whereas the univalent part describes a universal property that we expect to hold (opaquely) only up to homotopy. Our starting point in studying the theory of reflexive graphs has been the doctoral dissertation of \citet{rijke:2019}, who introduces many important concepts from the univalent point of view, including the univalence condition (called \emph{discreteness} by \opcit) as well as the notion of (covariant, contravariant) fibration of reflexive graphs. (It should be noted that the concept of fibred (non-reflexive) graph was introduced already by \citet{boldi-vigna:2002}.)

      \citet{schipp-von-branitz-buchholtz:2021,schipp-von-branitz:2020:thesis} introduced \emph{displayed reflexive graphs} as a more type theoretical presentation of homomorphisms of reflexive graphs; we have re-based \citet{rijke:2019}'s theory of reflexive graph fibrations in terms of displayed reflexive graphs in roughly the same way that \citet{ahrens-lumsdaine:2019} re-examine (Street) fibrations of categories in terms of displayed categories.
    \end{xsect}

    \begin{xsect}[sec:sip]{The structure identity principle}
      \citet{coquand-danielsson:2013} have famously observed that Voevodsky's univalence principle ensures that the identity types of various algebraic structures (\eg groups, rings, graphs, \etc) can be characterised so as to coincide precisely with the natural notion of invertible homomorphism: for example, an identification of groups valued in a univalent universe is precisely a group isomorphism. Results of this kind are called \emph{structure identity principles}; a convenient general form of the structure identity principle for concrete categories appears in the HoTT Book~\citep[\S~9.8]{hottbook}; a simpler version was introduced by \citet{escardo:2022:introduction} for his lectures at the Midlands Graduate School in 2019, and included as part of the \texttt{TypeTopology} and \texttt{cubical} Agda libraries~\citep{type-topology,agda:cubical-lib:2024}.

      These structure identity principles are the special case of the more general theory of path objects, restricted to the case of displaying an algebraic structure over the path object given by types and equivalences. Of course, not all important path objects take this form: for example, it is necessary to characterise the identity type of the type of sections of a given map, and the base of this thing is a function space rather than a universe. This is why it is important to develop the theory of path objects in its most general form, and then specialise it to obtain a variety of structure identity principles.
    \end{xsect}

    \begin{xsect}{Delta lenses and applied category theory}
      \citet{diskin-xiong-czarnecki:2011} have introduced \emph{delta lenses} as an algebraic structure modeling bidirectional transformations between systems (modelled as categories). As many authors have pointed out~\citep{chollet-et-al:2022:lenses,johnson-rosebrugh:2013,clarke:2020:tac}, delta lenses can be thought of as an algebraic generalisation of split opfibrations in which lifts are chosen functorially but are not required to be cocartesian.

      Our reflexive graph lenses resemble a version of delta lenses that replaces categories with reflexive graphs, but there are some important differences. When $\gB' :\equiv \CovDisp{\gA}{\gB}$ is the displayed reflexive graph associated to an oplax covariant lens $\gB$, we do indeed assign to each $u:\vrt{\gB'}\prn{x}$ and $p:x\Edge{\gA}y$ a displayed vertex $\Push{\gB}{p}{u} : \vrt{\gB'}\prn{x}$ and a lift $\bar{p}: u\Edge{\gB'}[p]\Push{\gB}{p}{u}$ given by the reflexivity datum on the former in $\gB\prn{y}$. In contrast to delta lenses, we do \emph{not} require that $\Push{\gB}{\Rx{\gA}{x}}u$ is equal to $u$: we require only an \emph{oplax unitor}, \ie an edge $\PushRx{\gB}[x]{u} : \Push{\gB}{\Rx{\gA}{x}}{u}\Edge{\gB\prn{y}}u$.

      We also consider contravariant versions of reflexive graph lenses, which turn out to be equally important in practice. Our alignment of (covariant, contravariant) and (oplax, lax) is not chosen arbitrarily, but is rather forced by the need to transform a lens of whichever variance into a displayed reflexive graph.

      In \cref{sec:definitional-lenses}, we study an important class of reflexive graph lenses in which pushforward along the reflexivity edge is definitionally equivalent to the identity function. These \emph{definitional lenses} correspond most closely to delta lenses, and arise frequently when establishing structure identity principles in the sense of \cref{sec:sip}.
    \end{xsect}
  \end{xsect}

  \begin{xsect}{The structure of this paper}
    In \cref{sec:reflexive-graphs-and-path-objects}, we shall recall the theory of (displayed) reflexive graphs and path objects in detail (at the risk of some repetition). In \cref{sec:lenses}, we expose the theory of (oplax covariant, lax contravariant, and unbiased dependent) lenses and prove the uniqueness of all these structures on families of path objects over a path object. In \cref{sec:classifying-rx-gphs}, we develop a case study that applies the theory of reflexive graph lenses to characterise the identity types of reflexive graphs, displayed reflexive graphs, \etc --- in short, we construct large reflexive graphs classifying small reflexive graphs, \etc. Finally, in \cref{sec:fibrations}, we develop the theory of reflexive graph fibrations (building on \citet{rijke:2019}) and conclude with our main result: an equivalence between reflexive graph fibrations and lenses of path objects.
  \end{xsect}

\end{xsect}

\nocite{awodey:2014}
\nocite{rijke:2025}
\nocite{hottbook}
\nocite{schipp-von-branitz:2020:thesis}
\nocite{schipp-von-branitz-buchholtz:2021}
\nocite{ahrens-north-shulman-tsementzis:2022}
\nocite{aczel:2012:sip}
\nocite{coquand-danielsson:2013}
\nocite{agda-unimath,1lab}
\begin{xsect}[sec:reflexive-graphs-and-path-objects]{Reflexive graphs and path objects}

  We recall the theory of reflexive graphs from \citet{rijke:2019}, incorporating the viewpoint of \emph{displayed} reflexive graphs from \citet{schipp-von-branitz-buchholtz:2021}.

  \begin{definition}[Reflexive graph]\label[definition]{def:reflexive-graph}
    A \DefEmph{reflexive graph} $\gA \equiv \prn{\vrt{\gA},\Edge{\gA},\Rx{\gA}}$ is defined to be a type $\vrt{\gA}:\TYPE$ of \emph{vertices} together with a family of types $x:\vrt{\gA},y:\vrt{\gA}\vdash x\Edge{\gA} y$ of \emph{edges} and a reflexivity datum $x:\vrt{\gA}\vdash \Rx{\gA}{x}:x\Edge{\gA} x$.
  \end{definition}

  At times we may refer to \emph{homomorphisms} of reflexive graphs.

  \begin{definition}
    Let $\gA$ and $\gB$ be reflexive graphs. A \emph{homomorphism} $f\colon \gA\to \gB$ is given by the following data:
    \begin{align*}
      x:\vrt{\gA}&\vdash \vrt{f}(x): \vrt{\gB}\\
      x,y:\vrt{\gA};p:x\Edge{\gA}y&\vdash f^{\approx}_{x,y}(p) \colon \vrt{f}(x)\Edge{\gB}\vrt{f}(y)\\
      x:\vrt{\gA}&\vdash \Id{\vrt{f}x\Edge{\gB}\vrt{f}x}{f^{\approx}_{x,x}(\Rx{\gA}x)}{\Rx{\gB}(\vrt{f}x)}
    \end{align*}
  \end{definition}

  \begin{definition}[Displayed reflexive graph]\label[definition]{def:displayed-reflexive-graph}
    Let $\gA$ be a reflexive graph. A \DefEmph{displayed reflexive graph} $\gB \equiv \prn{\vrt{\gB},\Edge{\gB}[\bullet],\DRx{\gB}{\bullet}}$ over $\gA$ is given by the following data:
    \begin{align*}
      x:\vrt{\gA}
       & \vdash \vrt{\gB}\prn{x}:\TYPE\tag{vertices}
      \\
      p:x\Edge{\gA} y,u:\vrt{\gB}\prn{x},v:\vrt{\gB}\prn{y}
       & \vdash u\Edge{\gB}[p] v:\TYPE\tag{edges}
      \\
      x:\vrt{\gA},u:\vrt{\gB}\prn{x}
       & \vdash \DRx{\gB}{x}{u}:u \Edge{\gB}[\Rx{\gA}{x}] u\tag{reflexivity}
    \end{align*}
  \end{definition}

  \begin{definition}[Components of a displayed reflexive graph]\label[definition]{def:component}
    Let $\gB$ be a displayed reflexive graph over $\gA$. Then the \DefEmph{component} of $\gB$ at an element $x:\vrt{\gA}$ is the reflexive graph $\gB\prn{x}$ defined from $\gB$ as follows:
    \begin{align*}
      \vrt{\gB\prn{x}}
       & :\equiv \vrt{\gB}\prn{x}
      \\
      u\Edge{\gB\prn{x}}v
       & :\equiv u \Edge{\gB}[\Rx{\gA}{x}] v
      \\
      \Rx{\gB\prn{x}}{u}
       & :\equiv \DRx{\gB}{x}{u}
    \end{align*}
  \end{definition}

  \begin{notation}[Diagonal family of reflexive graphs]
    When $\gB$ is a displayed reflexive graph over $\gA$, we will at times write $\Diag{\gB}$ for the ``diagonal'' $\vrt{\gA}$-indexed family of reflexive graphs assigning to each $x:\vrt{\gA}$ the component $\Diag{\gB}\prn{x} :\equiv \gB\prn{x}$.
  \end{notation}

\begin{xsect}{Basic constructions on reflexive graphs}
  The most basic reflexive graph structure is the \emph{discrete} one provided by its identity type.

  \begin{definition}[Discrete reflexive graph]\label[definition]{ex:discrete-rx-gph}
    The \DefEmph{discrete reflexive graph} $\DiscPO{A}$ on a type $A$ is given by its identity type, \ie we define $\DiscPO{A} :\equiv \prn{A, \IdCon{A},\Refl}$.
  \end{definition}

  On the other extreme, we may consider the \emph{codiscrete} reflexive graph structure.

  \begin{definition}[Codiscrete reflexive graph]\label[definition]{ex:codiscrete-rx-gph}
    A type $A$ can be, conversely, equipped with the \DefEmph{codiscrete reflexive graph} structure $\CodiscPO{A}$ that identifies any two elements:
    \begin{align*}
      \vrt{\CodiscPO{A}}
      & :\equiv A
      \\
      x \Edge{\CodiscPO{A}}y
      & :\equiv \mathbf{1}
      \\
      \Rx{\CodiscPO{A}}{x}
      & :\equiv *
    \end{align*}
  \end{definition}

  \begin{definition}[Total reflexive graph]\label[definition]{def:total-gph}
    If $\gA$ is a reflexive graph and $\gB$ is a displayed reflexive graph over $\gA$, then the \DefEmph{total reflexive graph} $\gA.\gB$ can be defined with vertices in $\Sum{x:\vrt{\gA}}\vrt{\gB}\prn{x}$ and edges $\prn{a_0,b_0} \Edge{\gA.\gB} \prn{a_1,b_1}$ in $\Sum{p:a_0\Edge{\gA} a_1} b_0 \Edge{\gB}[p] b_1$, with reflexivity data given by $\Rx{\gA.\gB}{\prn{a,b}} :\equiv \prn{\Rx{A}{a},\Rx{B}{b}}$. Naturally, we have a homomorphism $\pi_{\gB}\colon \gA.\gB \to \gA$ of reflexive graphs:
    \begin{align*}
      \vrt{\pi_{\gB}}\,\prn{x,y}
       & :\equiv x
      \\
      \pi_{\gB}^{\approx}\,\prn{p,q}
       & :\equiv
      p
      \\
      \pi_{\gB}^{\RxCon}\,\prn{x,y}
       & :\equiv
      \Refl
    \end{align*}
  \end{definition}

  \begin{definition}[Binary product of reflexive graphs]\label[definition]{ex:bin-prod-rx-gph}
    If $\gA$ and $\gB$ are reflexive graphs, then we may form the product of reflexive graphs $\gA\times\gB$ with vertices in $\vrt{\gA}\times\vrt{\gB}$, where the type of edges $\prn{a_0,b_0}\Edge{\gA\times\gB} \prn{a_1,b_1}$ is the cartesian product $\prn{a_0\Edge{\gA} a_1} \times \prn{b_0\Edge{\gB} b_1}$ with reflexivity data defined by $\Rx{\gA\times\gB}{\prn{x,y}} :\equiv \prn{\Rx{\gA}{x},\Rx{\gB}{y}}$.
  \end{definition}

  We can generalise \cref{ex:bin-prod-rx-gph} to arbitrary families.

  \begin{definition}[Product of reflexive graphs]\label[definition]{ex:prod-rx-gph}
    Let $A$ be a type and let $\gB\prn{x}$ be a reflexive graph for each $x:A$. The product $\Prod{x:A}\gB\prn{x}$ of the family of $\gB$ reflexive graphs is defined like so:
    \begin{align*}
      \vrt{\Prod{x:A}\gB\prn{x}}
       & :\equiv
      \Prod{x:A}\vrt{\gB\prn{x}}
      \\
      f\Edge{\Prod{x:A}\gB\prn{x}}g
       & :\equiv
      \Prod{x:A}{fx}\Edge{\gB\prn{x}}{gx}
      \\
      \Rx{\Prod{x:A}\gB\prn{x}}{f}
       & :\equiv
      \lambda\Bind{x} \Rx{\gB\prn{x}}\prn{fx}
    \end{align*}
  \end{definition}

  \begin{definition}[Coproduct of reflexive graphs]\label[definition]{ex:coprod-rx-gph}
    Let $A$ be a type and let $\gB\prn{x}$ be a reflexive graph for each $x:A$. The coproduct $\Coprod{x:A}\gB\prn{x}$ of the family $\gB$ of reflexive graphs is defined like so:
    \begin{align*}
      \vrt{\Coprod{x:A}\gB\prn{x}}
       & :\equiv
      \Sum{x:A}\vrt{\gB\prn{x}}
      \\
      \prn{a_0,b_0} \Edge{\Coprod{x:A}\gB\prn{x}} \prn{a_1,b_1}
       & :\equiv
      \Sum{p:\Id{A}{a_0}{a_1}} p_*^{\vrt{\gB\prn{-}}}b_0 \Edge{\gB\prn{a_1}}  b_1
      \\
      \Rx{\Coprod{x:A}\gB\prn{x}}{\prn{a,b}}
       & :\equiv
      \prn{\Refl, \Rx{\gB\prn{a}}{b}}
    \end{align*}
  \end{definition}

  Later on in \cref{example:def-lens-over-disc-rx-gph} we will see how to refactor the construction of the coproduct of reflexive graphs in terms of a \emph{lens} over the discrete reflexive graph $\DiscPO{A}$.

  \begin{definition}[Tensor and cotensor of reflexive graphs]\label[definition]{ex:tensor-cotensor}
    Let $A$ be a type and let $\gB$ be a reflexive graph; then the \DefEmph{tensor} and \DefEmph{cotensor} of $\gB$ by $A$ are defined respectively using coproducts and products as below:
    \begin{align*}
      A\cdot\gB      & :\equiv \Coprod{x:A}\gB\tag{tensor} \\
      A\pitchfork\gB & :\equiv \Prod{x:A}\gB\tag{cotensor}
    \end{align*}
  \end{definition}

  \begin{definition}[Constant displayed reflexive graph]\label[definition]{ex:constant-disp-rx-gph}
    Let $\gA$ and $\gB$ be reflexive graphs. The \DefEmph{constant displayed reflexive graph} over $\gA$ induced by $\gB$ is the displayed reflexive graph specified $\gA^*\gB$ below:
    \begin{align*}
      \vrt{\gA^*\gB}\prn{x}
       & :\equiv
      \vrt{\gB}
      \\
      u \Edge{\gA^*\gB}[p] v
       & :\equiv
      u \Edge{\gB} v
      \\
      \DRx{\gA^*\gB}{x}{u}
       & :\equiv
      \Rx{\gB}{u}
    \end{align*}

    Thus, evidently, we have $\prn{\gA^*\gB}\prn{x}\equiv \gB$ for all $x:\vrt{\gA}$.
  \end{definition}

  \begin{observation}[Binary product as a total reflexive graph]\label[observation]{obs:rx-gph-bin-prod-vs-sum}
    For any two reflexive graphs $\gA$ and $\gB$, we have $\gA.\prn{\gA^*\gB} \equiv \gA\times \gB$.
  \end{observation}

  \begin{definition}[Subgraph comprehension]\label[definition]{def:rx-gph-comprehension}
    Let $\gA$ be a reflexive graph, and let $P\prn{x}$ be a proposition for each vertex $x:\gA$. We define the \DefEmph{comprehension} of $P$ to be the following reflexive graph:
    \begin{align*}
      \vrt{\Compr{x:\gA}{P\prn{x}}}
       & :\equiv
      \Sum{x:\vrt{\gA}}P\prn{x}
      \\
      \prn{x,p} \Edge{\Compr{x:\gA}{P\prn{x}}} \prn{y,q}
       & :\equiv
      x \Edge{\gA} y
      \\
      \Rx{\Compr{x:\gA}{P\prn{x}}}\prn{x,p}
       & :\equiv
      \Rx{\gA}{x}
    \end{align*}
  \end{definition}

  \begin{construction}[Restriction of iterated displayed reflexive graphs]\label[construction]{con:rst-disp-rx-gph}
    Let $\gA$ be a reflexive graph object, and let $\gB$ be a displayed reflexive graph over $\gA$, and let $\gC$ be a displayed reflexive graph over $\gA.\gB$. Then for any $x:\vrt{\gA}$, we may define a displayed reflexive graph $\gC\Sub{\vert \gB\prn{x}}$ over the component $\gB\prn{x}$ with vertices given as follows:

    \iblock{
      \mrow{
        u:\vrt{\gB}\prn{x}
        \vdash
        \vrt{\gC\Sub{\vert\gB\prn{x}}}\prn{u} :\equiv \vrt{\gC}\prn{x,u}
      }
    }

    Displayed edges and reflexivity data are defined in terms of those of $\gC$ like so:

    \iblock{
      \mhang{
        u,v:\vrt{\gB}\prn{x}; p : u\Edge{\gB\prn{x}} v;
        c : \vrt{\gC}\prn{x,u}, d : \vrt{\gC}\prn{x,v}
      }{
        \mrow{
          \vdash
          c \Edge{\gC\Sub{\vert\gB\prn{x}}}[p] d :\equiv
          c \Edge{\gC}[
            \prn{\Rx{\gA}{x},p}
          ] d
        }
      }
      \mhang{
        u:\vrt{\gB}\prn{x}, c : \vrt{\gC}\prn{x,u}
      }{
        \mrow{
          \vdash
          \DRx{\gC\Sub{\vert \gB\prn{x}}}{u}{c}
          :\equiv
          \DRx{\gC}{\prn{\Rx{\gA}{x},u}}{c}
        }
      }
    }
  \end{construction}

  \begin{computation}[Components of the restriction]\label[computation]{cmp:component-of-rst-disp-rx-gph}
    Let $\gA$ be a reflexive graph, and let $\gB$ be a displayed reflexive graph over $\gA$, and let $\gC$ be a displayed reflexive graph over $\gA.\gB$. Given $x:\vrt{\gB}$ and $u:\vrt{\gB}\prn{x}$, the component $\prn{\gC\Sub{\vert \gB\prn{x}}}\prn{u}$ of the restriction of $\gC$ to $\gB\prn{x}$ at $u$ is definitionally equal to the component $\gC\prn{\prn{x,u}}$ of $\gC$ at $\prn{x,u}:\vrt{\gA.\gB}$.
  \end{computation}

\end{xsect}

\begin{xsect}[sec:rx-gph-duality]{Duality involution for reflexive graphs}

  \begin{definition}[Opposite reflexive graph]\label[definition]{def:op-rx-gph}
    We define the \DefEmph{opposite} of a reflexive graph $\gA$ as follows:
    \begin{align*}
      \vrt{\gA\Op}
      &:\equiv
      \vrt{\gA}
      \\
      x \Edge{\gA\Op} y
      &:\equiv
      y\Edge{\gA} x
      \\
      \Rx{\gA\Op}{x}
      &:\equiv
      \Rx{\gA}{x}
    \end{align*}
  \end{definition}

  \begin{observation}\label[observation]{obs:op-rx-gph-involution}
    The opposite reflexive graph operation is definitionally involutive: we have $\prn{\gA\Op}\Op \equiv \gA$ definitionally.
  \end{observation}

  \begin{definition}[Total opposite of a displayed reflexive graph]\label[definition]{def:tot-op-disp-rx-gph}
    Let $\gA$ be a reflexive graph, and let $\gB$ be a displayed reflexive graph over $\gA$. We define the \DefEmph{total opposite} $\gB\TotOp$ of $\gB$ to be the following displayed reflexive graph over $\gA\Op$,
    \begin{align*}
      \vrt{\gB\TotOp}\prn{x}
      &:\equiv
      \vrt{\gB}\prn{x}
      \\
      u \Edge{\gB\TotOp}[p] v
      &:\equiv
      v\Edge{\gB}[p]u
      \\
      \DRx{\gB\TotOp}{x}{u}
      &:\equiv
      \DRx{\gB}{x}{u}\text{,}
    \end{align*}
    so that we have $(\gA.\gB)\Op \equiv \gA\Op.\gB\TotOp$.
  \end{definition}

  Note that \cref{def:tot-op-disp-rx-gph} does not define the actual ``opposite'' of a displayed reflexive graph $\gB$, which would naturally have the same base as $\gB$; opposites of arbitrary displayed reflexive graphs do not make sense (for the same reason that B\'enabou's definition of opposites applies only to displayed categories that are additionally fibrations).

  \cref{obs:op-rx-gph-involution} extends to the following duality involution on displayed reflexive graphs vs \cref{def:tot-op-disp-rx-gph}.

  \begin{observation}[Duality involution for displayed reflexive graphs]\label[observation]{def:disp-rx-gph-duality}
    The operation sending a displayed reflexive graph to its total opposite is definitionally involutive, \ie we have $\prn{\gB\TotOp}\TotOp \equiv \gB$.
  \end{observation}

\end{xsect}
\begin{xsect}{Path objects and the univalence condition}
  \begin{definition}[Fans of a vertex]\label[definition]{def:fan}
    Let $\gA$ be a reflexive graph. The \DefEmph{fan} of a vertex $x:\vrt{\gA}$ is defined to be the type $\Fan{\gA}{x} :\equiv \Sum{y:A}x\Edge{\gA}y$ of vertices equipped with an edge from $x$; dually, the \DefEmph{co-fan} of a vertex $x:\vrt{\gA}$ is defined to be the type $\CoFan{\gA}{x} :\equiv \Sum{y:A}y\Edge{\gA}x$ of vertices equipped with an edge \emph{toward} $x$.
  \end{definition}

  \begin{lemma}\label[lemma]{lem:inward-vs-outward-fans-propositional}
    Let $\gA$ be a reflexive graph; then every fan of $\gA$ is a proposition if and only if every co-fan of $\gA$ is a proposition.
  \end{lemma}
  \begin{proof}
    First we assume that every fan is contractible:

    \iblock{
      \mhang{
        x:\vrt{\gA}; u_0,u_1:\CoFan{\gA}{x} \vdash \Id{\CoFan{\gA}{x}}{u_0}{u_1}
      }{
        \commentrow{by reassociating}
        \mrow{
          x,y_0,y_1:\vrt{\gA}; p_0: y_0\Edge{\gA}x, p_1:y_1\Edge{\gA}x
          \vdash \Id{\CoFan{\gA}{x}}{\prn{y_0,p_0}}{\prn{y_1,p_1}}
        }
        \commentrow{$\Fan{\gA}{y_0}$ is contractible by assumption}
        \mrow{
          y_0,y_1:\vrt{\gA}; p_1:y_1\Edge{\gA}y_0
          \vdash \Id{\CoFan{\gA}{y_0}}{\prn{y_0,\Rx{\gA}{y_0}}}{\prn{y_1,p_1}}
        }
        \commentrow{$\Fan{\gA}{y_1}$ is contractible by assumption}
        \mrow{
          y_1:\vrt{\gA}
          \vdash \Id{\CoFan{\gA}{y_1}}{\prn{y_1,\Rx{\gA}{y_1}}}{\prn{y_1,\Rx{\gA}{y_1}}}
        }
        \commentrow{by reflexivity}
      }
    }

    Conversely, assume that every co-fan is contractible.

    \iblock{
      \mhang{
        x:\vrt{\gA}; u_0,u_1:\Fan{\gA}{x} \vdash \Id{\Fan{\gA}{x}}{u_0}{u_1}
      }{
        \commentrow{by reassociating}
        \mrow{
          x,y_0,y_1:\vrt{\gA}; p_0:x\Edge{\gA}y_0, p_1:x\Edge{\gA}y_1 \vdash \Id{\Fan{\gA}{x}}{\prn{y_0,p_0}}{\prn{y_1,p_1}}
        }
        \commentrow{$\CoFan{\gA}{y_0}$ is contractible by assumption}
        \mrow{
          y_0,y_1:\vrt{\gA};p_1:y_0\Edge{\gA}y_1 \vdash \Id{\Fan{\gA}{y_0}}{\prn{y_0,\Rx{\gA}{y_0}}}{\prn{y_1,p_1}}
        }
        \commentrow{$\CoFan{\gA}{y_1}$ is contractible by assumption}
        \mrow{
          y_1:\vrt{\gA} \vdash \Id{\Fan{\gA}{y_1}}{\prn{y_1,\Rx{\gA}{y_1}}}{\prn{y_1,\Rx{\gA}{y_1}}}
        }
        \commentrow{by reflexivity}
        \qedhere
      }
    }
  \end{proof}

  \begin{construction}[From identifications to edges]\label[construction]{con:id-to-edge}
    The reflexivity datum of a reflexive graph $\gA$ induces a function from identifications to edges as follows:

    \iblock{
      \mhang{
        x,y:\vrt{\gA};p:\Id{\vrt{\gA}}{x}{y}
        \vdash
        \IdToEdge{\gA}[x,y]{p}
        : x\Edge{\gA}y
      }{
        \commentrow{by identification induction}
        \mrow{
          x:\vrt{\gA} \vdash
          \IdToEdge{\gA}[x,x]{\Refl}
          : x\Edge{\gA}x
        }
        \commentrow{by reflexivity}
        \mrow{
          x:\vrt{\gA}
          \vdash
          \IdToEdge{\gA}[x,x]{\Refl}
          :\equiv
          \Rx{\gA}{x}
        }
      }
    }

    When it causes no ambiguity, we will write $\IdToEdge{\gA}{p}$ for $\IdToEdge{\gA}[x,y]{p}$.
  \end{construction}

  \begin{lemma}[From edges to identifications via propositional fans]\label[lemma]{lem:edge-to-id}
    Suppose that each fan $\Fan{\gA}{x}$ of a reflexive graph $\gA$ is a proposition. Then each $\IdToEdge{\gA}[x,y]{-}$ has a quasi-inverse $\EdgeToId{\gA}[x,y]{-} \colon x\Edge{\gA}y\to \Id{\vrt{\gA}}{x}{y}$ and is therefore an equivalence.
  \end{lemma}
  \begin{proof}
    Let $\Phi_x$ be the proof that a given fan $\Fan{\gA}{x}$ is a proposition. Rather than defining $\EdgeToId{\gA}[x,y]{-} : x\Edge{\gA}y\to \Id{\vrt{\gA}}{x}{y}$ directly, we construct a slightly more general function that anticipates the contractibility of $\Fan{\gA}{x}$.

    \iblock{
      \mhang{
        x,y:\vrt{\gA}; p : x\Edge{\gA}y; \phi : \Id{\Fan{\gA}{x}}{\prn{x,\Rx{\gA}{x}}}{\prn{y,p}}
        \vdash
        \EdgeToIdGen{\gA}[x,y]{p}{\phi}
        :
        \Id{\vrt{\gA}}{x}{y}
      }{
        \commentrow{by identification induction}
        \mrow{
          x:\vrt{\gA}\vdash
          \EdgeToIdGen{\gA}[x,y]{\Rx{\gA}{x}}{\Refl}
          : \Id{\vrt{\gA}}{x}{x}
        }
        \commentrow{by reflexivity}
        \mrow{
          x:\vrt{\gA}\vdash
          \EdgeToIdGen{\gA}[x,y]{\Rx{\gA}{x}}{\Refl}
          :\equiv
          \Refl
        }
      }
    }

    Then we define $\EdgeToId{\gA}[x,y]{-}$ by instantiation.

    \iblock{
      \mrow{
        x,y:\vrt{\gA}; p:x\Edge{\gA} y \vdash \EdgeToId{\gA}[x,y]{p} :\equiv
        \EdgeToIdGen{\gA}[x,y]{p}{
          \Phi_x\,\prn{x,\Rx{\gA}{x}}\,\prn{y,p}
        }
      }
    }

    Next we prove that $\EdgeToId{\gA}[x,y]{-}$ is a section of $\IdToEdge{\gA}[x,y]{-}$.

    \iblock{
      \mhang{
        x,y:\vrt{\gA}; p : x\Edge{\gA} y \vdash
        \Id{x\Edge{\gA}y}{\IdToEdge{\gA}[x,y]{\EdgeToId{\gA}[x,y]{p}}}{p}
      }{
        \commentrow{by definition}
        \mrow{
          \ldots\vdash
          \Id{x\Edge{\gA}y}{
            \IdToEdge{\gA}[x,y]{
              \EdgeToIdGen{\gA}[x,y]{p}{
                \Phi_x\,\prn{x,\Rx{\gA}{x}}\,\prn{y,p}
              }
            }
          }{p}
        }
        \commentrow{by generalising over $\Phi_x\,\prn{x,\Rx{\gA}{x}}\,\prn{y,p}$}
        \mrow{
          x,y:\vrt{\gA}; p : x\Edge{\gA} y, \phi:\Id{\Fan{\gA}{x}}{\prn{x,\Rx{\gA}{x}}}{\prn{y,p}} \vdash
          \Id{x\Edge{\gA}y}{
            \IdToEdge{\gA}[x,y]{
              \EdgeToIdGen{\gA}[x,y]{p}{\phi}
            }
          }{p}
        }
        \commentrow{by identification induction}
        \mrow{
          x:\vrt{\gA} \vdash
          \Id{x\Edge{\gA}x}{
            \IdToEdge{\gA}[x,x]{
              \EdgeToIdGen{\gA}[x,x]{\Rx{\gA}{x}}{\Refl}
            }
          }{\Rx{\gA}{x}}
        }
        \commentrow{by definition of $\EdgeToIdGen{\gA}[x,x]{-}{-}$}
        \mrow{
          x:\vrt{\gA} \vdash
          \Id{x\Edge{\gA}x}{
            \IdToEdge{\gA}[x,x]{
              \Refl
            }
          }{\Rx{\gA}{x}}
        }
        \commentrow{by definition of $\IdToEdge{\gA}[x,x]{-}$}
        \mrow{
          x:\vrt{\gA} \vdash
          \Id{x\Edge{\gA}x}{
            \Rx{\gA}{x}
          }{\Rx{\gA}{x}}
        }
        \commentrow{by reflexivity}
      }
    }

    Finally, we prove that $\EdgeToId{\gA}[x,y]{-}$ is a retraction of $\IdToEdge{\gA}[x,y]{-}$.

    \iblock{
      \mhang{
        x,y:\vrt{\gA}; p : \Id{\vrt{\gA}}{x}{y}
        \vdash \Id{\Id{\vrt{A}}{x}{y}}{\EdgeToId{\gA}[x,y]{\IdToEdge{\gA}[x,y]{p}}}{p}
      }{
        \commentrow{by identification induction}
        \mrow{
          x:\vrt{\gA} \vdash
          \Id{\Id{\vrt{A}}{x}{x}}{\EdgeToId{\gA}[x,x]{\IdToEdge{\gA}[x,x]{\Refl}}}{\Refl}
        }
        \commentrow{by definition}
        \mrow{
          x:\vrt{\gA} \vdash
          \Id{\Id{\vrt{A}}{x}{x}}{
            \EdgeToIdGen{\gA}[x,x]{
              \Rx{\gA}x
            }{
              \Phi_x\,(x,\Rx{\gA}x)\,(x,\Rx{\gA}x)
            }
          }{\Refl}
        }
        \commentrow{because $\Fan{\gA}{x}$ is a set}
        \mrow{
          x:\vrt{\gA} \vdash
          \Id{\Id{\vrt{A}}{x}{x}}{
            \EdgeToIdGen{\gA}[x,x]{
              \Rx{\gA}x
            }{
              \Refl
            }
          }{\Refl}
        }
        \commentrow{by definition}
        \mrow{
          x:\vrt{\gA} \vdash
          \Id{\Id{\vrt{A}}{x}{x}}{\Refl}{\Refl}
        }
        \commentrow{by reflexivity}
        \qedhere
      }
    }
  \end{proof}

  \begin{lemma}\label[lemma]{lem:propositional-fan-vs-id-to-edge-equiv}
    For a reflexive graph $\gA$, each function $\IdToEdge{\gA}[x,y]{-}$ is an equivalence if and only if each fan $\Fan{\gA}{z}$ is a proposition.
  \end{lemma}

  \begin{proof}
    We have seen one direction already in \cref{lem:edge-to-id}. Conversely,
    suppose that each component of $\IdToEdge{\gA}[x,y]{-}$ is an equivalence; then clearly each $\Fan{\gA}{z}\equiv \Sum{y:\vrt{\gA}} z \Edge{\gA}y$ is a retract of $\Sum{y:\vrt{\gA}}\Id{\vrt{\gA}}{x}{y} \equiv \Singleton{\vrt{\gA}}{x}$, which is contractible by the contractibility of singletons. Therefore, each fan is contractible and hence a proposition.
  \end{proof}

  \begin{definition}[Univalent reflexive graph]\label[definition]{def:univalent-reflexive-graph}
    A reflexive graph $\gA$ is called \DefEmph{univalent} when any of the following equivalent conditions hold:
    \begin{enumerate}
      \item Each fan $\Fan{\gA}{x}$ is a proposition.
      \item Each co-fan $\CoFan{\gA}{x}$ is a proposition.
      \item Each fan $\Fan{\gA}{x}$ is contractible with centre $\prn{x,\Rx{\gA}{x}}$.
      \item Each co-fan $\CoFan{\gA}{x}$ is contractible with centre $\prn{x,\Rx{\gA}{x}}$.
      \item Each of the functions $\mathsf{idToEdge}_{x,y} : \Id{\vrt{\gA}}{x}{y} \to x\Edge{\gA}y$ defined in \cref{con:id-to-edge} is an equivalence.
    \end{enumerate}

    We shall refer to a univalent reflexive graph as a \DefEmph{path object}.
  \end{definition}

  \begin{proof}
    The equivalence of the stated conditions follows from \cref{lem:inward-vs-outward-fans-propositional,lem:propositional-fan-vs-id-to-edge-equiv}.
  \end{proof}

  \begin{definition}[Univalent displayed reflexive graph]\label[definition]{def:univalent-displayed-reflexive-graph}
    Let $\gA$ be a reflexive graph. A displayed reflexive graph $\gB$ over $\gA$ is called \DefEmph{univalent} when for each $x:\vrt{\gA}$, the component (\cref{def:component}) $\gB\prn{x}$ is univalent in the sense of \cref{def:univalent-reflexive-graph}.
    We shall refer to a univalent displayed reflexive graph as a \DefEmph{displayed path object}, regardless of whether the base $\gA$ is univalent.
  \end{definition}

\end{xsect}   %
\begin{xsect}{Path algebra in a path object}
  Edges in an arbitrary reflexive graph cannot be concatenated or inverted; when the reflexive graph is \emph{univalent}, \ie a \emph{path object}, the type of edges becomes (by \cref{lem:propositional-fan-vs-id-to-edge-equiv}) equivalent to the first level of the canonical $\infty$-groupoid structure on the type of vertices, and as such brooks fully coherent concatenation and inversion operations.
  We will need these operations only rarely --- and indeed, the central thesis of the theory of path objects in univalent foundations is that it is often possible to avoid painful path algebra by very careful choices of path object structure.

  \begin{construction}[Path algebra toolkit]\label[construction]{con:path-alg-toolkit}
    Given any path object $\gA$, we may construct the following terms facilitating path algebra in $\gA$. 

    \iblock{
      \mrow{
        p : x\Edge{\gA}y, q : y \Edge{\gA}z
        \vdash
        p \ct_\gA q : x \Edge{\gA} z
      }
      \mrow{
        p : x\Edge{\gA}y
        \vdash
        p\Inv_\gA  : y \Edge{\gA}x
      }

      \row

      \mhang{
        p : x\Edge{\gA}y
      }{
        \mrow{
          \vdash \mathsf{runit}_p : \Id{x\Edge{\gA}y}{p \ct_\gA \Rx{\gA}{y}}{p}
        }
        \mrow{
          \vdash \mathsf{lunit}_p : \Id{x\Edge{\gA}y}{\Rx{\gA}{x} \ct_\gA p}{p}
        }
        \mrow{
          \vdash \mathsf{rsym}_p : \Id{x\Edge{\gA}x}{p \ct_\gA p\Inv_\gA}{\Rx{\gA}{x}}
        }
        \mrow{
          \vdash \mathsf{lsym}_p : \Id{y\Edge{\gA}y}{p\Inv_\gA\ct_\gA p}{\Rx{\gA}{y}}
        }
      }

      \row

      \mhang{
        p : u\Edge{\gA} v, q:v\Edge{\gA}w, r : w\Edge{\gA}x
      }{
        \mrow{
          \vdash
          \mathsf{assoc}_{p,q,r} :
          \Id{u\Edge{\gA}x}{\prn{p\ct_\gA q}\ct_\gA r}{p\ct_\gA \prn{q\ct_\gA r}}
        }
      }
    }
  \end{construction}

  Of course, there is no end to the possible combinators, so we have presented just a few representative examples that we will make use of.

  \begin{lemma}[Pre-concatenation equivalence]\label[lemma]{lem:pre-ct-is-equiv}
    For a path object $\gA$ and an edge $p:x\Edge{\gA}y$, we have a family of equivalences induced by pre-concatenation with $p$ as below,
    \[
      z:\vrt{\gA} \vdash \prn{p\ct_\gA -} \colon y \Edge{\gA}z \to x\Edge{\gA}z
    \]
    with chosen inverse $\prn{p\Inv_\gA \ct_\gA -} : x\Edge{\gA}z \to y\Edge{\gA}z$.
  \end{lemma}

  \begin{proof}
    Simple path algebra using $\mathsf{assoc}$, $\mathsf{lunit}$, $\mathsf{lsym}$, and $\mathsf{rsym}$ from \cref{con:path-alg-toolkit}.
  \end{proof}

\end{xsect}   %
\begin{xsect}{Univalent families and reflexive graph images}
  \begin{definition}[Reflexive graph image]\label[definition]{def:rx-gph-img}
    The \DefEmph{reflexive graph image} of a type $A$ under a family of types $x:A\vdash B\prn{x}:\TYPE$ is defined to be the following reflexive graph $A/B$:
    \begin{align*}
      \vrt{A/B}
       & :\equiv A
      \\
      x \Edge{A/B} y
       & :\equiv \mathsf{Equiv}\prn{B\prn{x},B\prn{y}}
      \\
      \Rx{A/B}{x}
       & :\equiv \mathsf{idnEquiv}\Sub{B\prn{x}}
    \end{align*}
  \end{definition}

  \begin{definition}[Univalent family of types]\label[definition]{def:univalent-family}
    Let $A$ be a type, and let $B$ be a family of types indexed in $A$. The pair $\prn{A,B}$ is said to be \DefEmph{univalent} when any of the following equivalent conditions hold:
    \begin{enumerate}
      \item For each $x:A$, the type $\Sum{y:A}\mathsf{Equiv}\prn{B\prn{x},B\prn{y}}$ is a proposition.
      \item For each $x:A$, the type $\Sum{y:A}\mathsf{Equiv}\prn{B\prn{x},B\prn{y}}$ is a contractible.
      \item For each $x,y:A$ the canonical map $\Id{A}{x}{y}\to \mathsf{Equiv}\prn{B\prn{x},B\prn{y}}$ sending $\Refl$ to the identity equivalence is an equivalence.
      \item The reflexive graph image $A/B$ is univalent.
    \end{enumerate}
  \end{definition}
  \begin{proof}
    For (1,2), inhabited types are propositions if and only if they are contractible. The remainder is the ``fundamental theorem of identity types''~\citep{rijke:2025}.
  \end{proof}

  One often refers to a univalent family of types $(U,E)$ as a \emph{universe} to emphasise that it may be closed under various type constructors, \etc.

  \begin{definition}[Propositional reflexive graph image]\label[definition]{def:prop-rx-gph-img}
    Let $B\prn{x}$ be a proposition for each $x:A$. The \DefEmph{propositional reflexive graph image} of $A$ under $B$ is defined to be the following reflexive graph $A/_{-1}B$:
    \begin{align*}
      \vrt{A/_{-1}B}
       & :\equiv A
      \\
      x \Edge{A/_{-1}B} y
       & :\equiv
      \prn{B\prn{x}\to B\prn{y}}\times\prn{B\prn{y}\to B\prn{x}}
      \\
      \Rx{A/_{-1}B}{x}
       & :\equiv
      \prn{\Lam{u}{u},\Lam{u}{u}}
    \end{align*}
  \end{definition}

  \begin{observation}
    Let $B$ be a family of propositions indexed in $A$. The pair $\prn{A,B}$ is univalent if and only if the propositional reflexive graph image $A/_{-1}B$ is univalent.
  \end{observation}

  \begin{definition}[Univalent family of path objects]
    A \DefEmph{univalent family of path objects} is defined to be a pair $\prn{U,\gE}$ where $U$ is a type and $A:U\vdash \gE\prn{A}$ is a family of path objects in $U$ such that the reflexive graph image $U/\vrt{\mathcal{E}}$ is univalent.
  \end{definition}

\end{xsect}
\begin{xsect}{Basic closure properties of path objects}
  \begin{lemma}[Opposite path object]\label[lemma]{cor:op-po}
    A reflexive graph $\gA$ is univalent if and only if its opposite $\gA\Op$ is univalent.
  \end{lemma}
  \begin{proof}
    This is an immediate consequence of \cref{lem:inward-vs-outward-fans-propositional}.
  \end{proof}

  \begin{lemma}[Total path object]\label[lemma]{lem:total-po}
    If $\gA$ is a path object and $\gB$ is a displayed path object over $\gA$, then $\gA.\gB$ is a path object.
  \end{lemma}

  \begin{proof}
    Univalence is established as follows.

    \iblock{
      \mhang{
        u:\vrt{\gA.\gB};
        s_0,s_1 : \Fan{\gA.\gB}{u}
        \vdash
        \Id{\Fan{\gA.\gB}{u}}{s_0}{s_1}
      }{
        \commentrow{by reassociating and unfolding hypotheses}
        \mrow{
          x,x_0,x_1:\vrt{\gA};
          y:\vrt{\gB\prn{x}},
          y_0:\vrt{\gB\prn{x_0}},
          y_1:\vrt{\gB\prn{x_1}},
        }
        \mrow{
          p_0:x\Edge{\gA}x_0,
          p_1:x\Edge{\gA}x_1,
          q_0:y\Edge{\gB}[p_0]y_0,
          q_1:y\Edge{\gB}[p_1]y_1
        }
        \iblock{
          \mrow{
            \vdash
            \Id{\Fan{\gA.\gB}{\prn{x,y}}}{
              \prn{\prn{x_0,y_0},\prn{p_0,q_0}}
            }{
              \prn{\prn{x_1,y_1},\prn{p_1,q_1}}
            }
          }
        }

        \commentrow{as $\gA$ is a path object and thus $\Fan{\gA}{x}$ is contractible}

        \mrow{
          x:\vrt{\gA},
          y:\vrt{\gB\prn{x}},
          y_0:\vrt{\gB\prn{x_0}},
          y_1:\vrt{\gB\prn{x_1}},
        }
        \mrow{
          q_0:y\Edge{\gB\prn{x}}y_0,
          q_1:y\Edge{\gB\prn{x}}y_1
        }
        \iblock{
          \mrow{
            \vdash
            \Id{\Fan{\gA.\gB}{\prn{x,y}}}{
              \prn{\prn{x,y_0},\prn{\Rx{\gA}{x},q_0}}
            }{
              \prn{\prn{x,y_1},\prn{\Rx{\gA}{x},q_1}}
            }
          }
        }

        \commentrow{as $\gB\prn{x}$ is a path object and thus $\Fan{\gB\prn{x}}{y}$ is contractible}

        \mrow{
          x:\vrt{\gA},
          y:\vrt{\gB\prn{x}}
          \vdash
          \Id{\Fan{\gA.\gB}{\prn{x,y}}}{
            \prn{\prn{x,y},\prn{\Rx{\gA}{x},\Rx{\gB\prn{x}}{y}}}
          }{
            \prn{\prn{x,y},\prn{\Rx{\gA}{x},\Rx{\gB\prn{x}}{y}}}
          }
        }

        \commentrow{by reflexivity}
        \qedhere
      }
    }
  \end{proof}

  \begin{observation}[Constant displayed path object]\label[observation]{lem:const-disp-po}
    Let $\gA$ be a reflexive graph and let $\gB$ be a path object. Then the constant displayed reflexive graph $\gA^*\gB$ over $\gA$ is univalent.
  \end{observation}

  \begin{proof}
    This follows immediately, considering that we have $\prn{\gA^*\gB}\prn{x}\equiv \gB$ for each $x:\vrt{\gA}$ by definition.
  \end{proof}

  \begin{corollary}[Binary product of path objects]\label[corollary]{lem:bin-prod-po}
    If $\gA$ and $\gB$ are path objects, then so is $\gA\times \gB$.
  \end{corollary}

  \begin{proof}
    By \cref{obs:rx-gph-bin-prod-vs-sum,lem:const-disp-po}.
  \end{proof}

  \begin{definition}[Function extensionality]\label[definition]{def:prod-po}
    \DefEmph{Dependent function extensionality} is precisely the property that for each type $A$ and family $\gB\prn{x}$ of path objects indexed in $x:A$, the product $\Prod{x:A}\gB\prn{x}$ is a path object. Likewise, \DefEmph{non-dependent function extensionality} states that for each type $A$ and path object $\gB$, the cotensor $A\pitchfork \gB$ is a path object.
  \end{definition}

  \begin{lemma}[Coproduct of path objects]\label[lemma]{lem:coprod-po}
    For any type $A$ and family $\gB\prn{x}$ of path objects indexed in $x:A$, the coproduct $\Coprod{x:A}\gB\prn{x}$ is a path object.
  \end{lemma}

  \begin{proof}
    We establish the univalence condition as follows.

    \iblock{
      \mhang{
        u:\Sum{x:A}\vrt{\gB\prn{x}};
        s_0,s_1:\Fan{\Coprod{x:A}\gB\prn{x}}{u}
        \vdash
        \Id{\Fan{\Coprod{x:A}\gB\prn{x}}{u}}{s_0}{s_1}
      }{
        \commentrow{by reassociating and unfolding hypotheses}

        \mrow{
          x,x_0,x_1:A;
          y:\vrt{\gB\prn{x}},
          y_0:\vrt{\gB\prn{x_0}},
          y_1:\vrt{\gB\prn{x_1}},
          p_0 : \Id{A}{x}{x_0},
          p_1 : \Id{A}{x}{x_1},
        }
        \mrow{
        q_0 : \prn{p_0}_*^{\vrt{\gB\prn{-}}}y \Edge{\gB\prn{x_0}} y_0,
        q_1 : \prn{p_1}_*{\vrt{\gB\prn{-}}}y\Edge{\gB\prn{x_1}} y_1
        }
        \iblock{
          \mrow{
            \vdash
            \Id{
              \Fan{\Coprod{x:A}\gB\prn{x}}{\prn{x,y}}
            }{
              \prn{\prn{x_0,y_0},\prn{p_0,q_0}}
            }{
              \prn{\prn{x_1,y_1},\prn{p_1,q_1}}
            }
          }
        }

        \commentrow{by based identification induction on $\prn{x_0,p_0}$}

        \mrow{
          x,x_1:A;
          y:\vrt{\gB\prn{x}},
          y_0:\vrt{\gB\prn{x}},
          y_1:\vrt{\gB\prn{x_1}},
          p_1 : \Id{A}{x}{x_1},
        }
        \mrow{
        q_0 : y \Edge{\gB\prn{x}} y_0,
        q_1 : \prn{p_1}_*{\vrt{\gB\prn{-}}}y\Edge{\gB\prn{x_1}} y_1
        }
        \iblock{
          \mrow{
            \vdash
            \Id{
              \Fan{\Coprod{x:A}\gB\prn{x}}{\prn{x,y}}
            }{
              \prn{
                \prn{x,y_0},
                \prn{\Refl,q_0}
              }
            }{
              \prn{
                \prn{x_1,y_1},
                \prn{p_1,q_1}
              }
            }
          }
        }

        \commentrow{by based identification induction on $\prn{x_1,p_1}$}

        \mrow{
          x:A;
          y:\vrt{\gB\prn{x}},
          y_0:\vrt{\gB\prn{x}},
          y_1:\vrt{\gB\prn{x}},
          q_0 : y \Edge{\gB\prn{x}} y_0,
          q_1 : y\Edge{\gB\prn{x}} y_1
        }
        \iblock{
          \mrow{
            \vdash
            \Id{
              \Fan{\Coprod{x:A}\gB\prn{x}}{\prn{x,y}}
            }{
              \prn{
                \prn{x,y_0},
                \prn{\Refl,q_0}
              }
            }{
              \prn{
                \prn{x,y_1},
                \prn{\Refl,q_1}
              }
            }
          }
        }

        \commentrow{as $\gB\prn{x}$ is a path object and thus $\Fan{\gB\prn{x}}{y}$ is contractible}

        \mrow{
          x:A;
          y:\vrt{\gB\prn{x}}
          \vdash
          \Id{
            \ldots
          }{
            \prn{\prn{x,y},\prn{\Refl,\Rx{\gB\prn{x}}{y}}}
          }{
            \prn{\prn{x,y},\prn{\Refl,\Rx{\gB\prn{x}}{y}}}
          }
        }

        \commentrow{by reflexivity}
        \qedhere
      }
    }
  \end{proof}

  \begin{corollary}[Tensor of path objects]\label[corollary]{lem:tensor-po}
    The tensor $A\cdot\gB$ of a path object $\gB$ by a type $A$ is a path object.
  \end{corollary}

  \begin{corollary}[Discrete path object]\label[corollary]{lem:disc-po}
    The discrete reflexive graph $\DiscPO{A}$ on any type $A$ is univalent.
  \end{corollary}

  \begin{proof}
    This is precisely the contractibility of singletons.
  \end{proof}

  \begin{lemma}[Codiscrete path object]\label[lemma]{lem:codisc-po}
    The codiscrete reflexive graph $\CodiscPO{A}$ on a type $A$ is univalent if and only if $A$ is a proposition.
  \end{lemma}

  \begin{proof}
    Suppose that $\gA$ is a proposition. Then univalence is established like so:

    \iblock{
      \mhang{
        x:A; s_0,s_1 : \Fan{\CodiscPO{A}}{x}\vdash \Id{\Fan{\CodiscPO{A}}{x}}{s_0}{s_1}
      }{
        \commentrow{by reassociating and unfolding hypotheses}
        \mrow{
          x,x_0,x_1:A \vdash \Id{\Fan{\CodiscPO{A}}{x}}{\prn{x_0,*}}{\prn{x_1,*}}
        }
        \commentrow{as $A$ is a proposition}
        \mrow{
          x:A\vdash \Id{\Fan{\CodiscPO{A}}{x}}{\prn{x,*}}{\prn{x,*}}
        }
        \commentrow{by reflexivity}
        \qedhere
      }
    }

    Conversely, assume that $\CodiscPO{A}$ is univalent. For any two $x,y:A$,
    we evidently have $*:x\Edge{\CodiscPO{A}}y$ and thus $\EdgeToId{\CodiscPO{A}}[x,y]{*} : \Id{A}{x}{y}$ by \cref{lem:edge-to-id}.
  \end{proof}

  \begin{lemma}[Path subobject comprehension]\label[lemma]{lem:compr-po}
    If $\gA$ is a path object and $P\prn{x}$ is a proposition for each vertex $x:\vrt{\gA}$, then the comprehension $\Compr{x:\gA}{P\prn{x}}$ is a path object.
  \end{lemma}

  \begin{proof}
    Univalence is established as follows:

    \iblock{
      \mhang{
        u:\Sum{x:\vrt{\gA}}P\prn{x};
        s_0,s_1:\Fan{\Compr{x:\gA}{P\prn{x}}}{u}
        \vdash
        \Id{\Fan{\Compr{x:\gA}{P\prn{x}}}{u}}{s_0}{s_1}
      }{
        \commentrow{by reassociating and unfolding hypotheses}

        \mhang{
          x,x_0,x_1:\vrt{\gA};
          y:P\prn{x},y_0:P\prn{x_0},y_1:P\prn{x_1},
          p_0 : x\Edge{\gA}x_0,
          p_1 : x\Edge{\gA}x_1
        }{
          \mrow{
            \vdash
            \Id{
              \Fan{\Compr{x:\gA}{P\prn{x}}}{\prn{x,y}}
            }{
              \prn{\prn{x_0,y_0},\prn{p_0,*}}
            }{
              \prn{\prn{x_1,y_1},\prn{p_1,*}}
            }
          }
        }

        \commentrow{as $\gA$ is a path object and so $\Fan{\gA}{x}$ is contractible}

        \mrow{
          x:\vrt{\gA};
          y,y_0,y_1:P\prn{x}
          \vdash
          \Id{
            \ldots
          }{
            \prn{\prn{x,y_0},\prn{\Rx{\gA}{x},*}}
          }{
            \prn{\prn{x,y_1},\prn{\Rx{\gA}{x},*}}
          }
        }

        \commentrow{as $P\prn{x}$ is a proposition}

        \mrow{
          x:\vrt{\gA},
          y:P\prn{x}
          \vdash
          \Id{
            \ldots
          }{
            \prn{\prn{x,y},\prn{\Rx{\gA}{x},*}}
          }{
            \prn{\prn{x,y},\prn{\Rx{\gA}{x},*}}
          }
        }

        \commentrow{by reflexivity}
        \qedhere
      }
    }
  \end{proof}

\NewDocumentCommand\IsUnivalent{}{\ \text{univalent}}
\NewDocumentCommand\IsProp{}{\ \text{proposition}}
\NewDocumentCommand\IsDispUnivalent{m}{\ \text{univalent over}\ #1}

\begin{figure}
  \begin{mathpar}
    \inferrule[\cref{cor:op-po}]{
      \gA\IsUnivalent
    }{
      \gA\Op\IsUnivalent
    }
    \and
    \inferrule[\cref{lem:total-po}]{
      \gA\IsUnivalent\\
      \gB\IsDispUnivalent{\gA}
    }{
      \gA.\gB\IsUnivalent
    }
    \and
    \inferrule[\cref{lem:bin-prod-po}]{
      \gA\IsUnivalent\\
      \gB\IsUnivalent
    }{
      {\gA\times \gB}\IsUnivalent
    }
    \and
    \inferrule[\cref{def:prod-po}]{
      \text{dep.\ funext.}
      \\
      x:A\vdash \gB\prn{x} \IsUnivalent
    }{
      \Prod{x:A}\gB\prn{x} \IsUnivalent
    }
    \and
    \inferrule[\cref{lem:coprod-po}]{
      x:A\vdash \gB\prn{x}\IsUnivalent
    }{
      \Coprod{x:A}\gB\prn{x}\IsUnivalent
    }
    \and
    \inferrule[\cref{lem:tensor-po}]{
      \gB\IsUnivalent
    }{
      A\cdot\gB\IsUnivalent
    }
    \and
    \inferrule[\cref{def:prod-po}]{
      \text{funext.}\\ 
      \gB\IsUnivalent
    }{
      A\pitchfork\gB\IsUnivalent
    }
    \and
    \inferrule[\cref{lem:disc-po}]{
    }{
      \DiscPO{A}\IsUnivalent
    }
    \and
    \inferrule[\cref{lem:codisc-po}]{
      A\IsProp
    }{
      \CodiscPO{A}\IsUnivalent
    }
    \and
    \inferrule[\cref{lem:compr-po}]{
      \gA\IsUnivalent\\ 
      x:\vrt{\gA}\vdash B\prn{x}\IsProp
    }{
      \Compr{x:\gA}{B\prn{x}}\IsUnivalent
    }
    \and 
    \inferrule[\cref{lem:const-disp-po}]{
      \gB\IsUnivalent
    }{
      \gA^*\gB\IsDispUnivalent{\gA}
    }
  \end{mathpar}
  \caption{Summary of univalence lemmas for (displayed) reflexive graphs.}
\end{figure}

\begin{figure}
  \begin{mathpar}
    \inferrule[\cref{lem:cov-disp-po}]{
      x:\vrt{\gA} \vdash \gB\prn{x}\IsUnivalent
    }{
      \CovDisp{\gA}{\gA} \IsDispUnivalent{\gA}
    }
    \and
    \inferrule[\cref{lem:ctrv-disp-po}]{
      x:\vrt{\gA} \vdash \gB\prn{x}\IsUnivalent
    }{
      \CtrvDisp{\gA}{\gA} \IsDispUnivalent{\gA}
    }
    \and
    \inferrule[\cref{lem:unb-disp-po}]{
      x,y:\vrt{\gA};p:x\Edge{\gA}y \vdash \gB\prn{p}\IsUnivalent
    }{
      \UnbDisp{\gA}{\gA} \IsDispUnivalent{\gA}
    }
  \end{mathpar}
  \caption{Summary of univalence lemmas for reflexive graph lenses.}
\end{figure}
 \end{xsect}
 \end{xsect}
\begin{xsect}[sec:lenses]{Lenses of reflexive graphs and path objects}

  In this section, we are concerned with various kinds of \emph{families} of ordinary reflexive graphs, and how they give rise to certain naturally occurring \emph{displayed} reflexive graphs.

\begin{xsect}{Lenses of reflexive graphs}

  We introduce two dual kinds of lens of reflexive graphs.

  \begin{definition}[Oplax covariant lens]\label[definition]{def:oplax-cov-lens}
    Let $\gA$ be a reflexive graph. An \DefEmph{oplax covariant lens} of reflexive graphs over $\gA$ is defined to be a family of reflexive graphs $\gB\prn{x}$ indexed in $x:\vrt{\gA}$ equipped with the following data:
    \begin{align*}
      x,y:\vrt{\gA}; p : x \Edge{\gA} y, u : \vrt{\gB\prn{x}}
       & \vdash
      \Push{\gB}{p}u : \vrt{\gB\prn{y}}
      \\
      x:\vrt{\gA}, u:\vrt{\gB\prn{x}}
       & \vdash
      \PushRx{\gB}[x]u :
      \Push{\gB}{\Rx{\gA}{x}}{u}
      \Edge{\gB\prn{x}} u
    \end{align*}
  \end{definition}

  \begin{definition}[Lax contravariant lens]\label[definition]{def:lax-ctrv-lens}
    Let $\gA$ be a reflexive graph. A \DefEmph{lax contravariant lens} of reflexive graphs over $\gA$ is defined to be a family of reflexive graphs $\gB\prn{x}$ indexed in $x:\vrt{\gA}$ equipped with the following data:
    \begin{align*}
      x,y:\vrt{\gA};p:x\Edge{\gA}y, u:\vrt{\gB\prn{y}}
       & \vdash
      \Pull{\gB}{p}u : \vrt{\gB\prn{x}}
      \\
      x:\vrt{\gA},u:\vrt{\gB\prn{x}}
       & \vdash
      \PullRx{\gB}[x]u:
      u \Edge{\gB\prn{x}} \Pull{\gB}{\Rx{\gA}{x}}{u}
    \end{align*}
  \end{definition}

  In the context of \cref{def:oplax-cov-lens} or \cref{def:lax-ctrv-lens}, we will refer to $\gB$ as a \DefEmph{lens of path objects} or a \DefEmph{univalent lens} when each $\gB\prn{x}$ is univalent.

  \begin{definition}[Display of oplax covariant lenses]\label[definition]{def:cov-disp}
    Let $\gB$ be an oplax covariant lens over a reflexive graph $\gA$. The (covariant) \DefEmph{display} of $\gB$ over $\gA$ is defined to be the following displayed reflexive graph $\CovDisp{\gA}\gB$ over $\gA$:
    \begin{align*}
      x:\vrt{\gA}
       & \vdash
      \vrt{\CovDisp{\gA}{\gB}}\prn{x}
      :\equiv
      \vrt{\gB\prn{x}}
      \\
      p:x\Edge{\gA}y, u:\vrt{\gB\prn{x}},v:\vrt{\gB\prn{y}}
       & \vdash
      u \Edge{\CovDisp{\gA}{\gB}}[p] v
      :\equiv
      \Push{\gB}{p}u \Edge{\gB\prn{y}} v
      \\
      x:\vrt{\gA},u:\vrt{\gB\prn{x}}
       & \vdash
      \DRx{\CovDisp{\gA}{\gB}}{x}{u} :\equiv
      \PushRx{\gB}[x]{u}
    \end{align*}
  \end{definition}

  \begin{definition}[Display of lax contravariant lenses]\label[definition]{def:ctrv-disp}
    Let $\gB$ be a lax contravariant lens over a reflexive graph $\gA$. The (contravariant) \DefEmph{display} of $\gB$ over $\gA$ is defined to be the following displayed reflexive graph $\CtrvDisp{\gA}{\gB}$ over $\gA$:
    \begin{align*}
      x:\vrt{\gA}
       & \vdash
      \vrt{\CtrvDisp{\gA}{\gB}}\prn{x}
      :\equiv
      \vrt{\gB\prn{x}}
      \\
      p:x\Edge{\gA}y, u:\vrt{\gB\prn{x}},v:\vrt{\gB\prn{y}}
       & \vdash
      u \Edge{\CtrvDisp{\gA}{\gB}}[p] v
      :\equiv
      u\Edge{\gB\prn{x}}\Pull{\gB}{p}{v}
      \\
      x:\vrt{\gA},u:\vrt{\gB\prn{x}}
       & \vdash
      \DRx{\CtrvDisp{\gA}{\gB}}{x}{u} :\equiv
      \PullRx{\gB}[x]{u}
    \end{align*}
  \end{definition}

  \begin{computation}[Components of the display of oplax covariant lenses]\label[computation]{cmp:component-of-oplax-cov-display}
    Let $\gB$ be an oplax covariant lens over a reflexive graph $\gA$. For any $x:\vrt{\gA}$, the component $\prn{\CovDisp{\gA}{\gB}}\prn{x}$ is precisely the following:
    \begin{align*}
      \vrt{\prn{\CovDisp{\gA}{\gB}}\prn{x}}
       & \equiv
      \vrt{\gB\prn{x}}
      \\
      u\Edge{\prn{\CovDisp{\gA}{\gB}}\prn{x}} v
       & \equiv
      \Push{\gB}{\Rx{\gA}{x}}{u} \Edge{\gB\prn{x}} v
      \\
      \Rx{\prn{\CovDisp{\gA}{\gB}}\prn{x}}{u}
       & \equiv
      \PushRx{\gB}[x]{u}
    \end{align*}

    We likewise have the following description of fans:
    \[
      \Fan{\prn{\CovDisp{\gA}{\gB}}\prn{x}}{u}
      \equiv \Fan{\gB\prn{x}}{\Push{\gB}{\Rx{\gA}{x}}{u}}
    \]
  \end{computation}

  \begin{computation}[Components of the display of lax contravariant lenses]\label[computation]{cmp:component-of-lax-ctrv-display}
    Let $\gB$ be a lax contravariant lens over a reflexive graph $\gA$. For any $x:\vrt{\gA}$, the component $\prn{\CtrvDisp{\gA}{\gB}}\prn{x}$ is unfolds as follows:
    \begin{align*}
      \vrt{\prn{\CtrvDisp{\gA}{\gB}}\prn{x}}
       & \equiv
      \vrt{\gB\prn{x}}
      \\
      u\Edge{\prn{\CtrvDisp{\gA}{\gB}}\prn{x}} v
       & \equiv
      u \Edge{\gB\prn{x}} \Pull{\gB}{\Rx{\gA}{x}}{v}
      \\
      \Rx{\prn{\CtrvDisp{\gA}{\gB}}\prn{x}}{u}
       & \equiv
      \PullRx{\gB}[x]{u}
    \end{align*}

    We have the following description of co-fans:
    \[
      \CoFan{\prn{\CtrvDisp{\gA}{\gB}}\prn{x}}{u}
      \equiv
      \CoFan{\gB\prn{x}}{\Pull{\gB}{\Rx{\gA}{x}}{u}}
    \]
  \end{computation}

  \begin{lemma}[Display of oplax covariant lenses of path objects]\label[lemma]{lem:cov-disp-po}
    Let $\gB$ be an oplax covariant lens over a reflexive graph $\gA$. If each $\gB\prn{x}$ is univalent, then the display $\CovDisp{\gA}{\gB}$ is univalent.
  \end{lemma}

  \begin{proof}
    We must check that each component $\prn{\CovDisp{\gA}{\gB}}\prn{x}$ is univalent. It suffices to check that for each $u:\vrt{\gB\prn{x}}$, the fan $\Fan{\prn{\CovDisp{\gA}{\gB}}\prn{x}}{u}$ is a proposition. By \cref{cmp:component-of-oplax-cov-display}, we have $\Fan{\prn{\CovDisp{\gA}{\gB}}\prn{x}}{u}\equiv \Fan{\gB\prn{x}}{\Push{\gB}{\Rx{\gA}{x}}{u}}$, which is a proposition by our assumption that $\gB\prn{x}$ is univalent.
  \end{proof}

  \begin{lemma}[Display of lax contravariant lenses of path objects]\label[lemma]{lem:ctrv-disp-po}
    Let $\gB$ be a lax contravariant lens over a reflexive graph $\gA$. If each $\gB\prn{x}$ is univalent, then the display $\CtrvDisp{\gA}{\gB}$ is univalent.
  \end{lemma}

  \begin{proof}
    Analogous to \cref{lem:cov-disp-po}.
  \end{proof}

\end{xsect}

\begin{xsect}{Unbiased dependent lenses}

  In this section, we describe a slightly strange technical device that generalises both oplax covariant and lax contravariant lenses that we name the \emph{unbiased dependent lens}. The reason for introducing these is to characterise the identity types of structures involving mixed variance (including, for example, the type of reflexive graphs!).

  \begin{definition}[Unbiased dependent lens]
    Let $\gA$ be a reflexive graph. A \DefEmph{unbiased dependent lens} over $\gA$ is defined to be a family of reflexive graphs $\gB\prn{p}$ indexed in edges $p:x\Edge{\gA}y$ equipped with the following data:
    \begin{align*}
      x,y:\vrt{\gA};p:x\Edge{\gA}y,u:\vrt{\gB\prn{\Rx{\gA}{x}}}
       & \vdash
      \LJ{\gB}{p}{u} : \vrt{\gB\prn{p}}
      \\
      x,y:\vrt{\gA};p:x\Edge{\gA}y,u:\vrt{\gB\prn{\Rx{\gA}{y}}}
       & \vdash
      \RJ{\gB}{p}{u} : \vrt{\gB\prn{p}}
      \\
      x:\vrt{\gA},u:\vrt{\gB\prn{\Rx{\gA}{x}}}
       & \vdash
      \MidJRx{\gB}[x]{u} : \LJ{\gB}{\Rx{\gA}{x}}{u} \Edge{\gB\prn{\Rx{\gA}{x}}} \RJ{\gB}{\Rx{\gA}{x}}{u}
      \\
      x:\vrt{\gA},u:\vrt{\gB\prn{\Rx{\gA}{x}}}
       & \vdash
      \RJRx{\gB}[x]{u} : u \Edge{\gB\prn{\Rx{\gA}{x}}} \RJ{\gB}{\Rx{\gA}{x}}{u}
    \end{align*}
  \end{definition}

  \begin{definition}[Display of unbiased dependent lenses]
    Let $\gB$ be a unbiased dependent lens over a reflexive graph $\gA$. The \DefEmph{display} of $\gB$ over $\gA$ is defined to be the following displayed reflexive graph $\UnbDisp{\gA}{\gB}$ over $\gA$:
    \begin{align*}
      x:\vrt{\gA}
       & \vdash
      \vrt{\UnbDisp{\gA}{\gB}}\prn{x}
      :\equiv
      \vrt{\gB\prn{\Rx{\gA}{x}}}
      \\
      p : x\Edge{\gA}y, u:\vrt{\gB\prn{x}}, v : \vrt{\gB\prn{y}}
       & \vdash
      u \Edge{\UnbDisp{\gA}{\gB}}[p] v
      :\equiv
      \LJ{\gB}{p}{u}
      \Edge{\gB\prn{p}}
      \RJ{\gB}{p}{u}
      \\
      x:\vrt{\gA}, u:\vrt{\gB\prn{x}}
       & \vdash
      \DRx{\UnbDisp{\gA}{\gB}}{x}{u} :\equiv
      \MidJRx{\gB}[x]{u}
    \end{align*}
  \end{definition}

  \begin{remark}
    Only $\MidJRx{\gB}[x]{u}$ was needed to define the display of unbiased dependent lenses, whereas  the lax unitor $u \Edge{\gB\prn{\Rx{\gA}{x}}} \RJ{\gB}{\Rx{\gA}{x}}{u}$ played no role. The reason for including the latter is that it is required in order for univalence of each $\gB\prn{p}$ to imply univalence for the display $\UnbDisp{\gA}{\gB}$.
    Naturally, univalence would \emph{also} follow if we instead included an oplax unitor $\LJ{\gB}{p}{u} \Edge{\gB\prn{\Rx{\gA}{x}}}u$. What we must not do is include both $\LJ{\gB}{p}{u} \Edge{\gB\prn{\Rx{\gA}{x}}}u$ and $u \Edge{\gB\prn{\Rx{\gA}{x}}} \RJ{\gB}{\Rx{\gA}{x}}{u}$, as we intend the type of \emph{fiberwise univalent} unbiased dependent lens structures on a family of path objects over a path object to be a proposition. This situation is somewhat similar to that of \emph{half-adjoint equivalences} in homotopy type theory being coherent only by virtue of omitting one of the snake identities which can be obtained from the other. Of course, our current situation in the world of reflexive graphs is a bit different as we cannot obtain the lax unitor from the oplax unitor, and vice versa: as these unitors play no computational role, however, we will let it slide.
  \end{remark}

  \begin{lemma}[Display of unbiased dependent lenses of path objects]\label[lemma]{lem:unb-disp-po}
    Let $\gB$ be a unbiased dependent lens over a reflexive graph $\gA$. If each $\gB\prn{x}$ is univalent, then so is the display $\UnbDisp{\gA}{\gB}$.
  \end{lemma}

  The following \cref{con:lens-upgrade} shows that unbiased dependent lenses generalise \emph{both} oplax covariant and lax contravariant lenses up to definitional equality.

  \begin{construction}[From biased to unbiased lenses]\label[construction]{con:lens-upgrade}
    Let $\gA$ be a path object.
    \begin{enumerate}

      \item Any oplax covariant lens $\gB$ over $\gA$ induces a unbiased dependent lens $\CovUpgrade{\gB}$ over $\gA$ such that $\CovDisp{\gA}{\gB}\equiv\UnbDisp{\gA}\prn{\CovUpgrade{\gB}}$ definitionally:
            \begin{align*}
              p:x\Edge{\gA}y
               & \vdash
              \CovUpgrade{\gB}\prn{p} :\equiv \gB\prn{y}
              \\
              p:x\Edge{\gA}y, u:\vrt{\gB\prn{x}}
               & \vdash
              \LJ{\CovUpgrade{\gB}}{p}u
              :\equiv \Push{\gB}{p}{u}
              \\
              p:x\Edge{\gA}y, u:\vrt{\gB\prn{y}}
               & \vdash
              \RJ{\CovUpgrade{\gB}}{p}u
              :\equiv u
              \\
              x:\vrt{\gA},u:\vrt{\gB\prn{x}}
               & \vdash
              \MidJRx{\CovUpgrade{\gB}}[x]{u} :\equiv \PushRx{\gB}[x]{u}
              \\
              x:\vrt{\gA},u:\vrt{\gB\prn{x}}
               & \vdash
              \RJRx{\CovUpgrade{\gB}}[x]{u} :\equiv \Rx{\gB\prn{x}}{u}
            \end{align*}

      \item Any lax contravariant lens $\gB$ over $\gA$ induces a unbiased dependent lens $\CtrvUpgrade{\gB}$ over $\gA$ such that $\CtrvDisp{\gA}{\gB} \equiv \UnbDisp{\gA}\prn{\CtrvUpgrade{\gB}}$ definitionally:
            \begin{align*}
              p:x\Edge{\gA}y
               & \vdash
              \CtrvUpgrade{\gB}\prn{p} :\equiv \gB\prn{x}
              \\
              p:x\Edge{\gA}y, u:\vrt{\gB\prn{x}}
               & \vdash
              \LJ{\CtrvUpgrade{\gB}}{p}u
              :\equiv
              u
              \\
              p:x\Edge{\gA}y,u:\vrt{\gB\prn{y}}
               & \vdash
              \RJ{\CtrvUpgrade{\gB}}{p}{u}
              :\equiv
              \Pull{\gB}{p}{u}
              \\
              x:\vrt{\gA},u:\vrt{\gB\prn{x}}
               & \vdash
              \MidJRx{\CtrvUpgrade{\gB}}[x]{u}
              :\equiv \PullRx{\gB}[x]{u}
              \\
              x:\vrt{\gA},u:\vrt{\gB\prn{x}}
               & \vdash
              \RJRx{\CtrvUpgrade{\gB}}[x]{u}
              :\equiv \PullRx{\gB}[x]{u}
            \end{align*}

    \end{enumerate}
  \end{construction}

  \begin{proof}[Computation]
    We compute the displays as follows.
    \begin{align*}
       &
      \vrt{\UnbDisp{\gA}{\CovUpgrade{\gB}}}\prn{x}
      \equiv
      \vrt{\CovUpgrade{\gB}}\prn{\Rx{\gA}{x}}
      \equiv
      \vrt{\gB\prn{x}}
      \equiv
      \vrt{\CovDisp{\gA}{\gB}}\prn{x}
      \\
       &
      \vrt{\UnbDisp{\gA}{\CtrvUpgrade{\gB}}}\prn{x}
      \equiv
      \vrt{\CtrvUpgrade{\gB}}\prn{\Rx{\gA}{x}}
      \equiv
      \vrt{\gB\prn{x}}
      \equiv
      \vrt{\CtrvDisp{\gA}{\gB}}\prn{x}
      \\
       &
      u \Edge{\UnbDisp{\gA}{\CovUpgrade{\gB}}}[p] v
      \equiv
      \LJ{\CovUpgrade{\gB}}{p}{u} \Edge{\CovUpgrade{\gB}\prn{p}} \RJ{\CovUpgrade{\gB}}{p}{v}
      \equiv
      \Push{\gB}{p}{u} \Edge{\gB\prn{y}} u
      \equiv
      u \Edge{\CovDisp{\gA}{\gB}}[p]v
      \\
       &
      u \Edge{\UnbDisp{\gA}{\CtrvUpgrade{\gB}}}[p] v
      \equiv
      \LJ{\CtrvUpgrade{\gB}}{p}{u} \Edge{\CtrvUpgrade{\gB}\prn{p}} \RJ{\CtrvUpgrade{\gB}}{p}{v}
      \equiv
      u \Edge{\gB\prn{y}} \Pull{\gB}{p}{v}
      \equiv
      u \Edge{\CtrvDisp{\gA}{\gB}}[p]v
      \\
       &
      \DRx{\UnbDisp{\gA}{\CovUpgrade{\gB}}}{x}{u}
      \equiv
      \MidJRx{\CovUpgrade{\gB}}[x]{u}
      \equiv
      \Push{\gB}{p}{u}
      \equiv
      \DRx{\CovDisp{\gA}{\gB}}{x}{u}
      \\
       &
      \DRx{\UnbDisp{\gA}{\CtrvUpgrade{\gB}}}{x}{u}
      \equiv
      \MidJRx{\CtrvUpgrade{\gB}}[x]{u}
      \equiv
      \Pull{\gB}{p}{u}
      \equiv
      \DRx{\CtrvDisp{\gA}{\gB}}{x}{u}
       &   & \qedhere
    \end{align*}
  \end{proof}

\end{xsect}
 
  \begin{xsect}{Duality involution for reflexive graph lenses}

    We have seen in \cref{def:op-rx-gph,obs:op-rx-gph-involution} the duality involution for reflexive graphs that flips the direction of edges. This involution hosts a (definitional) equivalence between the types of oplax covariant and lax contravariant lenses respectively. By virtue of this equivalence, from a result about all oplax covariant lenses one obtains automatically a dual result about all lax contravariant lenses.

    \begin{definition}[Total opposite of a lens]\label[definition]{def:tot-op-lens}
      Let $\gA$ be a reflexive graph, and let $\gB$ be an oplax covariant lens of reflexive graphs over $\gA$. We define the \DefEmph{total opposite} of $\gB$ to be the following lax contravariant lens $\gB\TotOp$ over $\gA\Op$:
      \begin{align*}
        \gB\TotOp\prn{x}
         & :\equiv
        \gB\prn{x}\Op
        \\
        \Pull{\gB\TotOp}{p}{u}
         & :\equiv
        \Push{\gB}{p}{u}
        \\
        \PullRx{\gB\TotOp}[x]{u}
         & :\equiv
        \PushRx{\gB}[x]{u}
      \end{align*}

      Conversely, if $\gB$ is a lax contravariant lens we define its total opposite to be the following lax contravariant lens $\gB\TotOp$ over $\gA\Op$:
      \begin{align*}
        \gB\TotOp\prn{x}
         & :\equiv
        \gB\prn{x}\Op
        \\
        \Push{\gB\TotOp}{p}{u}
         & :\equiv
        \Pull{\gB}{p}{u}
        \\
        \PushRx{\gB\TotOp}[x]{u}
         & :\equiv
        \PullRx{\gB}[x]{u}
      \end{align*}
    \end{definition}

    The purpose of \emph{total} opposites is to exhibit a duality involution identifying oplax covariant and lax contravariant lenses.

    \begin{observation}[Duality involution for lenses]\label[observation]{obs:lens-duality}
      By virtue of \cref{def:tot-op-lens}, we see that an oplax covariant lens over $\gA$ is precisely the same thing as a lax contravariant lens over $\gA\Op$; and, conversely, that a lax contravariant lens over $\gA$ is precisely the same thing as an oplax covariant lens over $\gA\Op$.
    \end{observation}

    \begin{lemma}[Display of total opposites]\label[lemma]{lem:display-of-tot-op}
      Let $\gB$ be an oplax covariant lens of reflexive graphs over a reflexive graph $\gA$. Then we have a definitional equivalence $\CtrvDisp{\gA\Op}{\prn{\gB\TotOp}} \equiv \prn{\CovDisp{\gA}{\gB}}\TotOp$ of displayed reflexive graphs.
    \end{lemma}

    \begin{proof}
      We compute explicitly.

      \iblock{
        \begin{multicols}{2}
          \mhang{
            \vrt{\CtrvDisp{\gA\Op}{\prn{\gB\TotOp}}}\prn{x}
          }{
            \commentrow{by \cref{def:ctrv-disp}}
            \mrow{
              {}\equiv
              \vrt{\gB\TotOp\prn{x}}
            }
            \commentrow{by \cref{def:tot-op-lens}}
            \mrow{
              {}\equiv
              \vrt{\gB\prn{x}\Op}
            }
            \commentrow{by \cref{def:op-rx-gph}}
            \mrow{
              {}\equiv
              \vrt{\gB\prn{x}}
            }
          }
          \columnbreak
          \mhang{
            \vrt{\prn{\CovDisp{\gA}{\gB}}\TotOp}\prn{x}
          }{
            \commentrow{by \cref{def:tot-op-disp-rx-gph}}
            \mrow{
              {}\equiv
              \vrt{\CovDisp{\gA}{\gB}}\prn{x}
            }
            \commentrow{by \cref{def:cov-disp}}
            \mrow{
              {}\equiv
              \vrt{\gB\prn{x}}
            }
          }
        \end{multicols}

        \begin{multicols}{2}
          \mhang{
            u \Edge{\CtrvDisp{\gA\Op}{\prn{\gB\TotOp}}}[p:x\Edge{\gA\Op}y] v
          }{
            \commentrow{by \cref{def:ctrv-disp}}
            \mrow{
              {}\equiv
              u \Edge{\gB\TotOp\prn{y}} \Pull{\gB\TotOp}{p}{v}
            }
            \commentrow{by \cref{def:tot-op-lens}}
            \mrow{
              {}\equiv
              u \Edge{\gB\prn{y}\Op} \Push{\gB}{p}{v}
            }
            \commentrow{by \cref{def:op-rx-gph}}
            \mrow{
              {}\equiv
              \Push{\gB}{p}{v} \Edge{\gB\prn{y}} u
            }
          }

          \columnbreak

          \mhang{
            u \Edge{\prn{\CovDisp{\gA}{\gB}}\TotOp}[p:x\Edge{\gA\Op}y] v
          }{
            \commentrow{by \cref{def:tot-op-disp-rx-gph}}
            \mrow{
              {}\equiv
              v\Edge{\CovDisp{\gA}{\gB}}[p:y\Edge{\gA}x] v
            }
            \commentrow{by \cref{def:cov-disp}}
            \mrow{
              {}\equiv
              \Push{\gB}{p}{v} \Edge{\gB\prn{y}} u
            }
          }
        \end{multicols}

        \begin{multicols}{2}
          \mhang{
            \DRx{\CtrvDisp{\gA\Op}{\prn{\gB\TotOp}}}{x}{u}
          }{
            \commentrow{by \cref{def:ctrv-disp}}
            \mrow{
              {}\equiv
              \PullRx{\gB\TotOp}[x]{u}
            }
            \commentrow{by \cref{def:tot-op-lens}}
            \mrow{
              {}\equiv
              \PushRx{\gB}[x]{u}
            }
          }
          \columnbreak
          \mhang{
            \DRx{\prn{\CovDisp{\gA}{\gB}}\TotOp}{x}{u}
          }{
            \commentrow{by \cref{def:tot-op-disp-rx-gph}}
            \mrow{
              {}\equiv
              \DRx{\CovDisp{\gA}{\gB}}{x}{u}
            }
            \commentrow{by \cref{def:cov-disp}}
            \mrow{
              {}\equiv
              \PushRx{\gB}[x]{u}
            }
            \qedhere
          }
        \end{multicols}
      }
    \end{proof}

    \begin{corollary}
      Let $\gB$ be a lax contravariant lens of reflexive graphs over a reflexive graph $\gA$. Then we have a definitional equivalence $\CovDisp{\gA\Op}{\prn{\gB\TotOp}} \equiv \prn{\CtrvDisp{\gA}{\gB}}\TotOp$ of displayed reflexive graphs.
    \end{corollary}

    \begin{proof}
      By \cref{lem:display-of-tot-op} via duality (\cref{obs:lens-duality}).
    \end{proof}

    Although many other definitions can be obtained directly by duality, we will continue to specify various concepts and results for both oplax covariant and lax contravariant lenses in order to facilitate easy reference.

  \end{xsect}

\begin{xsect}[sec:definitional-lenses]{Definitional lenses of reflexive graphs}
  An (oplax covariant, lax contravariant) lens equips a family with a pushforward or pullback operator together with a witness of a unit law holding up to a directed edge. That the unit law holds only up to an edge is the meaning of the ``oplax/lax'' terminology. An important class of lenses are the ones in which the unit law holds up to definitional equality, so that the oplax or lax witness or the unit law can be given by the reflexivity datum of the component reflexive graph. We study these \emph{definitional lenses} in the present section, with an eye to developing the theory of polynomials and partial products of lenses in \cref{sec:polynomials}.

  \begin{definition}[Definitional lens]\label[definition]{def:definitional-lens}
    A \DefEmph{definitional covariant lens} is an oplax covariant lens $\gB$ over $\gA$ in which the following equations hold definitionally:
    \begin{align*}
      x:\vrt{\gA}, u:\vrt{\gB\prn{x}}
       & \vdash
      \Push{\gB}{\Rx{\gA}{x}}{u} \equiv u
      \\
      x:\vrt{\gA},u:\vrt{\gB\prn{x}}
       & \vdash
      \PushRx{\gB}[x]{u} \equiv \Rx{\gB\prn{x}}{u}
    \end{align*}

    By the same token, a \DefEmph{definitional contravariant lens} is a lax contravariant lens in which the following equations hold definitionally:
    \begin{align*}
      x:\vrt{\gA}, u:\vrt{\gB\prn{x}}
       & \vdash
      \Pull{\gB}{\Rx{\gA}{x}}{u} \equiv u
      \\
      x:\vrt{\gA},u:\vrt{\gB\prn{x}}
       & \vdash
      \PullRx{\gB}[x]{u} \equiv \Rx{\gB\prn{x}}{u}
    \end{align*}
  \end{definition}

  \begin{remark}
  We note that \cref{def:definitional-lens} must be understood ``metatheoretically'' or ``judgmentally'', because definitional equality cannot be stated as a type in Martin-L\"of type theory. One way to formalise this would be by means of a conservative extension such as two-level type theory~\citep{annenkov-capriotti-kraus:2017}; we remain agnostic and treat \cref{def:definitional-lens} as a definitional extension of the \emph{judgements} of Martin-L\"of type theory rather than as a type.
  \end{remark}

  \begin{example}[Definitional lenses over discrete reflexive graphs]\label{example:def-lens-over-disc-rx-gph}
    Let $\gB(x)$ be a family of reflexive graphs indexed in $x:A$. Then we may turn $\gB$ canonically into an oplax covariant lens or a lax contravariant lens over the discrete reflexive graph $\DiscPO{A}$ as follows using \emph{transport}:
    \begin{align*}
      x,y:A; p:\Id{A}{x}{y}, u:\vrt{\gB(x)} &\vdash
      \Push{\gB}{p}u :\equiv p_*u\\
      x,y:A; p:\Id{A}{x}{y}, u:\vrt{\gB(y)} &\vdash
      \Pull{\gB}{p}u :\equiv p\Inv_*u
    \end{align*}

    The total reflexive graph of the display of $\gB$ \emph{qua} covariant lens is then definitionally equal to the coproduct $\Coprod{x:A}\gB\prn{x}$ from \cref{ex:coprod-rx-gph}.
  \end{example}

  \begin{computation}[Display of definitional lenses]
    The display of definitional lenses is, naturally, more simple than that of ordinary lenses. If $\gB$ is a definitional covariant lens over $\gA$, its display computes as follows:
    \begin{align*}
      \vrt{\CovDisp{\gA}{\gB}}\prn{x}
      &\equiv
      \vrt{\gB\prn{x}}
      \\
      u \Edge{\CovDisp{\gA}{\gB}}[p:x\Edge{\gA}y] v
      &\equiv
      \Push{\gB}{p}{u} \Edge{\gB\prn{y}} v
      \\
      \DRx{\CovDisp{\gA}{\gB}}{x}{u}
      &\equiv
      \Rx{\gB\prn{x}}{u}
    \end{align*}

    Likewise, if $\gB$ is a definitional contravariant lens, its display computes as follows:
    \begin{align*}
      \vrt{\CtrvDisp{\gA}{\gB}}\prn{x}
      &\equiv
      \vrt{\gB\prn{x}}
      \\
      u \Edge{\CtrvDisp{\gA}{\gB}}[p:x\Edge{\gA}y] v
      &\equiv
      u \Edge{\gB\prn{x}} \Pull{\gB}{p}{v}
      \\
      \DRx{\CtrvDisp{\gA}{\gB}}{x}{u}
      &\equiv
      \Rx{\gB\prn{x}}{u}
    \end{align*}
  \end{computation}

  \begin{xsect}{Example: definitional lenses from univalent families}
    \begin{example}[Definitional lenses of from univalent families]\label[example]{ex:defn-cl-uni-fam}
      Any univalent family of path objects has a canonical \emph{definitional} lens structure. In particular, given a univalent family of path objects $\prn{U,\gE}$ we can equip $\gE$ with the structure of a covariant and a contravariant lens over the reflexive graph image $U/\gE$:
      \begin{align*}
        A_0,A_1:U; f : A_0\Edge{U/\gE} A_1; a : \vrt{\gE\prn{A_0}}
         & \vdash
        \Push{\gE}{f}{a} :\equiv f a
        \\
        A_0,A_1:U; f : A_0\Edge{U/\gE} A_1; a : \vrt{\gE\prn{A_1}}
         & \vdash
        \Pull{\gE}{f}{a} :\equiv f\Inv a
        \\
        A:U; a:\vrt{\gE\prn{A}}
         & \vdash
        \PushRx{\gE}[A]{a} :\equiv \Rx{\gE\prn{A}}{a}
        \\
        A:U; a:\vrt{\gE\prn{A}}
         & \vdash
        \PullRx{\gE}[A]{a} :\equiv \Rx{\gE\prn{A}}{a}
      \end{align*}
    \end{example}

    \begin{example}[Definitional lenses from universes]\label[example]{ex:defn-cl-universe}
      Any univalent family of types $\prn{U,E}$, \eg a universe, determines a univalent family of path objects $\prn{U,\DiscPO{E}}$. In this case, $U/\vrt{\DiscPO{E}}\equiv U/E$ is precisely the canonical path object structure on $U$ determined by equivalences between components of $E$, and we may consider the corresponding definitional lens structures on $\DiscPO{E}$ by specialising \cref{ex:defn-cl-uni-fam}.
    \end{example}

    \begin{example}[Definitional lenses from subuniverses]\label[example]{ex:defn-cl-subuniverse}
      Let $\prn{U,\gE}$ be a univalent family, and let $P\prn{A}$ be a proposition for every $A:U$. We can restrict $\gE$ to a family of path objects $\gE_P$ over $\Compr{A:U/\vrt{\gE}}{P\prn{A}}$  setting $\gE_P\prn{A,h} :\equiv \gE\prn{A}$. Then $\gE_P$ can be equipped with the structure of a definitional covariant and contravariant lens over $U/\vrt{\gE}$ in the following way:
      \begin{align*}
        A, B:\vrt{\Compr{A:U/\vrt{\gE}}{P\prn{A}}}; f : A\Edge{U/\vrt{\gE}}B, x : \vrt{\mathcal{E}\prn{\pi_1{A}}}
         & \vdash
        \Push{\gE_P}{f}{x} :\equiv fx
        \\
        A, B:\vrt{\Compr{A:U/\vrt{\gE}}{P\prn{A}}}; f : A\Edge{U/\vrt{\gE}}B, x : \vrt{\mathcal{E}\prn{\pi_1{B}}}
         & \vdash
         \Pull{\gE_P}{f}{x} :\equiv f\Inv{x}
        \\
        A:\vrt{\Compr{A:U/\vrt{\gE}}{P\prn{A}}}, x : \vrt{\mathcal{E}\prn{\pi_1{A}}}
         & \vdash
        \PushRx{\gE_P}[A]{x} :\equiv \Rx{\gE\prn{\pi_1A}}{x}
        \\
        A:\vrt{\Compr{A:U/\vrt{\gE}}{P\prn{A}}}, x : \vrt{\mathcal{E}}\prn{\pi_1{A}}
         & \vdash
        \PullRx{\gE_P}[A]{x} :\equiv \Rx{\gE\prn{\pi_1A}}{x}
      \end{align*}
    \end{example}
  \end{xsect}

  \begin{xsect}[sec:finite-ordinals]{Example: definitional lenses for finite ordinals}
    For each $n:\mathbb{N}$, let $\gF\prn{n}$ be the following path object with vertices in the standard finite set with $n$ elements:
    \[
      \gF\prn{n} :\equiv
      \Compr{i:\DiscPO{\mathbb{N}}}{i < n}
    \]

    The reflexive graph image of $\mathbb{N}$ under $\gF$ is \emph{not} univalent because $\mathbb{N}$ is a set and thus each $\Id{\mathbb{N}}{m}{n}$ is a proposition whereas there is a proper set of equivalences from $\gF\prn{m}$ to $\gF\prn{m}$ for $m \geq 2$. Therefore, we cannot apply \cref{ex:defn-cl-uni-fam} nor \cref{ex:defn-cl-subuniverse} to obtain a good definitional lens structure on $\gF$. Nonetheless, we can impose with some ingenuity a more restricted path object structure on $\mathbb{N}$ over which it is not difficult to exhibit $\gF$ as a definitional lens.

    \begin{example}\label[example]{ex:augmented-simplices}
      We may classify the \DefEmph{augmented simplices} by the following reflexive graph structure $\AugSpx$ with vertices in the natural numbers and edges given by monotone equivalences between finite ordinals:
      \begin{align*}
        \vrt{\AugSpx}
         & :\equiv
        \mathbb{N}
        \\
        m \Edge{\AugSpx} n
         & :\equiv
        \Sum{f : \mathsf{Equiv}\prn{\vrt{\gF\prn{m}},\vrt{\gF\prn{n}}}}
        \mathsf{isMonotone}\,f
        \\
        \Rx{\AugSpx}{n}
         & :\equiv
        \prn{1_{\vrt{\gF\prn{n}}}, \mathsf{idnMono}}
      \end{align*}
    \end{example}

    We shall deduce that \cref{ex:augmented-simplices} yields a path object, using the fact that the category of augmented simplices is \emph{gaunt} in the sense that there is at most one monotone isomorphism between any two finite ordinals.

    \begin{lemma}[A path object classifying augmented simplices]
      The reflexive graph $\AugSpx$ of augmented simplices is univalent.
    \end{lemma}
    \begin{proof}
      Fixing $m:\mathbb{N}$, we must check that the fan $\Fan{\AugSpx}{m}$ is a proposition. Unfolding definitions, we have
      $
        \Fan{\AugSpx}{m}
        \equiv
        \Sum{n:\mathbb{N}}
        \Sum{f:\mathsf{Equiv}\prn{\vrt{\gF\prn{m}},\vrt{\gF\prn{n}}}}
        \mathsf{isMonotone}\,f
      $. As there can be at most one monotone equivalence between $\vrt{\gF\prn{m}}$ and some $\vrt{\gF\prn{n}}$, it remains only to observe that $\vrt{\gF\prn{n}}$ being equinumerous with $\vrt{\gF\prn{m}}$ implies $m=n$.
    \end{proof}

    \begin{example}[Definitional lens structure of finite ordinals]
      We may exhibit definitional covariant and contravariant lens structures  on $\gF$ over $\AugSpx$ like so:
      \begin{align*}
        m,n:\mathbb{N}; \prn{f,f^\leq} : m \Edge{\AugSpx} n, i : \vrt{\mathcal{F}\prn{m}}
         & \vdash
         \Push{\gF}{\prn{f,f^\leq}}{i}
        :\equiv
        f i
        \\
        m,n:\mathbb{N}; \prn{f,f^\leq} : m \Edge{\AugSpx} n, i : \vrt{\mathcal{F}\prn{n}}
         & \vdash
         \Pull{\gF}{\prn{f,f^\leq}}{i}
        :\equiv
        f\Inv i
        \\
        m:\mathbb{N}; i:\vrt{\mathcal{F}\prn{m}}
         & \vdash
        \PushRx{\gF}[m]{i} :\equiv \Rx{\gF\prn{m}}{i}
        \\
        m:\mathbb{N}; i:\vrt{\mathcal{F}\prn{m}}
         & \vdash
        \PullRx{\gF}[m]{i} :\equiv \Rx{\gF\prn{m}}{i}
      \end{align*}
    \end{example}
  \end{xsect}

  \begin{xsect}{Definitional replacement of lenses}
    We have seen that in some examples of definitional lenses, the reflexive graph structure on the base required some ingenuity. In this section, we describe a general method to reparameterise any lens as a \emph{definitional} lens over a different base.

    \begin{construction}[Flattening of lenses]
      Let $\gB$ be an oplax covariant or lax contravariant lens of reflexive graphs over a reflexive graph $\gA$.  We may equip $\vrt{\gA}$ with a new reflexive graph structure $\CovFlatten{\gA}{\gB}$ or $\CtrvFlatten{\gA}{\gB}$ respectively called the \DefEmph{flattening} of $\gB$ onto $\gA$.
      \begin{align*}
        \vrt{\CovFlatten{\gA}{\gB}}
         & :\equiv
        \vrt{\gA}
        \\
        x \Edge{\CovFlatten{\gA}{\gB}} y
         & :\equiv
        \Sum{p:x\Edge{\gA}y}
        \Sum{p' : \vrt{\gB\prn{x}}\to\vrt{\gB\prn{y}}}
        \Prod{u:\vrt{\gB\prn{x}}}
        \Push{\gB}{p}{u} \Edge{\gB\prn{y}} p'u
        \\
        \Rx{\CovFlatten{\gA}{\gB}}{x}
         & :\equiv
        \prn{
          \Rx{\gA}{x},
          1_{\vrt{\gB{x}}},
          \PushRx{\gB}[x]
        }
        \\[12pt]
        \vrt{\CtrvFlatten{\gA}{\gB}}
         & :\equiv
        \vrt{\gA}
        \\
        x \Edge{\CtrvFlatten{\gA}{\gB}} y
         & :\equiv
        \Sum{p:x\Edge{\gA}y}\Sum{p':\vrt{\gB\prn{y}}\to\vrt{\gB\prn{x}}}
        \Prod{u:\vrt{\gB\prn{y}}}
        p'u \Edge{\gB\prn{x}} \Pull{\gB}{x}u
        \\
        \Rx{\CtrvFlatten{\gA}{\gB}}{x}
         & :\equiv
        \prn{
          \Rx{\gA}{x},
          1_{\vrt{\gB{x}}},
          \PullRx{\gB}[x]
        }
      \end{align*}
    \end{construction}

    \begin{lemma}[Flattening of lenses of path objects onto path objects]\label[lemma]{lem:flattening-univalent}
      Let $\gB$ be an oplax covariant or lax contravariant lens of path objects over a path object $\gA$. Assuming function extensionality, the flattening of $\gB$ onto $\gA$ is univalent.
    \end{lemma}

    \begin{proof}
      We consider only the covariant case, as the contravariant one is similar. Fixing $x:\vrt{\gA}$, we must check that the fan $\Fan{\CovFlatten{\gA}{\gB}}{x}$ is a proposition. First we compute the fan up to definitional equality:
      \begin{align*}
        \Fan{\CovFlatten{\gA}{\gB}}{x}
         & \equiv
        \Sum{y:\gA} x\Edge{\CovFlatten{\gA}{\gB}} y
        \\
         & \equiv
        \Sum{y:\gA}
        \Sum{p:x\Edge{\gA}y}\Sum{p' : \vrt{\gB\prn{x}}\to\vrt{\gB\prn{y}}}
        \Prod{u:\vrt{\gB\prn{x}}}
        \Push{\gB}{p}{u} \Edge{\gB\prn{y}} p'u
        \\
         & \equiv
        \Sum{y:\gA}
        \Sum{p:x\Edge{\gA}y}\Sum{p' : \vrt{\gB\prn{x}}\to\vrt{\gB\prn{y}}}
        \Push{\gB}{p} \Edge{\vrt{\gB\prn{x}}\pitchfork\gB\prn{y}} p'
        \\
         & \equiv
        \Sum{y:\gA}
        \Sum{p:x\Edge{\gA}y}
        \Fan{\vrt{\gB\prn{x}}\pitchfork\gB\prn{y}}{\Push{\gB}{p}}
      \end{align*}

      As each $\gB\prn{y}$ is univalent, we deduce from function extensionality that the cotensor $\vrt{\gB\prn{x}}\pitchfork\gB\prn{y}$ is univalent. Therefore, each fan $\Fan{\vrt{\gB\prn{x}}\pitchfork\gB\prn{y}}{\Push{\gB}{p}}$ is contractible and so we have the following equivalences:
      \begin{align*}
        \Fan{\CovFlatten{\gA}{\gB}}{x}
         & \cong
        \Sum{y:\vrt{\gA}}
        \Sum{p:x\Edge{\gA}y}
        \Fan{\vrt{\gB\prn{x}}\pitchfork\gB\prn{y}}{\Push{\gB}{p}}
        \\
         & \cong
        \Sum{y:\vrt{\gA}} x\Edge{\gA}y
        \\
         & \equiv
        \Fan{\gA}{x}
      \end{align*}

      As $\gA$ is assumed to be univalent, $\Fan{\gA}{x}$ is a proposition and thus so is $\Fan{\CovFlatten{\gA}{\gB}}{x}$.
    \end{proof}

    \begin{construction}[The definitional replacement of a lens]
      Let $\gB$ be an oplax covariant lens of reflexive graphs over a reflexive graph $\gA$. The \DefEmph{covariant definitional replacement} of $\gB$ is defined to be the following \emph{definitional} covariant lens $\CovRepl{\gA}{\gB}$ over the flattening $\CovFlatten{\gA}{\gB}$ with $\gB$ as its underlying family of reflexive graphs.
      \begin{align*}
        \prn{\CovRepl{\gA}{\gB}}\prn{x}
        &:\equiv \gB\prn{x}
        \\
        \Push{\CovRepl{\gA}{\gB}}{\prn{p,p'}}{u}
        &:\equiv
        p'u
      \end{align*}
    \end{construction}

    On the other hand, if $\gB$ is a lax contravariant lens of reflexive graphs over a reflexive graph $\gA$, we define the \DefEmph{contravariant definitional replacement} of $\gB$ to be the following definitional contravariant lens $\CtrvRepl{\gA}{\gB}$ over the flattening $\CtrvFlatten{\gA}{\gB}$:
    \begin{align*}
      \prn{\CtrvRepl{\gA}{\gB}}\prn{x}
      &:\equiv \gB\prn{x}
      \\
      \Pull{\CtrvRepl{\gA}{\gB}}{\prn{p,p'}}{u}
      &:\equiv
      p'u
    \end{align*}

  \end{xsect}

\begin{xsect}[sec:polynomials]{Polynomials and partial products of reflexive graphs}
  Any family of types $x:A\vdash B\prn{x} : \TYPE$ induces a \DefEmph{polynomial} type operator $\Poly{A}{B}{C} :\equiv \Sum{x:A}\Prod{u:B\prn{x}}C$; the universal property of $\Poly{A}{B}{C}$ is that of a \emph{partial product}~\citep{dyckhoff-tholen:1987}. In this section, we generalise the construction of polynomials to lenses of reflexive graphs; the definitions of the polynomial reflexive graphs will then mirror the description of \emph{covariant} and \emph{contravariant} partial products in a 2-category~\citep{johnstone:1993,johnstone:2002}. We then show how to specialise the construction of polynomials to obtain suitable path object structures on the \emph{partial map classifier} of a given dominance as in synthetic domain theory~\citep{rosolini:1986}.

   In this section, we assume function extensionality.

  \NewDocumentCommand\CovPP{mmm}{\mathsf{pp}^+\Sub{#1}\prn{#3,#2}}
  \NewDocumentCommand\CtrvPP{mmm}{\mathsf{pp}^-\Sub{#1}\prn{#3,#2}}

  \begin{construction}\label[construction]{con:pp-right-cleaving}
    Let $\gB$ be a definitional covariant lens of reflexive graphs over a reflexive graph $\gA$, and let $\gC$ be an arbitrary reflexive graph. Then we may define a \emph{definitional contravariant} lens of reflexive graphs $\CovPP{\gA}{\gB}{\gC}$ over $\gA$ by setting each $\prn{\CovPP{\gA}{\gB}{\gC}}\prn{x} :\equiv \vrt{\gB\prn{x}}\pitchfork \gC$. The definitional contravariant lens structure of $\CovPP{\gA}{\gB}{\gC}$ over $\gA$ is defined from the definitional covariant lens structure of $\gB$ by precomposition like so:

    \iblock{
      \mhang{
        x,y:\vrt{\gA}; p:x\Edge{\gA}y; c:\vrt{\gB\prn{y}}\to\vrt{\gC}
      }{
        \mrow{
          \vdash
          \Pull{\CovPP{\gA}{\gB}{\gC}}{p}{c}
          :\equiv
          c \circ \Push{\gB}{p}
        }
      }
    }
  \end{construction}

  \begin{construction}\label[construction]{con:pp-left-cleaving}
    Let $\gB$ be a definitional contravariant lens of reflexive graphs over a reflexive graph $\gA$, and let $\gC$ be an arbitrary reflexive graph. Then we may define a \emph{definitional covariant} family of path objects $\CtrvPP{\gA}{\gB}{\gC}$ over $\gA$ by setting each $\prn{\CtrvPP{\gA}{\gB}\gC}\prn{x} :\equiv \vrt{\gB\prn{x}}\pitchfork \gC$. The definitional covariant lens structure of $\CtrvPP{\gA}{\gB}\gC$ over $\gA$ is defined from the definitional contravariant lens struture of $\gB$ by precomposition:

    \iblock{
      \mhang{
        x,y:\vrt{\gA}; p:x\Edge{\gA}y; c:\vrt{\gB\prn{x}}\to\vrt{\gC}
      }{
        \mrow{
          \vdash
          \Push{\CtrvPP{\gA}{\gB}{\gC}}{p}{c}
          :\equiv
          c \circ \Pull{\gB}{p}
        }
      }
    }
  \end{construction}

  \begin{definition}[Partial products of definitional lenses]
    Let $\gB$ be a definitional covariant lens of reflexive graphs over a reflexive graph $\gA$. Then the \DefEmph{covariant partial product} of a reflexive graph $\gC$ with $\gB$ over $\gA$ is defined to be the following total reflexive graph:
    \[ 
      \CovPoly{\gA}{\gB}{\gC} :\equiv 
      \gA.\CtrvDisp{\gA}\prn{\CovPP{\gA}{\gB}{\gC}}
    \] 

    When $\gB$ is a definitional contravariant lens, we define the \DefEmph{contravariant partial product} of $\gC$ with $\gB$ over $\gA$ as follows:
    \[ 
      \CtrvPoly{\gA}{\gB}{\gC} :\equiv 
      \gA.\CovDisp{\gA}\prn{\CtrvPP{\gA}{\gB}{\gC}}
    \] 
  \end{definition}

  \begin{computation}[Covariant partial product]
    The covariant partial product of a reflexive graph $\gC$ with a definitional covariant lens of reflexive graphs $\gB$ over a reflexive graphs $\gA$ computes as follows:
    \begin{align*}
      \vrt{\CovPoly{\gA}{\gB}{\gC}}
       & \equiv
      \Poly{\vrt{\gA}}{\vrt{\gB}}{\vrt{\gC}}
      \\
      \prn{x,c} \Edge{\CovPoly{\gA}{\gB}{\gC}} \prn{y,d}
       & \equiv
      \Sum{p:x\Edge{\gA}y}
      \Prod{u:\gB\prn{x}}
      cu \Edge{\gC} d\prn{\Push{\gB}{p}{u}}
      \\
      \Rx{\CovPoly{\gA}{\gB}{\gC}}\prn{x,c}
       & \equiv
      \prn{\Rx{\gA}{x},\Lam{u}\Rx{\gC}\prn{cu}}
    \end{align*}
  \end{computation}

  \begin{computation}[Contravariant partial product]
    The contravariant partial product of a reflexive graph $\gC$ with a definitional contravariant lens of reflexive graphs $\gB$ over a reflexive graph $\gA$ unravels up to definitional equivlance as follows:
    \begin{align*}
      \vrt{\CtrvPoly{\gA}{\gB}{\gC}}
       & \equiv
      \Poly{\vrt{\gA}}{\vrt{\gB}}{\vrt{\gC}}
      \\
      \prn{x,c} \Edge{\CtrvPoly{\gA}{\gB}{\gC}} \prn{y,d}
       & \equiv
      \Sum{p:x\Edge{\gA}y}
      \Prod{u:\gB\prn{y}}
      c\prn{\Pull{\gB}{p}{u}} \Edge{\gC} du
      \\
      \Rx{\CtrvPoly{\gA}{\gB}{\gC}}\prn{x,c}
       & \equiv
      \prn{\Rx{\gA}{x},\Lam{u}\Rx{\gC}\prn{cu}}
    \end{align*}
  \end{computation}

  \begin{example}[Partial map classifiers]
    A \DefEmph{dominance} in the sense of \citet{rosolini:1986} is a univalent universe $\prn{\Sigma, T}$ closed under the terminal type and internal sums such that for each $\phi:\Sigma$, the component $T\prn{\phi}$ is a proposition. In this case, we may equip each $T\prn{\phi}$ with the codiscrete path object structure $\CodiscPO\prn{T\prn{\phi}}$ so that $\prn{\Sigma,\CodiscPO{T}}$ is a univalent family of path objects and $\CodiscPO{T}$ has definitional covariant and contravariant lens structure over the propositional reflexive graph image $\Sigma/_{-1}T$.

    The partial map classifier of a given type $A$ in the dominance $\prn{\Sigma,T}$ is the polynomial type $\Poly{\Sigma}{T}{A} \equiv \Sum{\phi:\Sigma}\Prod{p:T\prn{\phi}} A$. A suitable lifting of the partial map classifiers to path objects $\gA$ can be obtained by means of covariant or contravariant partial products:
    \[
      \gA_\bot^\covsym :\equiv \CovPoly{\Sigma/_{-1}T}{\CodiscPO{T}}{\gA}
      \qquad
      \gA_\bot^\ctrvsym :\equiv \CtrvPoly{\Sigma/_{-1}T}{\CodiscPO{T}}{\gA}
    \]

    We compute the above reflexive graph structures as follows:
    \begin{align*}
       & \vrt{\gA_\bot^\covsym}
      \equiv
      \vrt{\gA_\bot^\ctrvsym}
      \equiv
      \Poly{\Sigma}{T}{\vrt{\gA}}
      \\
       & \prn{\phi,a} \Edge{\gA_\bot^\covsym} \prn{\psi,b}
      \equiv
      \Sum{\prn{p,q}:\prn{\phi\to \psi}\times\prn{\psi\to\phi}}
      \Prod{u:T\prn{\phi}}
      au \Edge{A} b\prn{pu}
      \\
       & \prn{\phi,a} \Edge{\gA_\bot^\ctrvsym} \prn{\psi,b}
      \equiv
      \Sum{\prn{p,q}:\prn{\phi\to \psi}\times\prn{\psi\to\phi}}
      \Prod{u:T\prn{\psi}}
      a\prn{qu} \Edge{A} bu
      \\
       & \Rx{\gA_\bot^\covsym}\prn{\phi,a}
      \equiv
      \Rx{\gA_\bot^\ctrvsym}\prn{\phi,a}
      \equiv
      \prn{\prn{\Lam{u}{u},\Lam{u}{u}},\Lam{u}\Rx{\gA}\prn{au}}
    \end{align*}
  \end{example}

  \begin{example}[Lists]
    A simple example of a polynomial type operator is that of \emph{lists}: letting $F\prn{n}$ be the standard finite type $\Sum{i:\mathbb{N}}i<n$, we may define $\mathsf{List}\prn{A} :\equiv \Poly{\mathbb{N}}{F}{A}$. This representation of lists lifts to path objects by means of partial products via our account in \cref{sec:finite-ordinals} of the path object $\AugSpx$ classifying finite ordinals and the definitional lens $\gF$ of elements of a given finite ordinal. In particular, for any path object $\gA$ we consider the following covariant and contravariant partial products:
    \[
      \mathsf{List}^\covsym\prn{\gA} :\equiv \CovPoly{\AugSpx}{\gF}{\gA}
      \qquad
      \mathsf{List}^\ctrvsym\prn{\gA} :\equiv \CtrvPoly{\AugSpx}{\gF}{\gA}
    \]

    These reflexive graphs compute as follows:
    \begin{align*}
       & \vrt{\mathsf{List}^\covsym\prn{\gA}}
      \equiv
      \vrt{\mathsf{List}^\ctrvsym\prn{\gA}}
      \equiv
      \mathsf{List}\prn{A} \equiv \Sum{n:\mathbb{N}}\Prod{i:F\prn{n}}A
      \\
       & \prn{m,a} \Edge{\mathsf{List}^\covsym\prn{\gA}} \prn{n,b}
      \equiv
      \Sum{\prn{f,f^\leq} : m \Edge{\AugSpx} n}
      \Prod{i:F\prn{m}}
      a i \Edge{\gA} b\prn{f i}
      \\
       & \prn{m,a} \Edge{\mathsf{List}^\ctrvsym\prn{\gA}} \prn{n,b}
      \equiv
      \Sum{\prn{f,f^\leq} : m \Edge{\AugSpx} n}
      \Prod{i:F\prn{n}}
      a \prn{f i} \Edge{\gA} b i
      \\
       & \Rx{\mathsf{List}^\covsym\prn{\gA}}\prn{m,a}
      \equiv
      \Rx{\mathsf{List}^\ctrvsym\prn{\gA}}\prn{m,a}
      \equiv
      \prn{\prn{\mathsf{idnEquiv}_{F\prn{m}},\mathsf{idnMono}},\Lam{i}\Rx{\gA}\prn{ai}}
    \end{align*}
  \end{example}

  \citet{buchholtz:2023:pairs} has shown how to use the classifying space of the symmetric group on two elements to define several variations on unordered pairs in univalent foundations. In what follows, we assume the existence of propositional truncations $\vvrt{-}$.

  \begin{example}[Homotopy unordered pairs~\citep{buchholtz:2023:pairs}]
    Let $\prn{U,E}$ be a univalent universe containing the two-element type $\mathbf{2}$; following \citet{buchholtz:2023:pairs} we may define the ``classifying type'' $B\Sigma_2$ to
    the type of $U$-small types merely equivalent to $\mathbf{2}$ as follows, \ie we define $B\Sigma_2 :\equiv \Sum{X:U}\vvrt{\mathsf{Equiv}\prn{E\prn{X},\mathbf{2}}}$. We define a decoding family $K:B\Sigma_2 \vdash E_2\prn{K}:\TYPE$ by setting $E_2\prn{X,h} :\equiv E\prn{X}$. Then \opcit defines the type of \DefEmph{homotopy unordered pairs} in a type $A$ to be the following polynomial type:
    \[
      \mathsf{hUP}\prn{A} :\equiv \Poly{B\Sigma_2}{E_2}{A}
      \equiv
      \Sum{K:B\Sigma_2}\Prod{i:E_2\prn{K}}A
    \]

    A suitable path object structure on $B\Sigma_2$ is obtained by path object comprehension:
    \[
      \gB\Sigma_2 :\equiv \Compr{X:U/E}{\vvrt{\mathsf{Equiv}\prn{E\prn{X},\mathbf{2}}}}
    \]

    Then $E_2$ lifts discretely to a definitional covariant and contravariant lens $\gE_2$ of path objects over $\gB\Sigma_2$ by \cref{ex:defn-cl-subuniverse}, setting $\gE_2\prn{K} :\equiv \DiscPO\prn{E_2\prn{K}}$. A suitable path object structure on homotopy unordered pairs is then obtained by either covariant or contravariant partial products:
    \[
      \mathsf{hUP}^\covsym\prn{\gA} :\equiv \CovPoly{\gB\Sigma_2}{\gE_2}{\gA}
      \qquad
      \mathsf{hUP}^\ctrvsym\prn{\gA} :\equiv \CtrvPoly{\gB\Sigma_2}{\gE_2}{\gA}
    \]

    These reflexive graph structures compute as follows:
    \begin{align*}
       & \vrt{\mathsf{hUP}^\covsym\prn{\gA}}
      \equiv
      \vrt{\mathsf{hUP}^\ctrvsym\prn{\gA}}
      \equiv
      \mathsf{hUP}\prn{\vrt{\gA}}
      \\
       &
      \prn{\prn{L,h_L},a} \Edge{\mathsf{hUP}^\covsym\prn{\gA}}\prn{\prn{K,h_K},b}
      \equiv
      \Sum{f : \mathsf{Equiv}\prn{E\prn{L},E\prn{K}}}
      \Prod{i:E\prn{L}}
      ai \Edge{\gA} b\prn{fi}
      \\
       &
      \prn{\prn{L,h_L},a} \Edge{\mathsf{hUP}^\ctrvsym\prn{\gA}}\prn{\prn{K,h_K},b}
      \equiv
      \Sum{f : \mathsf{Equiv}\prn{E\prn{L},E\prn{K}}}
      \Prod{i:E\prn{K}}
      a\prn{fi} \Edge{\gA} bi
      \\
       &
      \Rx{\mathsf{hUP}^\covsym\prn{\gA}}\prn{\prn{K,h},a}
      \equiv
      \Rx{\mathsf{hUP}^\ctrvsym\prn{\gA}}\prn{\prn{K,h},a}
      \equiv
      \prn{\mathsf{idnEquiv}_{E\prn{K}}, \Lam{i}\Rx{\gA}\prn{ai}}
    \end{align*}
  \end{example}
\end{xsect}
 
\end{xsect}
\begin{xsect}[sec:coherence]{Coherence in the theory of path object lenses}
  \begin{lemma}\label[lemma]{lem:rx-gph-univalence-is-prop}
    Assuming dependent function extensionality, it is a proposition that a given reflexive graph is univalent.
  \end{lemma}

  \begin{proof}
    Let $\gA$ be a reflexive graph; a witness that $\gA$ is univalent is an element of the following product:
    \[
      \Prod{x:\vrt{\gA}}
      \mathsf{isProp}\,\Fan{\gA}{x}
    \]

    As being a proposition is a proposition, this is a product of propositions. Dependent function extensionality implies the closure of propositions under products.
  \end{proof}

  \begin{lemma}\label[lemma]{lem:disp-rx-gph-univalence-is-prop}
    Let $\gA$ be a reflexive graph, and let $\gB$ be a displayed reflexive graph over $\gA$. Assuming dependent function extensionality, it is a proposition that $\gB$ is univalent.
  \end{lemma}

  \begin{proof}
    Univalence of $\gB$ amounts to the following product:
    \[
      \Prod{x:\vrt{\gA}}
      \mathsf{isUnivalent}\,\prn{\gB\prn{x}}
    \]

    By dependent function extensionality, it suffices to check that each $\mathsf{isUnivalent}\,\prn{\gB\prn{x}}$ is a proposition; but this is \cref{lem:rx-gph-univalence-is-prop}.
  \end{proof}

  \begin{lemma}\label[lemma]{lem:covariant-lens-structure-is-prop}
    Let $\gA$ be a path object and let $\gB$ be a family of path objects over $\gA$. Assuming dependent function extensionality, the type of oplax covariant lens structures on $\gB$ over $\gA$ is a proposition.
  \end{lemma}

  \begin{proof}
    We first write out in detail the type of oplax covariant lens structrures on $\gB$:

    \iblock{
      \mrow{
        \Sum{
          \Phi:
          \DelimProtect{
            \Prod{x,y:\vrt{\gA}}
            \Prod{p:x\Edge{\gA}y}
            \vrt{\gB\prn{x}}\to \vrt{\gB\prn{y}}
          }
        }
        \Prod{x:\vrt{\gA}}
        \Prod{u:\vrt{\gB\prn{x}}}
        \Phi\,x\,x\,\prn{\Rx{\gA}{x}}\,u \Edge{\gB\prn{x}} u
      }
    }

    By the distributivity of products over sums, the above is equivalent to the following:

    \iblock{
      \mrow{
        \Prod{x:\vrt{\gA}}
        \Sum{
          \Phi :
          \DelimProtect{
            \Prod{y:\vrt{\gA}}
            \Prod{p:x\Edge{\gA}y}
            \vrt{\gB\prn{x}}\to \vrt{\gB\prn{y}}
          }
        }
        \Prod{u:\vrt{\gB\prn{x}}}
        \Phi\,x\,\prn{\Rx{\gA}{x}}\,u \Edge{\gB\prn{x}} u
      }
    }

    By dependent function extensionality, it suffices to assume that for each $x:\vrt{\gA}$ the following type is a proposition:

    \iblock{
      \mrow{
        \Sum{
          \Phi :
          \DelimProtect{
            \Prod{y:\vrt{\gA}}
            \Prod{p:x\Edge{\gA}y}
            \vrt{\gB\prn{x}}\to \vrt{\gB\prn{y}}
          }
        }
        \Prod{u:\vrt{\gB\prn{x}}}
        \Phi\,x\,\prn{\Rx{\gA}{x}}\,u \Edge{\gB\prn{x}} u
      }
    }

    Re-associating, the above is equivalent to:
    \iblock{
      \mrow{
        \Sum{
          \Phi :
          \DelimProtect{
            \Prod{\prn{y,p}:\CoFan{\gA}{x}}
            \vrt{\gB\prn{x}}\to \vrt{\gB\prn{y}}
          }
        }
        \Prod{u:\vrt{\gB\prn{x}}}
        \Phi\,\prn{x,\Rx{\gA}{x}}\,u \Edge{\gB\prn{x}} u
      }
    }

      By our assumption that $\gA$ is univalent, we know that $\CoFan{\gA}{x}$ is contractible with centre $\prn{x,\Rx{\gA}{x}}$. Therefore, the above is equivalent to the following simplified type:
      \[
        \Sum{
          \Phi :
          \DelimProtect{
            \vrt{\gB\prn{x}}\to \vrt{\gB\prn{x}}
          }
        }
        \Prod{u:\vrt{\gB\prn{x}}}
        \Phi\,u \Edge{\gB\prn{x}} u
      \]

      The above is definitionally equal to the co-fan of the identity function in $\vrt{\gB\prn{x}} \pitchfork \gB\prn{x}$ as we can see below:
      \begin{align*}
         & \Sum{
          \Phi :
          \DelimProtect{
            \vrt{\gB\prn{x}}\to \vrt{\gB\prn{x}}
          }
        }
        \Prod{u:\vrt{\gB\prn{x}}}
        \Phi\,u \Edge{\gB\prn{x}} u
        \\
         & \quad\equiv
        \Sum{
          \Phi :
          \DelimProtect{
            \vrt{\gB\prn{x}}\to \vrt{\gB\prn{x}}
          }
        }
        \Phi \Edge{\vrt{\gB\prn{x}} \pitchfork \gB\prn{x}} \prn{\Lam{u}{u}}
        \\
         & \quad\equiv
        \CoFan{\gB\prn{x}}{
          \Lam{u}{u}
        }
      \end{align*}

      This co-fan is indeed a proposition, as we have assumed that each $\gB\prn{x}$ is univalent.
  \end{proof}

  \begin{lemma}\label[lemma]{lem:contravariant-lens-structure-is-prop}
    Let $\gA$ be a path object and let $\gB$ be a family of path objects over $\gA$. Assuming dependent function extensionality, the type of lax contravariant lens structures on $\gB$ over $\gA$ is a proposition.
  \end{lemma}

  \begin{proof}
    Analogous to \cref{lem:covariant-lens-structure-is-prop}.
  \end{proof}
  
  \begin{lemma}
    Let $\gA$ be a path object, and let $\gB\prn{p}$ be a path object for each edge $p : x \Edge{\gA}y$ in $\gA$. Assuming dependent function extensionality, the type of unbiased dependent lens structures on $\gB$ over $\gA$ is a proposition.
  \end{lemma}

  \begin{proof}
    The type of unbiased dependent lens structures on $\gB$ over $\gA$ can be written as follows:

    \iblock{
      \mrow{
        \Sum{
          \Phi:
          \DelimProtect{
            \Prod{x,y:\vrt{\gA}}\Prod{p:x\Edge{\gA}y} \vrt{\gB\prn{x,x,\Rx{\gA}{x}}}\to \vrt{\gB}\prn{x,y,p}
          }
        }
      }
      \mrow{
        \Sum{
          \Psi:
          \DelimProtect{
            \Prod{x,y:\vrt{\gA}}\Prod{p:x\Edge{\gA}y} \vrt{\gB\prn{y,y,\Rx{\gA}{y}}}\to \vrt{\gB}\prn{x,y,p}
          }
        }
      }
      \mrow{
        \prn{
          \Prod{x:\vrt{\gA}}
          \Prod{u:\vrt{\gB\prn{x,x,\Rx{\gA}{x}}}}
          \Phi\,x\,x\,\prn{\Rx{\gA}{x}}\,u \Edge{\gB\prn{x,x,\Rx{\gA}{x}}} \Psi\,x\,x\,\prn{\Rx{\gA}{x}}\,u
        }
      }
      \mrow{
        {}\times
        \prn{
          \Prod{x:\vrt{\gA}}
          \Prod{u:\vrt{\gB\prn{x,x,\Rx{\gA}{x}}}}
          u \Edge{\gB\prn{x,x,\Rx{\gA}{x}}} \Psi\,x\,x\,\prn{\Rx{\gA}{x}}\,u
        }
      }
    }

    Considering the distributivity of products over sums, it suffices by dependent function extensionality to prove that for each $x:\vrt{\gA}$, the following type is a proposition:

    \iblock{
      \mrow{
        \Sum{
          \Phi:
          \DelimProtect{
            \Prod{y:\vrt{\gA}}\Prod{p:x\Edge{\gA}y} \vrt{\gB\prn{x,x,\Rx{\gA}{x}}}\to \vrt{\gB}\prn{x,y,p}
          }
        }
      }
      \mrow{
        \Sum{
          \Psi:
          \DelimProtect{
            \Prod{y:\vrt{\gA}}\Prod{p:x\Edge{\gA}y} \vrt{\gB\prn{y,y,\Rx{\gA}{y}}}\to \vrt{\gB}\prn{x,y,p}
          }
        }
      }
      \mrow{
        \prn{
          \Prod{u:\vrt{\gB\prn{x,x,\Rx{\gA}{x}}}}
          \Phi\,x\,x\,\prn{\Rx{\gA}{x}}\,u \Edge{\gB\prn{x,x,\Rx{\gA}{x}}} \Psi\,x\,x\,\prn{\Rx{\gA}{x}}\,u
        }
      }
      \mrow{
        {}\times
        \prn{
          \Prod{u:\vrt{\gB\prn{x,x,\Rx{\gA}{x}}}}
          u \Edge{\gB\prn{x,x,\Rx{\gA}{x}}} \Psi\,x\,x\,\prn{\Rx{\gA}{x}}\,u
        }
      }
    }

    We re-associate to expose various fans in $\gA$:

    \iblock{
      \mrow{
        \Sum{
          \Phi:
          \DelimProtect{
            \Prod{\prn{y,p}:\Fan{\gA}{x}}
            \vrt{\gB\prn{x,x,\Rx{\gA}{x}}}\to \vrt{\gB}\prn{x,y,p}
          }
        }
      }
      \mrow{
        \Sum{
          \Psi:
          \DelimProtect{
            \Prod{\prn{y,p}:\Fan{\gA}{x}}
            \vrt{\gB\prn{y,y,\Rx{\gA}{y}}}\to \vrt{\gB}\prn{x,y,p}
          }
        }
      }
      \mrow{
        \prn{
          \Prod{u:\vrt{\gB\prn{x,x,\Rx{\gA}{x}}}}
          \Phi\,x\,x\,\prn{\Rx{\gA}{x}}\,u \Edge{\gB\prn{x,x,\Rx{\gA}{x}}} \Psi\,x\,x\,\prn{\Rx{\gA}{x}}\,u
        }
      }
      \mrow{
        {}\times
        \prn{
          \Prod{u:\vrt{\gB\prn{x,x,\Rx{\gA}{x}}}}
          u \Edge{\gB\prn{x,x,\Rx{\gA}{x}}} \Psi\,x\,x\,\prn{\Rx{\gA}{x}}\,u
        }
      }
    }

    As $\gA$ is univalent, the fan $\Fan{\gA}{x}$ contracts onto $\prn{x,\Rx{\gA}{x}}$ as reflected below:

    \iblock{
      \mrow{
        \Sum{
          \Phi:
          \DelimProtect{
            \vrt{\gB\prn{x,x,\Rx{\gA}{x}}}\to \vrt{\gB\prn{x,x,\Rx{\gA}{x}}}
          }
        }
      }
      \mrow{
        \Sum{
          \Psi:
          \DelimProtect{
            \vrt{\gB\prn{x,x,\Rx{\gA}{x}}}\to \vrt{\gB\prn{x,x,\Rx{\gA}{x}}}
          }
        }
      }
      \mrow{
        \prn{
          \Prod{u:\vrt{\gB\prn{x,x,\Rx{\gA}{x}}}}
          \Phi\,u \Edge{\gB\prn{x,x,\Rx{\gA}{x}}} \Psi\,u
        }
      }
      \mrow{
        {}\times
        \prn{
          \Prod{u:\vrt{\gB\prn{x,x,\Rx{\gA}{x}}}}
          u \Edge{\gB\prn{x,x,\Rx{\gA}{x}}} \Psi\,u
        }
      }
    }

    The above is definitionally equivalent to the following:

    \iblock{
      \mrow{
        \Sum{
          \Phi:
          \DelimProtect{
            \vrt{\gB\prn{x,x,\Rx{\gA}{x}}}\to \vrt{\gB\prn{x,x,\Rx{\gA}{x}}}
          }
        }
      }
      \mrow{
        \Sum{
          \Psi:
          \DelimProtect{
            \vrt{\gB\prn{x,x,\Rx{\gA}{x}}}\to \vrt{\gB\prn{x,x,\Rx{\gA}{x}}}
          }
        }
      }
      \mrow{
        {}\times
        \prn{
          \Phi \Edge{\vrt{\gB\prn{x,x,\Rx{\gA}{x}}}\pitchfork \gB\prn{x,x,\Rx{\gA}{x}}} \Psi
        }
      }
      \mrow{
        {}\times 
        \prn{
          \prn{\Lam{u}{u}}\Edge{\vrt{\gB\prn{x,x,\Rx{\gA}{x}}}\pitchfork \gB\prn{x,x,\Rx{\gA}{x}}} \Psi
        }
      }

    }

    We can re-associate the above as follows:

    \iblock{
      \mrow{
        \Sum{
          \Phi:
          \DelimProtect{
            \vrt{\gB\prn{x,x,\Rx{\gA}{x}}}\to \vrt{\gB\prn{x,x,\Rx{\gA}{x}}}
          }
        }
      }
      \mrow{
        \Sum{
          \prn{\Psi,\hat\Psi}:
          \Fan{
            \vrt{\gB\prn{x,x,\Rx{\gA}{x}}}\pitchfork {\gB}\prn{x,x,\Rx{\gA}{x}}
          }{\Lam{u}u}
        }
      }
      \mrow{
        \Phi \Edge{\vrt{\gB\prn{x,x,\Rx{\gA}{x}}}\pitchfork \gB\prn{x,x,\Rx{\gA}{x}}} \Psi
      }
    }

    As $\gB\prn{x,x,\Rx{\gA}{a}}$ is assumed to be univalent, so is its cotensor by $\vrt{\gB\prn{x,x,\Rx{\gA}{x}}}$; thus, the fan of $\Lam{u}{u}$ in the cotensor contracts to $\prn{\Lam{u}{u}, \Lam{u}{\Rx{\gB\prn{x,x,\Rx{\gA}}}{u}}}$:

    \iblock{
      \mrow{
        \Sum{
          \Phi:
          \DelimProtect{
            \vrt{\gB\prn{x,x,\Rx{\gA}{x}}}\to \vrt{\gB\prn{x,x,\Rx{\gA}{x}}}
          }
        }
      }
      \mrow{
        \Phi \Edge{\vrt{\gB\prn{x,x,\Rx{\gA}{x}}}\pitchfork \gB\prn{x,x,\Rx{\gA}{x}}} \Lam{u}{u}
      }
    }

    The above is just the co-fan $\CoFan{\vrt{\gB\prn{x,x,\Rx{\gA}{x}}}\pitchfork {\gB}\prn{x,x,\Rx{\gA}{x}}n}{\Lam{u}u}$, which is a proposition because $\gB\prn{x,x,\Rx{\gA}{x}}$ is assumed univalent.
  \end{proof}
\end{xsect}
 \end{xsect}
\begin{xsect}[sec:classifying-rx-gphs]{Case study: path objects classifying reflexive graphs}
  In this section, we assume dependent function extensionality.
  Let $\prn{U,E}$ be a univalent family in the sense of \cref{def:univalent-family} so that the reflexive graph image $U/E$ is univalent. Given $A:U$, we will write $A:\TYPE$ in place of $E\prn{A}$ when it causes no confusion.

  \begin{xsect}{Classifying reflexive graphs}
    \begin{construction}[A univalent lens for graph structures]
      A $U$-small graph structure on $A:U$ is given by a vertex of the path structure $\GphOn{A} :\equiv A\pitchfork A\pitchfork U/E$. We shall equip $\GphOn$ as a \emph{lax contravariant lens of path objects} over $U/E$ as follows:

      \iblock{
        \mhang{
          A_0,A_1:U;
          f : \mathsf{Equiv}\prn{A_0,A_1},
          E : A_1\to A_1\to U
          \vdash
        }{
          \mrow{
            \Push{\GphOn}{f}{E} : A_0\to A_0\to U
          }
          \mrow{
            \Push{\GphOn}{f}{E}\,x\,y :\equiv E \prn{fx}\,\prn{fy}
          }
        }
        \row
        \mhang{
          A:U, E:A\to A\to U\vdash
        }{
          \mrow{
            \PushRx{\GphOn}[A]{E} : \Push{\GphOn}{\mathsf{idnEquiv}_A}{E} \Edge{\GphOn{A}} E
          }
          \mrow{
            \PushRx{\GphOn}[A]{E} :\equiv \Rx{\GphOn{A}}{E}
          }
        }
      }
    \end{construction}

    \begin{computation}[A path object classifying graphs]
      Thus we obtain a path object classifying $U$-small graphs by taking the total path object of the display of the lax contravariant lens $\GphOn$ over $U/E$, defining $\Gph :\equiv \prn{U/E}.\CtrvDisp{U/E}{\GphOn}$.
      Unraveling definitions, we see that its underlying reflexive graph structure is precisely the natural one that would arise from invertible graph homomorphisms:
      \begin{align*}
        \vrt{\Gph}
         & \equiv
        \Sum{A:U}{A\to A\to U}
        \\
        \gG_0 \Edge{\Gph} \gG_1
         & \equiv
        \Sum{
          f:\mathsf{Equiv}\prn{\vrt{\gG_0},\vrt{\gG_1}}
        }
        \Prod{x,y:\vrt{\gG_0}}
        \mathsf{Equiv}\prn{x\Edge{\gG_0}y,fx\Edge{\gG_1}fy}
        \\
        \Rx{\Gph}{\gG}
         & \equiv
        \prn{\mathsf{idnEquiv}_{\vrt{\gG}},\Lam{x,y}\mathsf{idnEquiv}_{x\Edge{\gG}y}}
      \end{align*}

      Given $f : \gG_0\Edge{\Gph}\gG_1$, we will write $\vrt{f}:\mathsf{Equiv}\prn{\vrt{\gG_0},\vrt{\gG_1}}$ for the first component and $f^\approx : \Prod{x,y:\vrt{\gG_0}}\mathsf{Equiv}\prn{x\Edge{\gG_0}y,\vrt{f}x\Edge{\gG_1}\vrt{f}y}$  for the second component.
    \end{computation}

    In order to classify reflexive graphs, we naturally wish to define a further lens over $\Gph$ whose components classify reflexivity data on a given graph. Unfortunately, due to the mixed variance of reflexivity data in the underlying graph, it seems we can find neither a suitable lax contravariant nor oplax covariant lens for reflexivity data. It happens, however, that the desirable displayed path object over $\Gph$ arises via \cref{lem:unb-disp-po} from a particularly simple \emph{unbiased dependent lens} that we now define straightaway.

    \begin{construction}[A univalent unbiased dependent lens for reflexivity data]
      A reflexivity datum over an \emph{equivalence} of $U$-small graphs $f : \gA \Edge{\Gph} \gB$ is given by an assignment of self-edges in the image of $f$; these are the vertices of the following family of path objects:

      \iblock{
        \mhang{
          \gA,\gB : \vrt{\Gph}; f : \gA\Edge{\Gph} \gB \vdash
        }{
          \mrow{
            \RxOn{f} :\equiv \Prod{x:\vrt{\gA}} \DiscPO\prn{\vrt{f}x\Edge{\gB}\vrt{f}x}
          }
        }
      }

      When instantiated to the reflexivity datum in $\Gph$ on a specific graph $\gA$, we recover the usual notion of reflexivity data up to definitional equality:
      \[
        \RxOn{\Rx{\Gph}{\gA}} \equiv \Prod{x:\vrt{\gA}} \DiscPO\prn{x\Edge{\gA}x}
      \]

      We now equip $\RxOn$ with the structure of a unbiased dependent lens over $\Gph$:

      \iblock{
        \mhang{
          \gA,\gB:\vrt{\Gph};
          f : \gA \Edge{\Gph} \gB;
          \Rx{\gA} : \vrt{\RxOn{\Rx{\Gph}{\gA}}}
          \vdash
        }{
          \mrow{
            \LJ{\RxOn}{f}{\Rx{\gA}} : \vrt{\RxOn{f}}
          }
          \mrow{
            \LJ{\RxOn}{f}{\Rx{\gA}} :\equiv \Lam{x:\vrt{\gA}} f^\approx\,x\,x\,\prn{\Rx{\gA}{x}}
          }
        }

        \row

        \mhang{
          \gA,\gB:\vrt{\Gph};
          f : \gA \Edge{\Gph} \gB;
          \Rx{\gB} : \vrt{\RxOn{\Rx{\Gph}{\gB}}}
          \vdash
        }{
          \mrow{
            \RJ{\RxOn}{f}{\Rx{\gB}} : \vrt{\RxOn{f}}
          }
          \mrow{
            \RJ{\RxOn}{f}{\Rx{\gB}} :\equiv \Lam{x:\vrt{\gA}} \Rx{\gB}\prn{\vrt{f}x}
          }
        }
      }

      The laws of the dependent lens hold definitionally, so we are done!
    \end{construction}

    \begin{computation}[A path object classifying reflexive graphs]
      A path object classifying $U$-small reflexive graphs is obtained from the display of the unbiased dependent lens $\RxOn$ over $\Gph$ by setting $\RxGph :\equiv \Gph.\UnbDisp{\Gph}{\RxOn}$. It is not difficult to see that the underlying reflexive graph of $\RxGph$ is precisely the desired one:
      \begin{align*}
        \vrt{\RxGph}
         & \equiv
        \Sum{\gG:\vrt{\Gph}}
        \Prod{x:\vrt{\gG}} x \Edge{\gG} x
        \\
        \gA_0 \Edge{\RxGph} \gA_1
         & \equiv
        \Sum{f : \gA_0\Edge{\Gph}\gA_1}
        \Prod{x:\vrt{\gA_0}}
        \Id{\vrt{f}x\Edge{\gA_1}\vrt{f}x}{
          f^\approx\,x\,x\,\prn{\Rx{\gA_0}{x}}
        }{
          \Rx{\gA_1}\prn{\vrt{f}x}
        }
        \\
        \Rx{\RxGph}{\gA}
         & \equiv
        \prn{\Rx{\Gph}{\gA}, \Lam{x:\vrt{\gA}}\Refl}
      \end{align*}

      Given $f : \gA_0 \Edge{\RxGph}\gA_1$,  we will write $f : \gA_0\Edge{\Gph}\gA_1$ for the first component and $f^\mathsf{rx} : \Prod{x:\vrt{\gA_0}}
        \Id{\vrt{f}x\Edge{\gA_1}\vrt{f}x}{
          f^\approx\,x\,x\,\prn{\Rx{\gA_0}{x}}
        }{
          \Rx{\gA_1}\prn{\vrt{f}x}
        }$ for the second component.
    \end{computation}
  \end{xsect}

  \begin{xsect}[sec:drxgph-sip]{Classifying displayed reflexive graphs over a fixed base}
    In this section, let $\gA$ be a fixed reflexive graph.

    \begin{construction}[A lens for displayed graph structures over a fixed family of displayed vertices]
      We shall define a lax contravariant lens over $\vrt{\gA}\pitchfork U/E$ classifying displayed graph structures on a fixed family of types representing displayed vertices.

      \iblock{
        \mrow{
          B:\vrt{\gA}\to U\vdash
          \DGphOn{\gA}{B} :\equiv
          \Prod{x,y:\vrt{\gA}}
          \prn{x\Edge{\gA}y}
          \pitchfork
          Bx\pitchfork By \pitchfork U/E
        }
      }

      We equip the family of path objects above with the structure of a definitional contravariant lens over $\vrt{\gA}\pitchfork U/E$ as follows:

      \iblock{
        \mhang{
          B_0,B_1:\vrt{\gA}\to U; f : B_0\Edge{\vrt{\gA}\pitchfork U/E} B_1; E : \vrt{\DGphOn{\gA}{B_1}} \vdash
        }{
          \mrow{
            \Pull{\DGphOn{\gA}}{f}{E} : \vrt{\DGphOn{\gA}{B_0}}
          }
          \mrow{
            \Pull{\DGphOn{\gA}}{f}{E} :\equiv \Lam{x,y,p,u,v} E\,x\,y\,p\,\prn{f xu}\,\prn{fyv}
          }
        }
      }
    \end{construction}

    \begin{computation}[A path object for displayed graph structures over $\gA$]
      From the display of $\DGphOn{\gA}$ over $\vrt{\gA}\pitchfork U/E$, we obtain a suitable path object $\DGph{\gA} :\equiv \prn{\vrt{\gA}\pitchfork U/E}.\CtrvDisp{\vrt{\gA}\pitchfork U/E}{\DGphOn{\gA}}$ classifying displayed graph structures over $\gA$.

      \iblock{
      \mrow{
        \vrt{\DGph{\gA}}
        \equiv
        \Sum{B:\vrt{\gA}\to U}
        \Prod{x,y:\vrt{\gA}}
        \prn{x\Edge{\gA}y} \pitchfork Bx\pitchfork By \pitchfork U/E
      }
      \mhang{
        \gB_0 \Edge{\DGph{\gA}} \gB_1 \equiv
      }{
        \mrow{
          \Sum{
            \DelimProtect{
              f : \Prod{x:A}\mathsf{Equiv}\prn{\vrt{\gB_0}x,\vrt{\gB_1}x}
            }
          }
        }
        \mrow{
          \Prod{x,y:\vrt{\gA}}
          \Prod{p:x\Edge{\gA}y}
          \Prod{u:\vrt{\gB_0}x}
          \Prod{v:\vrt{\gB_0}y}
        }
        \mrow{
          \mathsf{Equiv}\prn{
            u \Edge{\gB_0}[p] v,
            fxu \Edge{\gB_1}[p] fxv
          }
        }
      }
      \mrow{
      \Rx{\DGph{\gA}}{\gB}
      \equiv
      \prn{
      \Lam{x}\mathsf{idnEquiv}_{\vrt{\gB}x},
      \Lam{x,y,p,u,v}\mathsf{idnEquiv}_{u\Edge{\gB}[p]v}
      }
      }
      }
    \end{computation}

    \begin{construction}[A univalent unbiased dependent lens for displayed reflexivity data]
      Next, we exhibit a unbiased dependent lens classifying displayed reflexivity data on a given displayed graph over $\gA$. In particular, we specify a family of path objects $\DRxOn{\gA}$ over the edges of $\DGph{\gA}$ in the following way:

      \iblock{
        \mhang{
          \gB_0,\gB_1:\vrt{\DGph{\gA}};
          f:\gB_0\Edge{\DGph{\gA}}\gB_1
          \vdash
        }{
          \mrow{
            \DRxOn{\gA}{f}
            :\equiv
            \Prod{x:\vrt{\gA}}
            \Prod{y:\vrt{\gB_0}x}
            \DiscPO\prn{\vrt{f}xy\Edge{\gB_1}[\Rx{\gA}{x}]\vrt{f}xy}
          }
        }
      }

      Along the diagonal we recover the desired notion of displayed reflexivity data:
      \[
        \DRxOn{\gA}{\Rx{\DGph{\gA}}{\gB}} \equiv
        \Prod{x:\vrt{\gA}}
        \Prod{y:\vrt{\gB}x}
        \DiscPO\prn{y \Edge{\gB}[\Rx{\gA}{x}] y}
      \]

      The dependent lens structure is given in such a way that the laws hold definitionally:

      \iblock{
        \mhang{
          \gB_0,\gB_1:\vrt{\DGph{\gA}}; f : \gB_0\Edge{\DGph{\gA}}\gB_1;
          \Rx{\gB_0} : \vrt{\DRxOn{\gA}{\Rx{\DGph{\gA}}{\gB_0}}}
          \vdash
        }{
          \mrow{
            \LJ{\DRxOn{\gA}}{f}{\Rx{\gB_0}} : \vrt{\DRxOn{\gA}{f}}
          }
          \mrow{
            \LJ{\DRxOn{\gA}}{f}\Rx{\gB_0} :\equiv
            \Lam{x,y} f^\approx \,x\,x\,\prn{\DRx{\gA}{x}}\,y\,y\,\prn{\Rx{\gB_0}{x}{y}}
          }
        }
        \row
        \mhang{
          \gB_0,\gB_1:\vrt{\DGph{\gA}}; f : \gB_0\Edge{\DGph{\gA}}\gB_1;
          \Rx{\gB_1} : \vrt{\DRxOn{\gA}{\Rx{\DGph{\gA}}{\gB_1}}}
          \vdash
        }{
          \mrow{
            \RJ{\DRxOn{\gA}}{f}{\Rx{\gB_1}} : \vrt{\DRxOn{\gA}{f}}
          }
          \mrow{
            \RJ{\DRxOn{\gA}}{f}\Rx{\gB_1} :\equiv
            \Lam{x,y}\DRx{\gB_1}{x}\prn{\vrt{f}xy}
          }
        }
      }
    \end{construction}

    \begin{computation}
      From the display of $\DRxOn{\gA}$ over $\DGph{\gA}$, we obtain a path object classifying displayed reflexive graphs over $\gA$, defining $\DRxGphOver{\gA}$ as follows:
      \[
        \DRxGphOver{\gA} :\equiv
        \DGph{\gA}.\UnbDisp{\DGph{\gA}}{\DRxOn{\gA}}
      \]

      Unraveling definitions, we see that the underlying reflexive graph $\DRxGphOver{\gA}$ is precisely the desirable one classifying displayed reflexive graphs over $\gA$:

      \iblock{
        \mrow{
          \vrt{\DRxGphOver{\gA}} \equiv
          \Sum{\gB:\DGph{\gA}}
          \Prod{x:\vrt{\gA}}\Prod{u:Bx} u \Edge{\gB}[\Rx{\gA}{x}] u
        }
        \mhang{
          \gB_0 \Edge{\DRxGphOver{\gA}} \gB_1 \equiv
        }{
          \mrow{
            \Sum{f:\gB_0\Edge{\DGph{\gA}}\gB_1}
            \Prod{x:\vrt{\gA}}
            \Prod{u:\vrt{\gB_0}}
          }
          \mrow{
            \Id{
              \vrt{f}xu \Edge{\gB_1}[\Rx{\gA}{x}] \vrt{f}xu
            }{
              f^\approx\,x\,x\,\prn{\DRx{\gA}{x}}\,y\,y\,\prn{\DRx{\gB_0}{x}{u}}
            }{
              \DRx{\gB_1}{x}\prn{\vrt{f}xu}
            }
          }
        }
        \mrow{
          \Rx{\DRxGphOver{\gA}}{\gB} \equiv
          \prn{
            \prn{
              \Lam{x}\mathsf{idnEquiv}\Sub{\vrt{\gB}x},
              \Lam{x,y,p,u,v}\mathsf{idnEquiv}\Sub{u\Edge{\gB}[p]v}
            },
            \Lam{x,u}\Refl
          }
        }
      }
    \end{computation}

  \end{xsect}

  \begin{xsect}{Classifying displayed reflexive graphs with variable base}

    \begin{construction}[A lax contravariant lens for $U$-small displayed reflexive graphs over a given reflexive graph]
      We can extend $\DRxGphOver$ with the structure of a definitional contravariant lens over $\RxGph$. We first describe how to restrict a displayed reflexive graph $\gB$ along a reflexive graph equivalence $f : \gA_0\Edge{\RxGph}\gA_1$; on vertices and edges, one simply uses the restriction of the corresponding components of $\gB$ along $f$:

      \iblock{
        \mhang{
          \gA_0,\gA_1:\vrt{\RxGph};
          f : \gA_0 \Edge{\RxGph} \gA_1,
          \gB : \vrt{\DRxGphOver{\gA_1}}
          \vdash
        }{
          \mrow{
            \Pull{\DRxGphOver}{f}{\gB} : \vrt{\DRxGphOver{\gA_0}}
          }
          \mrow{
            \vrt{\Pull{\DRxGphOver}{f}{\gB}}x :\equiv
            \vrt{\gB}\prn{fx}
          }
          \mrow{
            u \Edge{\Pull{\DRxGphOver}{f}{\gB}}[p : x\Edge{\gA}y] v
            :\equiv
            {u} \Edge{\gB}[f^\approx xyp] {v}
          }
        }
      }

      To define the restricted reflexivity datum, we must transport along the witness that $f$ preserves reflexivity data.

      \iblock[1]{
        \mrow{
          \DRx{\Pull{\DRxGphOver}{f}{\gB}}{x}{u} :
          u \Edge{\gB}[f^\approx\,x\,x\,\prn{\Rx{\gA}{x}}] u
        }
        \mrow{
          \DRx{\Pull{\DRxGphOver}{f}{\gB}}{x}{u} :\equiv
          \prn{
            f^{\mathsf{rx}}x
          }_*\Sup{u \Edge{\gB}[\bullet] u}\prn{
            \DelimMin{1}
            \DRx{\gB}{fx}{u}
          }
        }
      }

      The lens is definitional because $\gB \equiv \Pull{\DRxGphOver}{\Rx{\RxGph}{\gA}}\gB$ definitionally.
    \end{construction}

    \begin{computation}[A path object classifying $U$-small displayed reflexive graphs over a variable base]
      By computing the display of $\DRxGphOver$ over $\RxGph$, we obtain a path object classifying reflexive graphs $\gA$ equipped with a displayed reflexive graph $\gB$ over $\gA$.
      \[
        \DRxGph :\equiv \RxGph.\CtrvDisp{\RxGph}{\DRxGphOver}
      \]

      We compute as follows.

      \iblock{
        \mrow{
          \vrt{\DRxGph}
          \equiv
          \Sum{\gA:\vrt{\RxGph}}{
            \vrt{\DRxGphOver{\gA}}
          }
        }
        \mhang{
          \prn{\gA_0,\gB_0} \Edge{\DRxGph}
          \prn{\gA_1,\gB_1} \equiv
        }{
          \mrow{
            \Sum{f : \gA_0 \Edge{\RxGph} \gA_1}
          }
          \mrow{
            \Sum{
              \vrt{g} :
              \DelimProtect{
                \Prod{x:\vrt{\gA_0}}
                \mathsf{Equiv}\prn{
                  \vrt{\gB_0}x,
                  \vrt{\gB_1}\prn{fx}
                }
              }
            }
          }
          \mrow{
            \Sum{
              g^\approx :
              \DelimProtect{
                \Prod{x,y:\vrt{\gA_0}}
                \Prod{p:x\Edge{\gA_0}y}
                \Prod{u:\vrt{\gB_0}x}
                \Prod{v:\vrt{\gB_1}y}
                \mathsf{Equiv}\prn{
                  u \Edge{\gB_0}[p]v,
                  u \Edge{\gB_1}[f^\approx xyp]v
                }
              }
            }
          }
          \mrow{
            \Prod{x:\vrt{\gA_0}}
            \Prod{u:\vrt{\gB_0}}
            \Id{
              u\Edge{\gB_0}[\Rx{\gA}{x}] u
            }{
              \DRx{\gB_0}{x}{u}
            }{
              \prn{
                f^{\mathsf{rx}}x
              }_*\Sup{u \Edge{\gB}[\bullet] u}\prn{
                \DelimMin{1}
                \DRx{\gB_1}{fx}{u}
              }
            }
          }
        }
        \mrow{
          \Rx{\DRxGph}\prn{\gA,\gB} \equiv
          \prn{
            \Rx{\RxGph}{\gA},
            \Rx{\DRxGphOver{\gA}}{\gB}
          }
        }
      }
    \end{computation}
  \end{xsect}

  \begin{xsect}{Classifying families and lenses of reflexive graphs}

    \begin{construction}[A definitional covariant lens for vertices]
      We can define a definitional covariant lens $\Vtx$ over the path object $\RxGph$ whose component over a $U$-small reflexive graph $\gA$ is $\gA$ itself. Setting $\Vtx\prn{\gA} :\equiv \gA$, we define the definitional lens structure as follows:

      \iblock{
        \mhang{
          \gA_0,\gA_1:\vrt{\RxGph};
          f:\gA_0\Edge{\RxGph}\gA_1;
          x : \vrt{\gA}
        }{
          \mrow{
            \vdash
            \Push{\Vtx}{f}{x} :\equiv \vrt{f} x
          }
        }
      }
    \end{construction}

    \NewDocumentCommand\RxGphFam{}{\mathsf{RxGphFam}_U}

    \begin{computation}[A path object classifying $U$-small families of $U$-small reflexive graphs]
      Using our theory of partial products of lenses (\cref{sec:polynomials}), we can define a path object that classifies pairs $\prn{\gA, \gB}$ where $\gA$ is a $U$-small reflexive graph and $\gB$ is a family of $U$-small reflexive graphs indexed in the vertices of $\gA$.
      \[
        \RxGphFam :\equiv
        \CovPoly{\RxGph}{\Vtx}{\RxGph}
      \]

      We have the following computation:

      \iblock{
        \mrow{
          \vrt{\RxGphFam} \equiv
          \Sum{\gA:\vrt{\RxGph}}
          \Prod{x:\vrt{\gA}}
          \vrt{\RxGph}
        }
        \mhang{
          \prn{\gA_0,\gB_0} \Edge{\RxGphFam} \prn{\gA_1,\gB_1} \equiv
        }{
          \mrow{
            \Sum{f : \gA_0\Edge{\RxGph}\gA_1}\Prod{x:\vrt{\gA_0}}
            \gB_0x \Edge{\RxGph} \gB_1\prn{\vrt{f}x}
          }
        }
        \mrow{
          \Rx{\RxGphFam}\prn{\gA,\gB} \equiv
          \prn{\Rx{\RxGph}{\gA},\Lam{x}\Rx{\RxGph}{\gB\prn{x}}}
        }
      }
    \end{computation}

    It would be desirable to find a good path object classifying lens structures on a given family of reflexive graphs. Unfortunately, there are some limitations on our ability to do this: although the pushforward/pullback data are easy to classify, it is unclear how to define a notion of path between oplax/lax unitors \emph{over} a homotopy between pushforward/pullback data. Therefore, the best we can do is impose the \emph{discrete} reflexive graph structure on the first component of a lens, combining it with a family of path object structures classifying unitors by coproduct (\cref{ex:coprod-rx-gph,lem:coprod-po}); on the other hand, we can still exhibit a useful displayed reflexive graph structure classifying lenses with variable base by means of a unbiased dependent lens as we shall see in \cref{con:lens-of-lenses} below.

    \begin{construction}[A definitional unbiased dependent lens classifying oplax covariant lens structure]\label[construction]{con:lens-of-lenses}
      Let $\gA$ be a fixed reflexive graph. We define a componentwise discrete definitional unbiased dependent lens $\CovLensStr{\gA}$ over $\vrt{\gA}\pitchfork\RxGph$ as follows:

      \iblock{
        \mhang{
          f:\gB_0\Edge{\vrt{\gA}\pitchfork\RxGph}\gB_1 \vdash
          \CovLensStr{\gA}{f} :\equiv
        }{
          \mrow{
            \Coprod{
              \DelimProtect{
                \Phi : \Prod{x,y:\vrt{\gA}}\Prod{p:x\Edge{\gA}y}\vrt{\gB_0x} \to \vrt{\gB_1y}
              }
            }
          }
          \mrow{
            \Prod{x:\vrt{\gA}}\Prod{u:\vrt{\gB_0x}}
            \Phi\,x\,x\,\prn{\Rx{\gA}{x}}\,u \Edge{\gB_1 x} \vrt{f x}u
          }
        }
      }

      Naturally, $\CovLensStr{\gA}{\Rx{\RxGph}{\gB}}$ is precisely the type of oplax covariant lens structures on family of $U$-small reflexive graphs $\gB$ over $\gA$.

      \iblock{
        \mhang{
          f:\gB_0\Edge{\vrt{\gA}\pitchfork\RxGph}\gB_1; \prn{\Phi,\hat\Phi} : \CovLensStr{\gA}{\Rx{\RxGph}{\gB_0}}\vdash
        }{
          \mrow{
            \pi_1\prn{\LJ{\CovLensStr{\gA}}{f}\prn{\Phi,\hat\Phi}} :\equiv \Lam{x,y,p,u}
            \vrt{f y}\prn{\Phi xypu}
          }
          \mrow{
            \pi_2\prn{\LJ{\CovLensStr{\gA}}{f}\prn{\Phi,\hat\Phi}} :\equiv
            \Lam{x,u}
            \prn{fx}^\approx \,\prn{\Phi\,x\,x\,\prn{\Rx{\gA}{x}}\,u}\,u\,\prn{\hat\Phi xu}
          }
        }
        \mhang{
          f:\gB_0\Edge{\vrt{\gA}\pitchfork\RxGph}\gB_1; \prn{\Phi,\hat\Phi} : \CovLensStr{\gA}{\Rx{\RxGph}{\gB_1}}\vdash
        }{
          \mrow{
            \pi_1\prn{\RJ{\CovLensStr{\gA}}{f}\prn{\Phi,\hat\Phi}} :\equiv
            \Lam{x,y,p,u:\vrt{\gB_0x}}
            \Phi\,x\,y\,p\,\prn{\vrt{fx}u}
          }
          \mrow{
            \pi_2\prn{\RJ{\CovLensStr{\gA}}{f}\prn{\Phi,\hat\Phi}} :\equiv
            \Lam{x,u}
            \hat\Phi\,x\,\prn{\vrt{fx}u}
          }
        }
      }

      The unit laws of this unbiased dependent lens hold definitionally.
    \end{construction}

    \begin{computation}[A path object classifying oplax covariant lenses with fixed base]\label[computation]{cmp:lens-po}
      Finally, we may define path object classifying oplax covariant lenses over a fixed reflexive graph $\gA$.
      \[
        \CovLensOver{\gA} :\equiv
        \prn{\vrt{\gA}\pitchfork \RxGph}.%
        \UnbDisp{\vrt{\gA}\pitchfork \RxGph}{
          \CovLensStr{\gA}
        }
      \]
    \end{computation}
  \end{xsect}

\end{xsect}
\begin{xsect}[sec:fibrations]{Fibred reflexive graphs}

  A lens is kind of like a fibration in which the lifts have no universal property, as noted by \citet{chollet-et-al:2022:lenses}. In this section, we make the analogy between lenses and fibrations of reflexive graphs precise by characterising fibrations \emph{precisely} as lenses of path objects in \cref{cor:characterisation-of-fibs}.
  We begin by translating the notion of \emph{fibred} reflexive graphs from \citet{rijke:2019} into the language of \emph{displayed} reflexive graphs.

  \begin{definition}
    A \DefEmph{covariant fibration} over a reflexive graph $\gA$ is defined to be a displayed reflexive graph $\gB$ over $\gA$ such that for every edge $p:x\Edge{\gA}y$ and displayed vertex $u:\vrt{\gB}\prn{x}$, the sum $\Sum{v:\vrt{\gB}\prn{y}}u\Edge{\gB}[p]v$ is contractible.

    We will write $p_*^\gB u:\vrt{\gB}\prn{y}$ and $p_\dagger^\gB u : u \Edge{\gB}[p] p_*^\gB u$ for the first and second components component of the centre of contraction respectively.
  \end{definition}

  \begin{definition}
    A \DefEmph{contravariant fibration} over a reflexive graph $\gA$ is defined to be a displayed reflexive graph $\gB$ over $\gA$ such that for every edge $p:x\Edge{\gA}y$ and displayed vertex $v:\vrt{\gB}\prn{y}$, the sum $\Sum{v':\vrt{\gB}\prn{x}}v'\Edge{\gB}[p]v$ is contractible.

    We shall write $p^*_\gB v:\vrt{\gB}\prn{x}$ and $p^\dagger_\gB v : p^*_\gB v \Edge{\gB}[p] v$ for the first and second components of the centre of contraction respectively.
  \end{definition}

  \begin{xsect}[sec:fibration-duality]{Duality involution for fibred reflexive graphs}
    In this section, we show that covariant and contravariant fibrations are interchangeable via the duality involution for displayed reflexive graphs (\cref{sec:rx-gph-duality}).

    \begin{lemma}[Total opposite of a fibration]\label[lemma]{lem:fibration-duality}
      Let $\gB$ be a covariant (\resp contravariant) fibration of reflexive graphs over a reflexive graph $\gA$. Then the total opposite displayed reflexive graph $\gB\TotOp$ is a contravariant (\resp covariant) fibration over $\gA\Op$.
    \end{lemma}

    \cref{lem:fibration-duality} allows us to translate results about all covariant fibrations into results about all contravariant fibrations.
  \end{xsect}

  \begin{xsect}{Fibred reflexive graphs and univalence}

    Fibred reflexive graphs are always univalent, as we see in \cref{lem:fibrations-are-univalent}. Therefore, it is natural to refer to a fibration of reflexive graphs as a fibration of path objects.

    \begin{lemma}[Fibrations of reflexive graphs are univalent]\label[lemma]{lem:fibrations-are-univalent}
      Any (covariant, contravariant) fibration of reflexive graphs over a reflexive graph $\gA$ is univalent, \ie a displayed path object.
    \end{lemma}

    \begin{proof}
      Suppose without loss of generality (by \cref{lem:fibration-duality}) that $\gB$ is a covariant fibration. Fixing $x:\vrt{\gA}$ and $u:\vrt{\gB}\prn{x}$, it suffices to show that the fan $\Fan{\gB\prn{x}}{u}$ is a proposition. We have the following definitional computation of the fan:

      \iblock{
        \mhang{
          \Fan{\gB\prn{x}}{u}
        }{
          \commentrow{by \cref{def:fan}}
          \mrow{
            {}\equiv
            \Sum{w:\vrt{\gB}\prn{x}}
            u \Edge{\gB\prn{x}} w
          }
          \commentrow{by \cref{def:component}}
          \mrow{
            {}\equiv
            \Sum{w:\vrt{\gB}\prn{x}}
            u \Edge{\gB}[\Rx{\gA}{x}] w
          }
        }
      }

      The latter is contractible by our assumption that $\gB$ is a covariant fibration. Suppose, on the other hand, that $\gB$ is a contravariant fibration. To show that $\gB$ is univalent, it suffices to show that each of the co-fans $\CoFan{\gB\prn{x}}{u}$ is a proposition, which we compute below:

      \iblock{
        \mhang{
          \CoFan{\gB\prn{x}}{u}
        }{
          \commentrow{by \cref{def:fan}}
          \mrow{
            {}\equiv
            \Sum{w:\vrt{\gB}\prn{x}}
            w\Edge{\gB\prn{x}} u
          }
          \commentrow{by \cref{def:component}}
          \mrow{
            {}\equiv
            \Sum{w:\vrt{\gB}\prn{x}}
            w\Edge{\gB}[\Rx{\gA}{x}] u
          }
        }
      }

      The latter is contractible because $\gB$ is a contravariant assumption.
    \end{proof}

    A converse to \cref{lem:fibrations-are-univalent} cannot hold, as we see in \cref{ex:not-every-univalent-displayed-reflexive-graph-is-a-fibration}.

    \begin{example}\label[example]{ex:not-every-univalent-displayed-reflexive-graph-is-a-fibration}
      \emph{Not every univalent displayed reflexive graph is a fibration.} In particular, consider the following displayed reflexive graph over the codiscrete reflexive graph $\CodiscPO{\mathbf{2}}$:
      \iblock{
        \mrow{
          i:\mathbf{2} \vdash \vrt{\gG}\prn{i} :\equiv \Id{\mathbf{2}}{i}{1}
        }
        \mrow{
          i,j:\mathbf{2};p:\Id{\mathbf{2}}{i}{1}, q : \Id{\mathbf{2}}{i}{1} \vdash
          p\Edge{\gG}[*]q :\equiv \mathbf{1}
        }
        \mrow{
          i:\mathbf{2},p:\Id{\mathbf{2}}{i}{1} \vdash \DRx{\gG}{i}{p} :\equiv *
        }
      }

      We can see that for each $i:\mathbf{2}$, the component $\gG\prn{i}$ is the codiscrete reflexive graph $\CodiscPO\prn{\Id{\mathbf{2}}{i}{1}}$, which is univalent by \cref{lem:codisc-po} as $\mathbf{2}$ is a set. Therefore $\gG$ is univalent, and we will see that it can be neither a covariant nor a contravariant fibration.

      Suppose that $\gG$ were a covariant fibration. We have an edge ${*} : 1 \Edge{\CodiscPO{\mathbf{2}}} 0$ and thus a pushforward map $\prn{*}_* \colon \Id{\mathbf{2}}{1}{1}\to \Id{\mathbf{2}}{0}{1}$. On the other hand, if $\gG$ were a contravariant fibration, we would use pullback along ${*}:0 \Edge{\CodiscPO{\mathbf{2}}}1$ to obtain a contradiction.
    \end{example}

    \cref{ex:not-every-univalent-displayed-reflexive-graph-is-a-fibration} also shows that not every univalent displayed reflexive graph can be obtained from a lens structure on its components.
  \end{xsect}

\begin{xsect}[sec:fibrations-from-lenses]{Fibred reflexive graphs from lenses of path objects}

  We will now show that when a lens is valued in path objects, its display is a fibration.

  \begin{definition}[Universal pushforwards/pullbacks]
    An oplax covariant lens $\gB$ over a reflexive graph $\gA$ is said to have \DefEmph{universal pullforwards} when for every edge $p:x\Edge{\gA}y$ and vertex $u:\vrt{\gB\prn{x}}$, the fan $\Fan{\gB\prn{y}}{\Push{\gB}{p}{u}}$ is a proposition.

    Conversely, a lax contravariant lens $\gB$ over a reflexive graph $\gA$ is said to have \DefEmph{universal pullbacks} when for every edge $p:x\Edge{\gA}y$ and vertex $v:\vrt{\gB\prn{y}}$, the co-fan $\CoFan{\gB\prn{x}}{\Pull{\gB}{p}{u}}$ is a proposition.
  \end{definition}

  \begin{lemma}[Covariant fibrations from universal pushforwards]\label[lemma]{lem:cov-lens-to-fibration}
    Let $\gB$ be an oplax covariant lens over a reflexive graph $\gA$ that has universal pushforwards. Then the covariant display $\CovDisp{\gA}{\gB}$ can be made into a covariant fibration of reflexive graphs in which $p_* u \equiv \Push{\gB}{p}u$ and $p_\dagger u \equiv \Rx{\gB\prn{y}}\prn{\Push{\gB}{p}{u}}$.
  \end{lemma}
  \begin{proof}
    We note that for any edge $p:x\Edge{\gA}y$ and vertex $u:\vrt{\gB\prn{x}}$, we have the following sequence of definitional equivalences:

    \iblock{
      \mhang{
        \Sum{u':\vrt{\CovDisp{\gA}{\gB}}} u\Edge{\CovDisp{\gA}{\gB}}[p] u'
      }{
        \commentrow{by \cref{def:cov-disp}}
        \mrow{
          {}\equiv
          \Sum{u':\vrt{\gB\prn{y}}}
          \Push{\gB}{p}{u} \Edge{\gB\prn{y}} u'
        }
        \commentrow{by \cref{def:fan}}
        \mrow{
          {}\equiv
          \Fan{\gB\prn{y}}{\Push{\gB}{p}{u}}
        }
      }
    }

    Thus, we may choose a centre of contraction for the latter freely so long as it is a proposition (which we have by assumption).
  \end{proof}

  \begin{lemma}[Contravariant fibrations from universal pullbacks]\label[lemma]{lem:ctrv-lens-to-fibration}
    Let $\gB$ be a lax contravariant lens over a reflexive graph $\gA$ that has universal pullbacks. Then the contravariant display $\CtrvDisp{\gA}{\gB}$ can be made into a contravariant fibration of reflexive graphs in which $p^*u \equiv \Pull{\gB}{p}{u}$ and $p^\dagger u \equiv \Rx{\gB\prn{x}}\prn{\Pull{\gB}{p}{u}}$.
  \end{lemma}

  \begin{proof}
    Analogous to \cref{lem:cov-lens-to-fibration}.
  \end{proof}

  \begin{corollary}\label[corollary]{cor:fibrations-are-univalent}
    An (oplax covariant, lax contravariant) lens $\gB$ has universal (pushforwards, pullbacks) if and only if each component $\gB\prn{x}$ is univalent.
  \end{corollary}

  \begin{proof}
    Let $\gB$ be an oplax covariant lens over a reflexive graph $\gA$.

    \begin{enumerate}
      \item Suppose that $\gB$ has universal pushforwards. We want to show that each fan $\Fan{\gB\prn{x}}{u}$ is a proposition; using our assumption of universal pushforwards, we can show that $\Fan{\gB\prn{x}}{u}$ is a retract of $\Fan{\gB\prn{x}}{\Push{\gB}{\Rx{\gA}{x}}{u}}$, which is a proposition because $\gB$ has universal pushforwards.
      \item Suppose conversely that each component $\gB\prn{x}$ is univalent; then \emph{all} fans in the components of $\gB$ are propositions, including $\Fan{\gB\prn{y}}{\Push{\gB}{p}{u}}$ for some $p:x\Edge{\gA}y$ and $u:\vrt{\gB\prn{y}}$.
    \end{enumerate}

    The case for when $\gB$ is a lax contravariant lens is established by duality via \cref{lem:fibrations-are-univalent,lem:ctrv-lens-to-fibration}.
  \end{proof}
  
\end{xsect}

\NewDocumentCommand\Str{mo}{\mathsf{str}\Sub{#1}\IfValueT{#2}{\Sup{#2}}}
\NewDocumentCommand\StrGen{mo}{\mathsf{strGen}\Sub{#1}\IfValueT{#2}{\Sup{#2}}}
\NewDocumentCommand\Unstr{mo}{\mathsf{unstr}\Sub{#1}\IfValueT{#2}{\Sup{#2}}}
\NewDocumentCommand\UnstrGen{mo}{\mathsf{unstrGen}\Sub{#1}\IfValueT{#2}{\Sup{#2}}}

\begin{xsect}{The underlying lens of a covariant fibration of reflexive graphs}

  The pushforward operation $p_*^\gB$ of a covariant fibration gives rise to an oplax covariant lens structure on the components of $\gB$, which we construct and study in this section.

  \begin{enumerate}
    \item We first exhibit a ``straightening/unstraightening'' equivalence between displayed edges $u\Edge{\gB}[p:x\Edge{\gA}y]v$ and vertical edges $\Push{\gB}{p}{u}\Edge{\gB\prn{y}} v$ in \cref{con:straightening,con:unstraightening,lem:cov-str-equiv}.
    \item We use the straightening equivalence to exhibit the underlying lens of a displayed reflexive graph in \cref{con:underlying-cov-lens}.
    \item Finally, \cref{thm:underlying-lens-of-display-roundtrip} exhibits an equivalence of displayed reflexive graphs between any displayed reflexive graph $\gB$ and the display of its underlying lens.
  \end{enumerate}

  \begin{xsect}{Straightening of displayed edges}

    \begin{construction}[Straightening of edges]\label[construction]{con:straightening}
      Let $\gB$ be a covariant fibration of reflexive graphs over a reflexive graph $\gA$. We can define a ``straightening'' function that converts displayed edges in $\gB$ to vertical edges using the pushforward operation of the fibration.

      First, we define a more general function incorporating a witness of contraction.

      \iblock{
        \mrow{
          x,y:\vrt{\gA};
          p : x\Edge{\gA}y;
          u : \vrt{\gB}\prn{x};
          v : \vrt{\gB}\prn{y};
          q : u \Edge{\gB}[p] v;
        }
        \mrow{
          H : \Id{
            \Sum{w:\vrt{\gB}\prn{y}}
            u \Edge{\gB}[p] w
          }{
            \prn{p_*^\gB u,p_\dagger^\gB u}
          }{\prn{v,q}}
        }
        \iblock{
          \mrow{
            \vdash
            \StrGen{\gB}[p]\,q\,H : p_*^\gB u \Edge{\gB\prn{y}} v
          }
        }
        \commentrow{by indentification induction on $H$}
        \mrow{
          x,y:\vrt{\gA};
          p : x\Edge{\gA}y
          \vdash
          \StrGen{\gB}[p]\,\prn{p_\dagger^\gB u}\,\Refl : p_*^\gB u \Edge{\gB\prn{y}} p_*^\gB u
        }
        \commentrow{by reflexivity data}
        \mrow{
          x,y:\vrt{\gA};
          p : x\Edge{\gA}y
          \vdash
          \StrGen{\gB}[p]\,\prn{p_\dagger^\gB u}\,\Refl :\equiv
          \DRx{\gB}{y}\prn{p_*^\gB u}
        }
      }

      Then we define the straightening function by instantiation with one of the witnesses of contraction of $\Sum{w:\vrt{\gB}\prn{y}}u\Edge{\gB}[p]w$, which we leave nameless.

      \iblock{
        \mhang{
          x,y:\vrt{\gA};
          p : x\Edge{\gA}y;
          u : \vrt{\gB}\prn{x};
          v : \vrt{\gB}\prn{y};
          q : u \Edge{\gB}[p] v
          \vdash
        }{
          \mrow{
            \Str{\gB}[p]\,q : p_*^\gB u \Edge{\gB\prn{y}} v
          }
          \mrow{
            \Str{\gB}[p]\,q :\equiv
            \StrGen{\gB}[p]\,q\,\prn{\ldots}
          }
        }
      }
    \end{construction}

    \begin{construction}[Unstraightening of edges]\label[construction]{con:unstraightening}
      Let $\gB$ be a covariant fibration of reflexive graphs over a reflexive graph $\gA$. We can ``unstraighten'' vertical edges of the form $p_*^\gB u \Edge{\gB\prn{y}} v$ to displayed edges $u\Edge{\gB}[p] v$ as follows:

      \iblock{
        \mrow{
          x,y:\vrt{\gA}; p:x\Edge{\gA}y, u : \vrt{\gB}\prn{x}, v:\vrt{\gB}\prn{y}, q : p_*^\gB u \Edge{\gB\prn{y}} v,
        }
        \mrow{
          H : \Id{
            \Sum{w:\vrt{\gB}\prn{y}}
            p_*^\gB u \Edge{\gB\prn{y}}w
          }{
            \prn{p_*^\gB u, \DRx{\gB}{y}\prn{p_*^\gB u}}
          }{
            \prn{v,q}
          }
        }
        \iblock{
          \mrow{
            \vdash
            \UnstrGen{\gB}[p]\, q\, H :
            u \Edge{\gB}[p] v
          }
        }
        \commentrow{by indentification induction on $H$}
        \mhang{
          x,y:\vrt{\gA};
          p:x\Edge{\gA}y,
          u : \vrt{\gB}\prn{x}
        }{
          \mrow{
            \vdash
            \UnstrGen{\gB}[p]\,\prn{\DRx{\gB}{y}\prn{p_*^\gB u}}\,\Refl
            : u \Edge{\gB}[p] p_*^\gB u
          }
          \mrow{
            \vdash
            \UnstrGen{\gB}[p]\,\prn{\DRx{\gB}{y}\prn{p_*^\gB u}}\,\Refl
            :\equiv
            p_\dagger^\gB u
          }
        }
      }

      Then we define the unstraightening function by instantiation with an anonymous witness that can be obtained from the fact that $\Sum{w:\vrt{\gB}\prn{y}}p_*^\gB u \Edge{\gB\prn{y}}w$ is a proposition.

      \iblock{
        \mhang{
          x,y:\vrt{\gA}; p : x\Edge{\gA} y, u :\vrt{\gB}\prn{x}, v : \vrt{\gB}\prn{y};q : p_*^\gB u \Edge{\gB\prn{y}} v \vdash
        }{
          \mrow{\Unstr{\gB}[p]\, q : u\Edge{\gB}[p] v}
          \mrow{\Unstr{\gB}[p]\, q :\equiv \UnstrGen{\gB}[p]\,q\,\prn{\ldots}}
        }
      }
    \end{construction}

    \begin{lemma}[Straightening is an equivalence]\label[lemma]{lem:cov-str-equiv}
      Let $\gB$ be a covariant fibration of reflexive graphs over a reflexive graph $\gA$. Each straightening function $\Str{\gB}[p] : u \Edge{\gB}[p]v \to p_*^\gB u\Edge{\gB\prn{y}} v$ is an equivalence, with the inverse tracked by $\Unstr{\gB}[p] : p_*^\gB u\Edge{\gB\prn{y}} v \to u \Edge{\gB}[p]v$.
    \end{lemma}

    \begin{proof}
      We first check that straightening is a retraction of unstraightening. It suffices to prove the following slightly generalised lemma.

      \iblock{
        \mrow{
          x,y:\vrt{\gA};
          p:x\Edge{\gA}y,
          u:\vrt{\gB}\prn{x},
          v:\vrt{\gB}\prn{y},
          q : p_*^\gB u \Edge{\gB\prn{y}}v,
        }
        \mrow{
          H :
          \Id{
            \Sum{w:\vrt{\gB}\prn{y}}
            p_*^\gB u \Edge{\gB\prn{y}} w
          }{
            \prn{p_*^\gB u, \DRx{\gB}{y}\prn{p_*^\gB u}}
          }{\prn{v,q}},
        }
        \mrow{
          I :
          \Id{
            \Sum{w:\vrt{\gB}\prn{y}}
            u \Edge{\gB}[p] w
          }{
            \prn{p_*^\gB u,p_\dagger^\gB u}
          }{
            \prn{v, \UnstrGen{\gB}[p]\,q\,H}
          }
        }
        \iblock{
          \mrow{
            \vdash
            \Id{p_*^\gB u \Edge{\gB\prn{y}} v}{
              \StrGen{\gB}[p]\,\prn{\UnstrGen{\gB}[p]\,q\,H}\,I
            }{q}
          }
        }
        \commentrow{by identification induction on $H$}
        \mrow{
          x,y:\vrt{\gA};
          p:x\Edge{\gA}y,
          u:\vrt{\gB}\prn{x},
          v:\vrt{\gB}\prn{y},
        }
        \mrow{
          I :
          \Id{
            \Sum{w:\vrt{\gB}\prn{y}}
            u \Edge{\gB}[p] w
          }{
            \prn{p_*^\gB u,p_\dagger^\gB u}
          }{
            \prn{p_*^\gB u, \UnstrGen{\gB}[p]\,\prn{\DRx{\gB}{y}\prn{p_*^\gB u}}\,\Refl}
          }
        }
        \iblock{
          \mrow{
            \vdash
            \Id{p_*^\gB u \Edge{\gB\prn{y}} v}{
              \StrGen{\gB}[p]\,\prn{\UnstrGen{\gB}[p]\,\prn{\DRx{\gB}{y}\prn{p_*^\gB u}}\,\Refl}\,I
            }{\DRx{\gB}{y}\prn{p_*^\gB u}}
          }
        }
        \commentrow{by unfolding the definition of $\UnstrGen{\gB}[p]$}
        \mrow{
          x,y:\vrt{\gA};
          p:x\Edge{\gA}y,
          u:\vrt{\gB}\prn{x},
          v:\vrt{\gB}\prn{y}
        }
        \mrow{
          I :
          \Id{
            \Sum{w:\vrt{\gB}\prn{y}}
            u \Edge{\gB}[p] w
          }{
            \prn{p_*^\gB u,p_\dagger^\gB u}
          }{
            \prn{p_*^\gB u, p_\dagger^\gB u}
          }
        }
        \iblock{
          \mrow{
            \vdash
            \Id{p_*^\gB u \Edge{\gB\prn{y}} v}{
              \StrGen{\gB}[p]\,\prn{p_\dagger^\gB u}\,I
            }{\DRx{\gB}{y}\prn{p_*^\gB u}}
          }
        }
        \commentrow{%
          because $\gB$ is a covariant fibration, $\Sum{w:\vrt{\gB}\prn{y}}u\Edge{\gB}[p]w$ is a set
        }
        \mrow{
          x,y:\vrt{\gA};
          p:x\Edge{\gA}y,
          u:\vrt{\gB}\prn{x},
          v:\vrt{\gB}\prn{y}
        }
        \iblock{
          \mrow{
            \vdash
            \Id{p_*^\gB u \Edge{\gB\prn{y}} v}{
              \StrGen{\gB}[p]\,\prn{p_\dagger^\gB u}\,\Refl
            }{\DRx{\gB}{y}\prn{p_*^\gB u}}
          }
          \commentrow{by unfolding the definition of $\StrGen{\gB}[p]$}
          \mrow{
            \vdash
            \Id{p_*^\gB u \Edge{\gB\prn{y}} v}{
              \DRx{\gB}{y}\prn{p_*^\gB u}
            }{\DRx{\gB}{y}\prn{p_*^\gB u}}
          }
          \commentrow{by reflexivity}
        }
      }

      Conversely, we check that straightening is a section of unstraightening by means of the following generalised lemma.

      \iblock{
        \mrow{
          x,y:\vrt{\gA}; p:x\Edge{\gA}y, u:\vrt{\gB}\prn{x},v:\vrt{\gB}\prn{y},q:u\Edge{\gB}[p]v,
        }
        \mrow{
          H :
          \Id{
            \Sum{w:\vrt{\gB}\prn{y}}
            u\Edge{\gB}[p]w
          }{
            \prn{p_*^\gB u, p_\dagger^\gB u}
          }{
            \prn{v,q}
          },
        }
        \mrow{
          I :
          \Id{
            \Sum{w:\vrt{\gB}\prn{y}}p_*^\gB u \Edge{\gB}[\Rx{\gA}{y}]w
          }{
            \prn{p_*^\gB u, \DRx{\gB}{y}\prn{p_*^\gB u}}
          }{
            \prn{v,\StrGen{\gB}[p]\,q\,H}
          }
        }
        \iblock{
          \mrow{
            \vdash
            \Id{u\Edge{\gB}[p]v}{
              \UnstrGen{\gB}[p]\,\prn{\StrGen{\gB}[p]\,q\,H}\,I
            }{q}
          }
        }
        \commentrow{by identification induction on $H$}
        \mrow{
          x,y:\vrt{\gA}; p:x\Edge{\gA}y, u:\vrt{\gB}\prn{x}
        }
        \mrow{
          I :
          \Id{
            \Sum{w:\vrt{\gB}\prn{y}}p_*^\gB u \Edge{\gB}[\Rx{\gA}{y}]w
          }{
            \prn{p_*^\gB u, \DRx{\gB}{y}\prn{p_*^\gB u}}
          }{
            \prn{p_*^\gB u,\StrGen{\gB}[p]\,\prn{p_\dagger^\gB u}\,\Refl}
          }
        }
        \iblock{
          \mrow{
            \vdash
            \Id{u\Edge{\gB}[p]v}{
              \UnstrGen{\gB}[p]\,\prn{\StrGen{\gB}[p]\,\prn{p_\dagger^\gB u}\,\Refl}\,I
            }{p_\dagger^\gB u}
          }
        }
        \commentrow{%
          by unfolding the definition of $\StrGen{\gB}[p]$
        }
        \mrow{
          x,y:\vrt{\gA}; p:x\Edge{\gA}y, u:\vrt{\gB}\prn{x}
        }
        \mrow{
          I :
          \Id{
            \Sum{w:\vrt{\gB}\prn{y}}p_*^\gB u \Edge{\gB}[\Rx{\gA}{y}]w
          }{
            \prn{p_*^\gB u, \DRx{\gB}{y}\prn{p_*^\gB u}}
          }{
            \prn{p_*^\gB u,\DRx{\gB}{y}\prn{p_*^\gB u}}
          }
        }
        \iblock{
          \mrow{
            \vdash
            \Id{u\Edge{\gB}[p]v}{
              \UnstrGen{\gB}[p]\,\prn{\DRx{\gB}{y}\prn{p_*^\gB u}}\,I
            }{p_\dagger^\gB u}
          }
        }
        \commentrow{%
          because $\gB$ is a covariant fibration, $\Sum{w:\vrt{\gB}\prn{y}}p_*^\gB u \Edge{\gB}[\Rx{\gA}{y}]w$ is a set
        }
        \mrow{
          x,y:\vrt{\gA}; p:x\Edge{\gA}y, u:\vrt{\gB}\prn{x}
        }
        \iblock{
          \mrow{
            \vdash
            \Id{u\Edge{\gB}[p]v}{
              \UnstrGen{\gB}[p]\,\prn{\DRx{\gB}{y}\prn{p_*^\gB u}}\,\Refl
            }{p_\dagger^\gB u}
          }
          \commentrow{by unfolding the definition of $\UnstrGen{\gB}[p]$}
          \mrow{
            \vdash
            \Id{u\Edge{\gB}[p]v}{
              p_\dagger^\gB u
            }{p_\dagger^\gB u}
          }
          \commentrow{by reflexivity}
          \qedhere
        }
      }
    \end{proof}

  \end{xsect}

  \begin{xsect}{The underlying lens of a covariant fibration}
    We are now equipped to construct the underlying lens of a given covariant fibration.

    \begin{construction}[The underlying lens of a covariant fibration]\label[construction]{con:underlying-cov-lens}
      A covariant fibration of reflexive graphs $\gB$ over a reflexive graph $\gA$ induces a (canonical, as we shall see) oplax covariant lens structure on the diagonal family $\Diag{\gB}$ via straightening.

      \iblock{
        \begin{multicols}{2}
          \mhang{
            x,y:\vrt{\gA};
            p:x\Edge{\gA}y,
            u:\vrt{\gB}\prn{x}
            \vdash
          }{
            \mrow{
              \Push{\Diag{\gB}}{p}{u} : \vrt{\gB}\prn{y}
            }
            \mrow{
              \Push{\Diag{\gB}}{p}{u} :\equiv p_*^\gB u
            }
          }

          \columnbreak

          \mhang{
            x:\vrt{\gA},u:\vrt{\gB}\prn{x} \vdash
          }{
            \mrow{
              \PushRx{\Diag{\gB}}[x]u : \prn{\Rx{\gA}{x}}_*^\gB u \Edge{\gB\prn{x}} u
            }
            \mrow{
              \PushRx{\Diag{\gB}}[x]u :\equiv
              \Str{\gB}[\Rx{\gA}{x}]\prn{\DRx{\gB}{x}{u}}
            }
          }
        \end{multicols}
      }
    \end{construction}

    \cref{lem:display-of-underlying-lens-roundtrip} below shows that if we take the display of the underlying lens of a covariant fibration, we obtain the original displayed reflexive graph.

    \begin{therm}[Roundtrip for fibrations of reflexive graphs]\label[theorem]{lem:display-of-underlying-lens-roundtrip}
      Let $\prn{U,E}$ be a univalent universe, and let $\gB$ be a $U$-small covariant fibration of reflexive graphs over a reflexive graph $\gA$.
      Then the display $\CovDisp{\gA}{\Diag{\gB}}$ of the underlying lens of $\gB$ may be identified with $\gB$ by means of an equivalence of displayed reflexive graphs $\gB\Edge{\DRxGphOver{\gA}} \CovDisp{\gA}{\Diag{\gB}}$.
    \end{therm}

    \begin{proof}
      We shall use the characterisation of identifications of displayed reflexive graphs that we established in \cref{sec:drxgph-sip}. In particular, we shall construct an equivalence of displayed reflexive graphs $\phi : \gB \Edge{\DRxGphOver{\gA}} \CovDisp{\gA}{\Diag{\gB}}$. We can use the identity equivalence on vertices, and the straightening--unstraightening equivalence on edges:

      \iblock{
        \mrow{
          \vrt{\phi} : \Prod{x:\vrt{\gA}} \mathsf{Equiv}\prn{\vrt{\gB}\prn{x},\vrt{\Diag{\gB}\prn{x}}}
        }
        \mrow{
          \vrt{\phi} x :\equiv \mathsf{idnEquiv}\Sub{\vrt{\gB}\prn{x}}
        }
        \row
        \mrow{
          \phi^\approx :
          \Prod{x,y:\vrt{\gA}}\Prod{p:x\Edge{\gA}y}
          \Prod{u:\vrt{\gB}\prn{x}}
          \Prod{v:\vrt{\gB}\prn{y}}
          \mathsf{Equiv}\prn{
            u\Edge{\gB}[p]v,
            u \Edge{\CovDisp{\gA}{\Diag{\gB}}}[p] v
          }
        }
        \commentrow{by introduction}
        \mhang{
          x,y:\vrt{\gA}; p:x\Edge{\gA}y,u:\vrt{\gB}\prn{x},v:\vrt{\gB}\prn{y}
        }{
          \mrow{
            \vdash \phi^\approx x y p u v :
            \mathsf{Equiv}\prn{
              u\Edge{\gB}[p]v,
              u \Edge{\CovDisp{\gA}{\Diag{\gB}}}[p] v
            }
          }
          \commentrow{by computation}
          \mrow{
            \vdash \phi^\approx x y p u v :
            \mathsf{Equiv}\prn{
              u\Edge{\gB}[p]v,
              p_*^\gB u\Edge{\gB\prn{y}} v
            }
          }
          \commentrow{by straightening}
          \mrow{
            \vdash \phi^\approx x y p u v :\equiv
            \Str{\gB}[p]
          }
        }
      }

      The displayed reflexivity data are preserved up to definitional equality, as we have:

      \iblock{
        \mrow{
          \phi^{\mathsf{rx}} : \Prod{x:\vrt{\gA}} \Prod{u:\vrt{\gB}\prn{x}}
          \Id{
            u \Edge{\CovDisp{\gA}{\Diag{\gB}}}[\Rx{\gA}{x}] u
          }{
            \phi^\approx\,x\,x\,\prn{\DRx{\gA}{x}}\,u\,u\,\prn{\Rx{\gB}{x}{u}}
          }{
            \DRx{\CovDisp{\gA}{\Diag{\gB}}}{x}{u}
          }
        }
        \commentrow{by introduction}
        \mhang{
          x:\vrt{\gA},u:\vrt{\gB}\prn{x}
        }{
          \mrow{
            \vdash
            \phi^{\mathsf{rx}}xu :
            \Id{
              u \Edge{\CovDisp{\gA}{\Diag{\gB}}}[\Rx{\gA}{x}] u
            }{
              \phi^\approx\,x\,x\,\prn{\DRx{\gA}{x}}\,u\,u\,\prn{\Rx{\gB}{x}{u}}
            }{
              \DRx{\CovDisp{\gA}{\Diag{\gB}}}{x}{u}
            }
          }
          \commentrow{by computation}
          \mrow{
            \vdash
            \phi^{\mathsf{rx}}xu :
            \Id{
              \prn{\Rx{\gA}{x}}_*^\gB u \Edge{\gB\prn{x}} u
            }{
              \Str{\gB}[p]\prn{\DRx{\gB}{x}{u}}
            }{
              \Str{\gB}[p]\prn{\DRx{\gB}{x}{u}}
            }
          }
          \commentrow{by reflexivity}
          \mrow{
            \vdash
            \phi^{\mathsf{rx}}xu :\equiv \Refl
          }
          \qedhere
        }
      }
    \end{proof}
  \end{xsect}
\end{xsect}

\begin{xsect}{Characterisation of fibred reflexive graphs}

  Let $\gB$ be an oplax covariant lens of $U$-small path objects over a reflexive graph $\gA$. By \cref{lem:cov-lens-to-fibration}, the display $\CovDisp{\gA}{\gB}$ is a covariant fibration. By \cref{con:underlying-cov-lens}, we obtain an oplax covariant lens $\Diag{\prn{\CovDisp{\gA}{\gB}}}$ that we may compare with $\gB$. We first recall from \cref{cmp:component-of-oplax-cov-display} the underlying reflexive graph of each component $\Diag{\prn{\CovDisp{\gA}{\gB}}}\prn{x}$:
  \begin{align*}
    \vrt{\Diag{\prn{\CovDisp{\gA}{\gB}}}\prn{x}}
     & \equiv
    \vrt{\gB\prn{x}}
    \\
    u\Edge{\Diag{\prn{\CovDisp{\gA}{\gB}}}\prn{x}} v
     & \equiv
    \Push{\gB}{\Rx{\gA}{x}}{u} \Edge{\gB\prn{x}} v
    \\
    \Rx{\Diag{\prn{\CovDisp{\gA}{\gB}}}\prn{x}}{u}
     & \equiv
    \PushRx{\gB}[x]{u}
  \end{align*}

  \begin{computation}\label[computation]{cmp:lens-B'}
    We compute the lens structure of $\Diag{\prn{\CovDisp{\gA}{\gB}}}$ as follows, recalling \cref{con:underlying-cov-lens,lem:cov-lens-to-fibration}.
    \begin{align*}
      \Push{\Diag{\prn{\CovDisp{\gA}{\gB}}}}{p}{u}
       & \equiv
      p_*\Sup{\CovDisp{\gA}{\gB}}u
      \equiv
      \Push{\gB}{p}{u}
      \\
      \PushRx{\Diag{\prn{\CovDisp{\gA}{\gB}}}}[x]{u}
       & \equiv
      \Str{\CovDisp{\gA}{\gB}}[\Rx{\gA}{x}]{\prn{\DRx{\CovDisp{\gA}{\gB}}{x}{u}}}
      \equiv
      \Str{\CovDisp{\gA}{\gB}}[\Rx{\gA}{x}]{\prn{\PushRx{\gB}[x]{u}}}
    \end{align*}
  \end{computation}

  \begin{construction}
    We have an equivalence $\eta_\gB : \gB\Edge{\vrt{\gA}\pitchfork\RxGph}\Diag{\prn{\CovDisp{\gA}{\gB}}}$ of underlying families of reflexive graphs over $\gA$, which we define below making use of \cref{cor:fibrations-are-univalent,con:path-alg-toolkit,lem:pre-ct-is-equiv}.
    \begin{align*}
      \eta_{\gB}
       & :
      \gB \Edge{\vrt{\gA}\pitchfork\RxGph} \Diag{\prn{\CovDisp{\gA}{\gB}}}
      \\
      \vrt{\eta_\gB x}
       & :\equiv
      \Rx{U/E}{\vrt{\gB\prn{x}}}
      \\
      \prn{\eta_\gB x}^\approx
       & :\equiv
      \Lam{u,v}
      \prn{\PushRx{\gB}[x]{u}\ct\Sub{\gB\prn{x}}-}
      \\
      \prn{\eta_\gB x}^{\mathsf{rx}}
       & :\equiv
      \Lam{u}
      \mathsf{runit}\Sub{\PushRx{\gB}[x]{u}}
    \end{align*}
  \end{construction}

  \begin{computation}\label[computation]{cmp:extrusion-of-B-and-B'}
    Write $\LJ{\CovLensStr{\gA}}{\eta_\gB}\gB : \CovLensStr{\gA}{\eta_\gB}$ for the extension of the lens structure from $\gB$ onto the equivalence $\eta_\gB : \gB\Edge{\vrt{\gA}\pitchfork\RxGph}\Diag{\prn{\CovDisp{\gA}{\gB}}}$ of families of reflexive graphs. Likewise, write $\RJ{\CovLensStr{\gA}}{\eta_\gB}\Diag{\prn{\CovDisp{\gA}{\gB}}} : \CovLensStr{\gB}{\eta_\gB}$ for the corresponding extension of the lens structure of $\Diag{\prn{\CovDisp{\gA}{\gB}}}$ onto $\eta_\gB$. We have the definitional computations of the lens structures:
    \begin{align*}
      \Push{\LJ{\CovLensStr{\gA}}{\eta_\gB}\gB}{p}{u}
       & \equiv
      \Push{\gB}{p}{u}
      \\
      \PushRx{\LJ{\CovLensStr{\gA}}{\eta_\gB}\gB}[x]{u}
       & \equiv
      \PushRx{\gB}[x]{
        \prn{\Push{\gB}{\Rx{\gA}{x}}{u}}
      } \ct\Sub{\gB\prn{x}} \PushRx{\gB}[x]{u}
      \\[8pt]
      \Push{\RJ{\CovLensStr{\gA}}{\eta_\gB}\Diag{\prn{\CovDisp{\gA}{\gB}}}}{p}{u}
       & \equiv
      \Push{\gB}{p}{u}
      \\
      \PushRx{\RJ{\CovLensStr{\gA}}{\eta_\gB}\Diag{\prn{\CovDisp{\gA}{\gB}}}}[x]{u}
       & \equiv
      \Str{\CovDisp{\gA}{\gB}}[\Rx{\gA}{x}]{\prn{\PushRx{\gB}[x]{u}}}
    \end{align*}
  \end{computation}

  \begin{proof}
    We outline the explicit computations below, unfolding \cref{con:lens-of-lenses,con:underlying-cov-lens,cmp:lens-B'}.

    \iblock{
      \setlength\columnsep{-2cm}
      \begin{multicols}{2}
        \mhang{
          \smash{\Push{\LJ{\CovLensStr{\gA}}{\eta_\gB}\gB}{p}{u}}\vphantom{\sum_\sum}
        }{
          \commentrow{by \cref{con:lens-of-lenses}}
          \mrow{
            {}\equiv
            \vrt{\eta_\gB y}\prn{\Push{\gB}{p}{u}}
          }
          \commentrow{by \cref{con:underlying-cov-lens}}
          \mrow{
            {}\equiv
            \Push{\gB}{p}{u}
          }
        }
        \columnbreak
        \mhang{
          \smash{\PushRx{\LJ{\CovLensStr{\gA}}{\eta_\gB}\gB}[x]{u}}
        }{
          \commentrow{by \cref{con:lens-of-lenses}}
          \mrow{
            {}\equiv
            \vrt{\eta_\gB x}^\approx\,\prn{\Push{\gB}{\Rx{\gA}{x}}{u}}\,u\,\prn{\PushRx{\gB}[x]{u}}
          }
          \commentrow{by \cref{con:underlying-cov-lens}}
          \mrow{
            {}\equiv
            \PushRx{\gB}[x]{
              \prn{\Push{\gB}{\Rx{\gA}{x}}{u}}
            } \ct\Sub{\gB\prn{x}} \PushRx{\gB}[x]{u}
          }
        }
      \end{multicols}
      \begin{multicols}{2}
        \mhang{
          \smash{\Push{\RJ{\CovLensStr{\gA}}{\eta_\gB}\Diag{\prn{\CovDisp{\gA}{\gB}}}}{p}{u}}
          \vphantom{\sum_\sum}
        }{
          \commentrow{by \cref{con:lens-of-lenses}}
          \mrow{
            {}\equiv
            \Push{\Diag{\prn{\CovDisp{\gA}{\gB}}}}{p}{\prn{\vrt{\eta_\gB x}u}}
          }
          \commentrow{by \cref{con:underlying-cov-lens}}
          \mrow{
            {}\equiv
            \Push{\Diag{\prn{\CovDisp{\gA}{\gB}}}}{p}{u}
          }
          \commentrow{by \cref{cmp:lens-B'}}
          \mrow{
            {}\equiv
            \Push{\gB}{p}{u}
          }
        }
        \columnbreak
        \mhang{
          \smash{\PushRx{\RJ{\CovLensStr{\gA}}{\eta_\gB}\Diag{\prn{\CovDisp{\gA}{\gB}}}}[x]{u}}
        }{
          \commentrow{by \cref{con:lens-of-lenses}}
          \mrow{
            {}\equiv
            \PushRx{\Diag{\prn{\CovDisp{\gA}{\gB}}}}[x]{\prn{\vrt{\eta_\gB x}u}}
          }
          \commentrow{by \cref{con:underlying-cov-lens}}
          \mrow{
            {}\equiv
            \PushRx{\Diag{\prn{\CovDisp{\gA}{\gB}}}}[x]{u}
          }
          \commentrow{by \cref{cmp:lens-B'}}
          \mrow{
            {}\equiv
            \Str{\CovDisp{\gA}{\gB}}[\Rx{\gA}{x}]{\prn{\PushRx{\gB}[x]{u}}}
          }
          \qedhere
        }
      \end{multicols}
    }
  \end{proof}

  \begin{therm}[Roundtrip for oplax covariant lenses of path objects]\label[theorem]{thm:underlying-lens-of-display-roundtrip}
    The equivalence $\eta_\gB : \gB \Edge{\vrt{\gA}\pitchfork\RxGph} \Diag{\prn{\CovDisp{\gA}{\gB}}}$ of families of reflexive graphs extends to an equivalence $\gB \Edge{\CovLensOver{\gA}} \Diag{\prn{\CovDisp{\gA}{\gB}}}$ of oplax covariant lenses over $\gA$.
  \end{therm}

  \begin{proof}
    We must construct a displayed equivalence $\gB\Edge{\UnbDisp{\vrt{\gA}\pitchfork \RxGph}{\CovLensStr{\gA}}}[\eta_\gB]
      \Diag{\prn{\CovDisp{\gA}{\gB}}}$ of lens structures over $\eta_\gB$. This consists of an (ordinary) equivalence of lenses as below:
    \[
      \LJ{\CovLensStr{\gA}}{\eta_\gB}\gB
      \Edge{\CovLensStr{\gA}{\eta_\gB}}
      \RJ{\CovLensStr{\gA}}{\eta_\gB}\Diag{\prn{\CovDisp{\gA}{\gB}}}
    \]

    Unfolding the definition of $\CovLensOver{\gA}{\eta_\gB}$ from \cref{con:lens-of-lenses} and using \cref{cmp:extrusion-of-B-and-B'}, we see that the pushforward datum for each lens depicted above is definitionally equal to that of $\gB$. Therefore, it suffices to construct an equivalence between oplax unitors as below:

    \iblock{
      \mhang{
        x:\vrt{\gA}, u:\vrt{\gB\prn{x}}
      }{
        \mrow{
          \vdash
          \Id{
            \ldots
          }{
            \PushRx{\LJ{\CovLensStr{\gA}}{\eta_\gB}\gB}[x]{u}
          }{
            \PushRx{\RJ{\CovLensStr{\gA}}{\eta_\gB}\Diag{\prn{\CovDisp{\gA}{\gB}}}}[x]{u}
          }
        }
        \commentrow{by \cref{cmp:extrusion-of-B-and-B'}}
        \mrow{
          \vdash
          \Id{\ldots}{
            \PushRx{\gB}[x]{
              \prn{\Push{\gB}{\Rx{\gA}{x}}{u}}
            } \ct\Sub{\gB\prn{x}} \PushRx{\gB}[x]{u}
          }{
            \Str{\CovDisp{\gA}{\gB}}[\Rx{\gA}{x}]{\prn{\PushRx{\gB}[x]{u}}}
          }
        }
        \commentrow{by \cref{con:straightening}}
        \mrow{
          \vdash
          \Id{\ldots}{
            \PushRx{\gB}[x]{
              \prn{\Push{\gB}{\Rx{\gA}{x}}{u}}
            } \ct\Sub{\gB\prn{x}} \PushRx{\gB}[x]{u}
          }{
            \StrGen{\CovDisp{\gA}{\gB}}[\Rx{\gA}{x}]{\prn{\PushRx{\gB}[x]{u}}}\,\prn{\ldots}
          }
        }
      }
    }

    Up to definitional equality, the goal above is a substitution instance of the following more general goal which we establish below.

    \iblock{
      \mrow{
        x,y:\vrt{\gA}; p:x\Edge{\gA}y,
        u:\vrt{\gB\prn{x}},
        v:\vrt{\gB\prn{y}},
        q:\Push{\gB}{p}{u} \Edge{\gB\prn{y}} v,
      }
      \mrow{
        H :
        \Id{
          \Sum{\vrt{\gB\prn{y}}}
        }{
          \prn{
            p_*\Sup{\CovDisp{\gA}{\gB}}u,
            p_\dagger\Sup{\CovDisp{\gA}{\gB}}u
          }
        }{
          \prn{v,q}
        }
      }
      \iblock{
        \mrow{
          \vdash
          \Id{
            \Push{\gB}{\Rx{\gA}{x}}{
              \prn{
                \Push{\gB}{p}{u}
              }
            }
            \Edge{\gB\prn{y}}
            v
          }{
            \PushRx{\gB}[y]{
              \prn{\Push{\gB}{p}{u}}
            }\ct_\gB q
          }{
            \StrGen{\CovDisp{\gA}{\gB}}[p]{q}\,H
          }
        }
      }
      \row
      \commentrow{by \cref{lem:cov-lens-to-fibration,cmp:lens-B'}}
      \row

      \mrow{
        x,y:\vrt{\gA}; p:x\Edge{\gA}y,
        u:\vrt{\gB\prn{x}},
        v:\vrt{\gB\prn{y}},
        q:\Push{\gB}{p}{u} \Edge{\gB\prn{y}} v,
      }
      \mrow{
        H :
        \Id{
          \Sum{\vrt{\gB\prn{y}}}
        }{
          \prn{
            \Push{\gB}{p}{u},
            \Rx{\gB\prn{y}}\prn{\Push{\gB}{p}{u}}
          }
        }{
          \prn{v,q}
        }
      }
      \iblock{
        \mrow{
          \vdash
          \Id{
            \Push{\gB}{\Rx{\gA}{x}}{
              \prn{
                \Push{\gB}{p}{u}
              }
            }
            \Edge{\gB\prn{y}}
            v
          }{
            \PushRx{\gB}[y]{
              \prn{\Push{\gB}{p}{u}}
            }\ct_\gB q
          }{
            \StrGen{\CovDisp{\gA}{\gB}}[p]{q}\,H
          }
        }
      }

      \row
      \commentrow{by identification induction on $H$}
      \row

      \mhang{
        x,y:\vrt{\gA}; p:x\Edge{\gA}y,
        u:\vrt{\gB\prn{x}}
      }{
        \mrow{
          \vdash
          \Id{
            \Push{\gB}{\Rx{\gA}{x}}{
              \prn{
                \Push{\gB}{p}{u}
              }
            }
            \Edge{\gB\prn{y}}
            \Push{\gB}{p}{u}
          }{
            \PushRx{\gB}[y]{
              \prn{\Push{\gB}{p}{u}}
            }\ct_\gB \Rx{\gB\prn{y}}\prn{\Push{\gB}{p}{u}}
          }{
            \StrGen{\CovDisp{\gA}{\gB}}[p]{\prn{\Rx{\gB\prn{y}}\prn{\Push{\gB}{p}{u}}}}\,\Refl
          }
        }
      }

      \row
      \commentrow{by \cref{lem:cov-lens-to-fibration}}
      \row

      \mhang{
        x,y:\vrt{\gA}; p:x\Edge{\gA}y,
        u:\vrt{\gB\prn{x}}
      }{
        \mrow{
          \vdash
          \Id{
            \Push{\gB}{\Rx{\gA}{x}}{
              \prn{
                \Push{\gB}{p}{u}
              }
            }
            \Edge{\gB\prn{y}}
            \Push{\gB}{p}{u}
          }{
            \PushRx{\gB}[y]{
              \prn{\Push{\gB}{p}{u}}
            }\ct_\gB \Rx{\gB\prn{y}}\prn{\Push{\gB}{p}{u}}
          }{
            \StrGen{\CovDisp{\gA}{\gB}}[p]{\prn{\DelimMin{1}p_\dagger\Sup{\CovDisp{\gA}{\gB}}{u}}}\,\Refl
          }
        }
      }

      \row
      \commentrow{by \cref{con:straightening,lem:cov-lens-to-fibration}}
      \row

      \mhang{
        x,y:\vrt{\gA}; p:x\Edge{\gA}y,
        u:\vrt{\gB\prn{x}}
      }{
        \mrow{
          \vdash
          \Id{
            \Push{\gB}{\Rx{\gA}{x}}{
              \prn{
                \Push{\gB}{p}{u}
              }
            }
            \Edge{\gB\prn{y}}
            \Push{\gB}{p}{u}
          }{
            \PushRx{\gB}[y]{
              \prn{\Push{\gB}{p}{u}}
            }\ct_\gB \Rx{\gB\prn{y}}\prn{\Push{\gB}{p}{u}}
          }{
            \PushRx{\gB}[y]\prn{\Push{\gB}{p}{u}}
          }
        }
      }

      \row
      \commentrow{by $\mathsf{runit}\Sub{\PushRx{\gB}[y]\prn{\Push{\gB}{p}{u}}}$ from \cref{con:path-alg-toolkit}}
      \qedhere
    }
  \end{proof}

  \begin{corollary}[Characterisation of fibred reflexive graphs]\label[corollary]{cor:characterisation-of-fibs}
    Let $\gA$ be a reflexive graph, and let $(U,E)$ be a univalent universe. Assuming function extensionality, the following two types are equivalent:
    \begin{enumerate}
      \item The type of oplax covariant (\resp lax contravariant) lenses of $U$-small path objects over $\gA$.
      \item The type of covariant (\resp contravariant) fibrations of $U$-small reflexive graphs over $\gA$.
    \end{enumerate}
  \end{corollary}
  \begin{proof}
    An oplax covariant lens of path objects $\gB$ over $\gA$ is sent to covariant fibration $\CovDisp{\gA}{\gB}$ as described in \cref{lem:cov-lens-to-fibration}. Conversely, a covariant fibration $\gB$ over $\gA$ is sent to its underlying lens of path objects $\Diag{\gB}$ as specified in \cref{con:underlying-cov-lens}. That these transformations are mutually inverse follows from \cref{lem:display-of-underlying-lens-roundtrip,thm:underlying-lens-of-display-roundtrip}, as both $\CovLensOver{\gA}$ and $\DRxGphOver{\gA}$ are univalent in the presence of function extensionality.

    The equivalence between lax contravariant lenses and contravariant fibrations is established from the above by duality (\cref{sec:rx-gph-duality,sec:fibration-duality}).
  \end{proof}

\end{xsect}
 
\end{xsect}

\bibliographystyle{plainnat}
\bibliography{refs-bibtex}

\end{document}